\documentclass{article}

\usepackage{amsmath, amsthm, amssymb, dsfont, mathtools}
\usepackage{graphicx}
\usepackage{verbatim}
\usepackage{natbib}
\usepackage{caption}
\usepackage{subcaption}
\usepackage{enumerate}
\usepackage[inline]{enumitem}
\usepackage{relsize}
\usepackage{hyperref}
\usepackage[margin=1.5in]{geometry}
\hypersetup{colorlinks,citecolor=blue,urlcolor=blue,linkcolor=blue}
\usepackage{diagbox}
\usepackage{nikos_tex}
\usepackage{float}
\usepackage{booktabs}
\usepackage{subcaption}
\usepackage{multirow}
\usepackage{footmisc}
\usepackage{adjustbox}
\usepackage{makecell}

\usepackage[utf8]{inputenc}
\usepackage{booktabs}

\usepackage{tikz}
\usetikzlibrary{positioning, calc}

\newcommand{\dd}{\mathrm{d}}

\newcommand{\beginsupptables}{
  \setcounter{table}{0}
  \renewcommand{\thetable}{S\arabic{table}}}

\newcommand{\beginsuppfigs}{
  \setcounter{figure}{0}
  \renewcommand{\thefigure}{S\arabic{figure}}}
\graphicspath{{./figures/}}

\theoremstyle{definition}
\newtheorem{prop}{Proposition}

\newtheorem{lemm}[prop]{Lemma}
\newtheorem{theo}[prop]{Theorem}
\newtheorem{rema}[prop]{Remark}

\newtheorem{assum}[prop]{Assumption}
\newtheorem{defi}[prop]{Definition}

\theoremstyle{remark}

\newtheoremstyle{modelstyle}
{3pt}{3pt}
{}{}
{\bfseries}{}
{0.5em}
{\thmname{#1}\ \thmnumber{#2}\thmnote{: #3}}
\theoremstyle{modelstyle}

 \usepackage{epigraph}

 \usepackage{etoolbox}
 
 \makeatletter
 \patchcmd{\epigraph}{\@epitext{#1}}{\@epitext{#1}}{}{}
 \makeatother

\date{Draft manuscript: March, 2026}

\title{Empirical Bayes learning from selectively reported confidence intervals} 
\author{
\begin{tabular}{ll}
\vspace{0.2cm}\\
Hunter Chen & \hspace{1cm} Junming Guan\\
\href{mailto:hunterchen@uchicago.edu}{\textcolor{black}{\texttt{hunterchen@uchicago.edu}}} & \hspace{1cm} \href{mailto:junmingguan@uchicago.edu}{\textcolor{black}{\texttt{junmingguan@uchicago.edu}}}\\
\vspace{0.05cm}\\
Erik van Zwet & \hspace{1cm} Nikolaos Ignatiadis\\
\href{mailto:E.W.van_Zwet@lumc.nl}{\textcolor{black}{\texttt{E.W.van\_Zwet@lumc.nl}}} & \hspace{1cm} \href{mailto:ignat@uchicago.edu}{\textcolor{black}{\texttt{ignat@uchicago.edu}}}\\
\vspace{0.2cm}\\
\end{tabular}
}

\begin{document}

\maketitle

\begin{abstract}
We develop a statistical framework for empirical Bayes learning from selectively reported confidence intervals, and apply it to provide context for interpreting results published in MEDLINE abstracts. We use a collection of 326,060 z-scores from MEDLINE abstracts (2000–2018) as the input for an empirical Bayes analysis, with publication bias as a key methodological challenge. We address publication bias through a selective tilting approach that extends empirical Bayes confidence intervals to truncated sampling. Our framework provides coverage guarantees for functionals including posterior estimands describing idealized replications and the
symmetrized posterior mean, which we justify decision-theoretically as optimal among sign-equivariant (odd) estimators.
\\ 

\noindent \textbf{Keywords:} publication bias, truncation models, selective tilting, $F$-Localization

\end{abstract}

\newpage

\section{Introduction}
\begin{figure}
  \centering
  \includegraphics[width=0.7\textwidth]{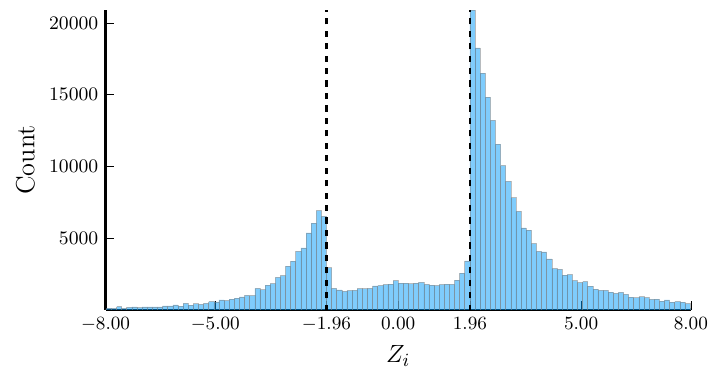}
  \caption{Histogram of $326,060$ z-scores from abstracts (one z-score per abstract) appearing in MEDLINE (2000--2018). See Supplement~\ref{sec:medline} for preprocessing details.}
  \label{fig:hist_med}
\end{figure}

Figure~\ref{fig:hist_med} shows a histogram of z-scores from $326,060$ abstracts indexed in MEDLINE, the bibliographic database of the National Library of Medicine (NLM), between 2000--2018. This now well-known histogram derives from \citet{vanzwet2021significance}, who convert to z-scores the confidence intervals for ratio estimands (hazard ratios, odds ratios, relative risks) originally scraped by \citet{georgescu2018algorithmic, BarnetteMEDLINE2019}; see \citet{vanzwet2025fifth} for further discussion.
Staring at the histogram, we may be tempted to lament the state of academic publishing wherein z-scores just above the cutoff for statistical significance ($|z| \geq 1.96$) are much more likely to appear in abstracts.
Instead, in this paper we ask: what can we learn from this histogram? More concretely, suppose we read the following in another abstract, indexed in MEDLINE in 2019, that we deem exchangeable with the $326,060$ abstracts above, in the sense that we treat its z-score as a draw from the same distribution:

\epigraph{``The hazard of MRSA [methicillin-resistant \emph{Staphylococcus aureus}] infection was significantly lower in the
decolonization group than in the education group (hazard ratio, 0.70; 95\% confidence interval
[CI], 0.52 to 0.96; P=0.03.''~~~~\citep{decolonization2019medline}}

\noindent How can we use the z-scores in Fig.~\ref{fig:hist_med} to provide context for interpreting the result of the abstract by~\citet{decolonization2019medline}?
Our basic supposition is that there is a population of true signal-to-noise ratios (SNRs), denoted by $\mu_i$, drawn from a distribution $G$ that represents studies that could potentially appear in MEDLINE abstracts,
\begin{equation}
\mu_i \sim G.
\label{eq:snr_prior}
\end{equation}
Our definition of the SNR $\mu_i$ is as the $i$-th study's true effect divided by its standard error. The z-score $Z_i$ is equal to $\mu_i$ observed with standard normal noise,
\begin{equation}
Z_i \mid \mu_i \sim \mathrm{N}(\mu_i, 1).
\label{eq:z_normal}
\end{equation}
Going back to the abstract of~\citet{decolonization2019medline}, the z-score that is used to form the confidence interval is equal to the estimated log hazard ratio divided by its standard error. We can read off that \smash{$\log(\mathrm{HR}_i) \stackrel{\cdot}{=} -0.35$} and  \smash{$\mathrm{SE}_i \stackrel{\cdot}{=}  0.16$}, so that \smash{$Z_i \stackrel{\cdot}{=}  -2.22$}.\footnote{\label{footnote:ci_to_z}
Specifically, our computation of z-scores from confidence intervals proceeds as follows. Let $[L_i,U_i]$ be the confidence interval, e.g., $L_i=0.52$, $U_i=0.96$ in our example.  We compute the standard error as \smash{$\text{SE}_i = (\log(U_i) -\log(L_i))/(2\cdot1.96)$}, and the z-score as \smash{$Z_i = (\log(U_i) +\log(L_i))/(2\cdot\text{SE}_i)$.
}} Herein, $\mu_i := \mathbb E[\log(\mathrm{HR}_i)]/\mathrm{SE}_i$.   The normality of $Z_i$ about $\mu_i$ in~\eqref{eq:z_normal} approximately follows from the central limit theorem, which implies that the sampling distribution of the estimator is approximately normal with variance given by the squared standard error. This approximation also underlies the construction of the reported confidence interval.

If we knew the distribution $G$ in~\eqref{eq:snr_prior}, then we could interpret the results of~\citet{decolonization2019medline} via the posterior distribution of $\mu_i \mid Z_i$; 
Bayesian inference
would proceed in the usual way even under selection on $Z_i$~\citep{dawid1994selection, senn2008note}. The difficulty is that $G$ is unknown and must be learned from a selectively observed collection of z-scores. This is a deconvolution problem further complicated by publication bias: studies with non-significant results may go unpublished, and even when published, confidence intervals for non-significant results may be omitted from the abstract.

In the analysis of MEDLINE reported in Section~\ref{sec:estimands_results}, we use the z-scores in Fig.~\ref{fig:hist_med} to draw inference for individual studies. In particular, we reach the following conclusions for the result reported in the abstract of ~\citet{decolonization2019medline}:
\begin{itemize}[noitemsep,leftmargin=*]
\item A shrunken point estimate for $\mu_i$ in the range $[-1.49, -1.39]$ would have lower mean squared error than just estimating $\mu_i$ by $Z_i=-2.22$. This reflects the potential effect of exaggeration of the original point estimate.
\item If we conduct an idealized replication of the study of~\citet{decolonization2019medline}, then the probability that the p-value of the replication will also be significant with the same sign for the hazard ratio is between $30.8\%$ to $33.4\%$.
\item The posterior probability that $\mu_i$ has the same sign as $Z_i$ is at least $90.7\%$.
\end{itemize}

These claims are derived from confidence intervals for the corresponding posterior functionals, rather than from point estimates alone. Such confidence intervals for empirical Bayes (EB) estimands are rarely reported in practice.  As emphasized by~\citet{ignatiadis2022confidence}, EB procedures can be highly sensitive to the estimated prior, so point summaries by themselves can be misleading. Selective reporting only heightens this sensitivity, which is precisely why we report confidence intervals throughout.

This paper introduces a general statistical framework for reaching conclusions of the above type. While our inference requires strong statistical modeling assumptions---especially about selection mechanisms---we carefully develop the rationale and justification for each in Section~\ref{sec:assumptions}.
Our framework integrates three strands of research (see Section~\ref{sec:related_work} for related work): (i) selective inference adjustments that account for publication bias to learn about collections of published studies; (ii) frequentist uncertainty quantification for EB procedures; (iii) truncation models and (empirical) Bayesian inference.
Our main contributions are as follows.
\begin{enumerate}[noitemsep,leftmargin=*]
\item \textbf{Confidence intervals for EB analysis under truncation.}
To generalize the confidence intervals of~\citet{ignatiadis2022confidence},
we exploit an observational equivalence between two Bayesian selection mechanisms discussed by~\citet{yekutieli2012adjusted}, formalized through selective tilting and untilting operations. We construct both simultaneous confidence intervals and shorter pointwise confidence intervals, each with coverage guarantees (Section~\ref{sec:methodology}). (Meanwhile, bootstrap approaches may undercover, see Section~\ref{subsec:bootstrap}.)
\item \textbf{Sign-agnostic and identifiable estimands.} As we do not trust signs in the z-scores (a positive z-score need not indicate benefit, and investigators may choose contrast directions post-hoc), we propose new estimands and show that these and other existing estimands are identifiable without sign information.
Our estimands (Section~\ref{sec:estimands_results}) include:
\begin{enumerate}[noitemsep,leftmargin=*]
\item The symmetrized posterior mean (Section~\ref{sec:symm_posterior_mean}) for shrinkage estimation, previously informally used by~\citet{vanzwet2024evaluating}. We back
it with two decision theoretic justifications: one as the mean squared error optimal denoiser subject to the constraint of being sign-equivariant (odd), and second as the minimax optimal solution among all priors for $\mu$ in~\eqref{eq:snr_prior} that imply the same prior for $\abs{\mu}$ .
\item Posterior probabilities about idealized in-silico replications $Z_i'$ conditional on $\lvert Z_i\rvert$ as surrogates of actual costly replications (Section~\ref{sec:repl_posterior_estimands}).
These answer questions such as: if we  repeat the experiment, then what is the probability that $Z_i'$ will be significant and have the same sign as $Z_i$? What is the probability that the confidence interval for $\mu_i$ based on $Z_i'$ contains $Z_i$? What is the probability that $|Z_i'|\geq |Z_i|$? 
\item The risk ratio of publication of a significant result versus a non-significant result~\citep{hedges1992modeling} (Section~\ref{sec:publication_prob}).
\end{enumerate}
\item \textbf{Analysis of MEDLINE.} We conduct a comprehensive analysis of the MEDLINE dataset. For each estimand, we report confidence intervals under three nested classes of SNR distributions, enabling assessment of sensitivity to the assumed SNR class (Section~\ref{sec:estimands_results}).  We show that substantial uncertainty from extrapolation persists even with our large sample size, in ways that depend on the estimands in a nuanced manner, highlighting the value of reporting confidence intervals rather than point estimates alone.
\end{enumerate}
Finally, in Section~\ref{sec:additional_results} we conduct several supplementary analyses. We analyze the Cochrane Database of Systematic Reviews (Section~\ref{subsec:Cochrane_analysis}), where publication bias is thought to be limited, applying our framework with and without the selection adjustment. This comparison quantifies the precision cost of guarding against publication bias when such bias may be limited. We also conduct simulation studies where the ground truth is known (Section~\ref{subsec:bootstrap}) to empirically demonstrate the coverage of our proposed intervals.

\section{Modeling assumptions}
\label{sec:assumptions}
Our goal is to learn functionals of the SNR distribution from the observed $Z_i$ in the histogram shown in Fig.~\ref{fig:hist_med}. Our basic modeling assumption is encoded in~\eqref{eq:snr_prior} and~\eqref{eq:z_normal}, and is further modified to account for the following two concerns. 
\begin{itemize}[leftmargin=*]
    \item \textbf{Selection into abstracts}: Not all z-scores appear in abstracts. The histogram shows a large gap near $0$, with few z-scores in $(-1.96, 1.96)$, consistent with selection against non-significant results. Studies with statistically significant  results are more likely to be highlighted in abstracts, while null findings may remain unpublished, or be relegated to the main text. Even when null findings are published, confidence intervals for them are less likely to be reported in abstracts. This selection mechanism means the observed $Z_i$ provides a biased sample.
    \item \textbf{Sign information}: The histogram has more positive than negative z-scores, but sign information is difficult to interpret. The direction of a reported effect does not consistently indicate clinical benefit or harm; a positive coefficient might represent increased survival in one abstract but increased mortality in another, depending on outcome coding. Moreover, when comparing treatments $A$ and $B$ with no natural control, investigators may report $\mathrm{A}-\mathrm{B}$ or $\mathrm{B}-\mathrm{A}$ depending on which yields a positive result, introducing a data-driven asymmetry; similarly, investigators may change the coding of binary outcomes. For these reasons, we model only $\abs{Z_i}$, discarding the sign as unreliable.
\end{itemize}

In our modeling, we posit a population of $n_{\text{all}}$ latent studies described as the triplets $(\mu_i, |Z_i|, D_i) \in \RR \times\RR_{\geq 0} \times \cb{0,1}$ drawn from a distribution $\mathbb P$.
Here $\mu_i$ is the SNR of latent study $i$ (as in~\eqref{eq:snr_prior}), $|Z_i|$ is the observed absolute z-score (with $Z_i$ as in \eqref{eq:z_normal} but with its sign discarded), and $D_i \in \cb{0,1}$ is a binary indicator
of whether the study's z-score was published in the abstract of an article in MEDLINE. The indicator $D_i$ depends on multiple factors, e.g., both the researchers' choice and the journal review process. We only get to observe the absolute z-scores $|Z_i|$ of studies with $D_i=1$, and the total number of latent studies $n_{\text{all}}$ is unknown. 
More formally, we consider
a model of publication bias following~\citet{hedges1992modeling,andrews2019identification} in which  $(\mu_i, |Z_i|, D_i)$ are generated as follows for $i=1,\dotsc, n_{\text{all}}$:
\begin{equation}
\label{eq:publication_bias_model}
\mu_i \sim G,\quad\;
|Z_i| \,\,\mid \,\,\mu_i \sim |\mathrm{N}(\mu_i,1)|,\quad\;
D_i \mid (|Z_i|,\mu_i) \sim \mathrm{Ber}\!\left(\pi(|Z_i|)\right).
\end{equation}
The function $\pi(\cdot):\RR_{\geq 0}\to [0,1]$ is defined by $\pi(z):=\mathbb{P}(D=1 \mid \,\, |Z|=z)$, which determines the probability of publication of a study in terms of its absolute z-score. Here $\abs{\mathrm{N}(\mu_i,1)}$ denotes the folded normal distribution. 

To streamline notation, let $(\mu, \abs{Z}, D)$ denote a generic triplet drawn from model~\eqref{eq:publication_bias_model}.
Since we only observe absolute z-scores from $\PP{\cdot \mid D =1}$, but want to learn about $G$,
we need to make an assumption about the selection mechanism.
\begin{assum}[No publication bias after truncation]
  Let $\selection \subset \RR_{\geq 0}$ be a pre-specified measurable set with $\int_{\selection}dz >0$. We assume that:
  $$\cb{\abs{Z} \;\mid \;\p{\abs{Z} \in \selection},\;D=1}\;\;\;\; \stackrel{\mathcal{D}}{=}\;\;\;\; \cb{\abs{Z} \;\mid \; \p{\abs{Z} \in \selection}}.$$
  \label{assum:publication_bias}
\end{assum}
\noindent  The assumption specifies that publication bias does not distort
the distribution of studies with absolute z-score $\abs{Z} \in \selection$. 
Hence we can learn about the distribution $\PP{\cdot \mid \abs{Z} \in \selection}$ by restricting our attention to studies with $D=1$ and $\abs{Z} \in \selection$. Assumption~\ref{assum:publication_bias}
allows for studies with $\abs{Z} \in \selection$ not to be published, i.e., to have $D=0$.\footnote{
\citet{hedges1988selection} notes that assuming that $D=\ind(\abs{Z} \in \selection)$, as made in older literature, is not reasonable. Even studies with very large z-scores may not be published. Assumption~\ref{assum:publication_bias} does not impose such a restriction.
} We have an equivalent characterization.

\begin{prop}[A necessary and sufficient condition for Assumption~\ref{assum:publication_bias}]Under model~\eqref{eq:publication_bias_model}, 
Assumption~\ref{assum:publication_bias} holds if and only if
there exists a constant $a \in (0,1]$ such that $\pi(\abs{z}) = a$ almost everywhere on $\selection$.
\label{prop:iff_cond_a1}
\end{prop}
\noindent In practical terms, among studies with $\abs{Z} \in \selection$, the publication probability is constant, irrespective of the magnitude of $\abs{Z}$. This equivalence reveals that Assumption~\ref{assum:publication_bias} imposes a strong structural constraint on the publication probability function $\pi(z)$. \citet{Benjamini2014discussion} explain some ways in which this assumption may fail. For instance, a researcher may compute 20 z-scores and only report the largest in the abstract. Our model also does not handle sign-dependent selection (beyond the sign-flips mentioned above).
While we acknowledge that Assumption~\ref{assum:publication_bias} is a strong assumption, we note that is consistent with established practice in the literature on publication bias. For example,~\citet{hedges1992modeling} posits
that $\pi(z)$ is piecewise constant with known discontinuity points, and so Assumption~\ref{assum:publication_bias}
would be applicable if we take $\selection$ to be the halfline from the largest discontinuity point to $\infty$.
\citet{andrews2019identification} consider nonparametric identification of model~\eqref{eq:publication_bias_model}, based on e.g.,
replication studies; however their empirical strategy posits that $\pi(z)$ is constant for $z \geq 1.96$, i.e.,
for studies significant at the conventional $5\%$ level. Phrased in terms of p-values, $P_i = 2\Phi(-\abs{Z_i})$ with $\Phi$ the standard normal distribution function, our Assumption~\ref{assum:publication_bias} is identical to Assumption 1 of~\citet{HungFithian2020} and to the assumption made in Proposition 1 of~\citet[Supplement] {jagerleek2014falsepositive} for the selection rule $P_i \leq 0.05$ (equivalently, $|Z_i| \geq 1.96$). The assumption is also implicitly used in the p-curve method~\citep{Simonsohn2013PCurve} and the z-curve method~\citep{bartos2022zcurve}. 

In most of this work, we take $\selection = [2.1, \infty)$.
Assumption~\ref{assum:publication_bias} with a set $\selection'$ implies the same assumption with $\selection \subseteq \selection'$.
Hence, by choosing a smaller set $\selection$, the assumption becomes more plausible; our choice of $\selection = [2.1, \infty)$ is more
conservative compared to the common choice $\selection' = [1.96, \infty)$. In choosing  $\selection = [2.1, \infty)$, we hope to partially avoid our inference being impacted by studies with z-scores
that barely cross the $1.96$ threshold due to ``slow'' p-hacking. In Supplement~\ref{sec:alterna_trunc_set}, we also consider the alternative truncation set $\selection_\text{half} = [2.24, \infty)$, which corresponds to the ``half p-curve'' of~\citet{Simonshon2015pcurve.025} that restricts attention to $P_i\leq 0.025$ (equivalently, $\abs{Z_i} \geq 2.24$).

Assumption~\ref{assum:publication_bias} enables us to identify salient aspects of the SNR distribution $G$.
In so far as we discard sign information in $Z_i$, we cannot identify any sign information encoded in $G$. This is formalized in the following two definitions.

\begin{defi}[Symmetrized and folded prior]
Let $G$ be a distribution on $\RR$. The symmetrized distribution $\Symm{G}$ is defined via $\Symm{G}(A) := \cb{G(A) + G(-A)}/2$ for all Borel sets $A$. The folded distribution $\Fold{G}$ is defined via $\Fold{G}(A) = G(A) + G(-A)$ for any Borel set \smash{$A \subset [0,\infty)$} and \smash{$\Fold{G}((-\infty,0))=0$}. 
\end{defi}
\noindent In words, if $\mu \sim G$, then $\Fold{G}$ is the distribution of $\abs{\mu}$ and $\Symm{G}$ is the distribution of $\varepsilon \cdot \mu$, where $\varepsilon$ is an independent uniform sign flip ($\varepsilon \in \cb{\pm 1}$ and $\PP{\varepsilon=1}=1/2$). The maps $G \mapsto \Fold{G}$ and $G \mapsto \Symm{G}$ encode the same information, retaining the magnitude of SNRs drawn from $G$ but discarding their sign.
We have the following identification result.

\begin{theo}
  Suppose that model~\eqref{eq:publication_bias_model} holds along with Assumption~\ref{assum:publication_bias}. Then, $\Fold{G}$ (equivalently, $\Symm{G}$), is identified, that is, $\Fold{G}$ is uniquely determined from the observable distribution of $\cb{\abs{Z} \;\mid \;\p{\abs{Z} \in \selection},\;D=1}$.
  \label{theo:Identifiability}
\end{theo}

\noindent Insofar as we already have established identification, we next make some further assumptions that enable statistical inference, i.e., forming confidence intervals about functionals of $\Fold{G}$. Our first such assumption pertains to the independence of the studies. 

\begin{assum}[Independence]
We assume that the triplets $(\mu_i, \abs{Z_i}, D_i)$ for $i=1,\ldots, n_{\text{all}}$ are jointly independent.
\end{assum}
\noindent We consider this assumption to be reasonable to first order. In support of this, from each paper (encoded by a unique PubMed ID), we keep only a single z-score at random (see Supplement~\ref{sec:medline}). Nevertheless, the assumption does not hold exactly, since, for instance, similar patient cohorts or datasets can be used across different papers.

Next we make a structural assumption on the SNR distribution.

\begin{assum}[Class of symmetrized SNR distributions]
    We assume that
    $\Symm{G} \in \mathcal{G}$, a known convex class of distributions. Specifically, we consider the following three choices of $\mathcal{G}$ that impose increasingly weaker assumptions on $G$.
\begin{itemize}[noitemsep]
    \item $\mathcal{G}^{\mathrm{sN}}$: the class of normal scale mixtures centered at $0$.
    \item $\mathcal{G}^{\mathrm{unm}}$: the class of all distributions with a density that is unimodal about $0$.
    \item $\mathcal{G}^{\mathrm{all}}$: the class of all distributions on $\mathbb R$ with a density.
\end{itemize}
\label{assum:class_priors}
\end{assum}
\noindent Notice $\mathcal{G}^{\mathrm{sN}} \subset \mathcal{G}^{\mathrm{unm}} \subset \mathcal{G}^{\mathrm{all}}$. In Section~\ref{sec:estimands_results}, we
report our inference results for all three classes; our confidence intervals become wider as we enlarge $\mathcal{G}$. The last class, \smash{$\mathcal{G}^{\mathrm{all}}$}, imposes effectively no assumption on $\Symm{G}$.\footnote{
\label{footnote:point_null}
One might ask whether $\Symm{G}$ should be allowed to place a point
mass at $0$. Allowing an atom at $0$ would not change most results
in Section~\ref{sec:estimands_results}, since the relevant estimands are
insensitive to replacing a point mass at $0$ with a tightly concentrated
continuous component. The quantity whose interpretation genuinely depends on
the presence of an atom at $0$ is the sign-agreement probability (Section~\ref{sec:prob_same_sign}); see \citet{xie2022discussion}.}
The unimodality assumption \smash{$\mathcal{G}^{\mathrm{unm}}$} has been suggested, for example, by~\citet{cordy1997deconvolution} and \citet{stephens2017false}. Finally, \smash{$\mathcal{G}^{\mathrm{sN}}$} restricts the SNR distribution to the normal scale-mixture family \citep{efron1978how}, a well-studied class that has been used in several EB applications, including \citet{stephens2017false, zwet2021statistical, yang2024largescale}.

\section{Core ideas from prior work}
\label{sec:related_work}

\subsection{Selective inference for understanding collections of studies}
\label{subsec:selective_inf_studies}

Our methodology draws its conceptual foundations from~\citet{jagerleek2014falsepositive}. The authors collect p-values from abstracts and, under Assumption~\ref{assum:publication_bias} stated in terms of p-values (with selection region $P_i \leq 0.05$), estimate the science-wise false discovery rate, defined (in our notation) as $\int \ind\{\mu = 0\}G(\dd\mu)$. Their paper is accompanied by discussion articles~\citep{Benjamini2014discussion, Cox2014discussion, Gelman2014discussion, Goodman2014discussion, Ioannidis2014discussion, schuemie2014discussion} and a rejoinder~\citep{jager2014rejoinder}. In Supplement~\ref{sec:jager_leek} we examine the extent to which our methodology addresses these concerns.

Beyond~\citet{jagerleek2014falsepositive}, other authors have also suggested using p-values smaller than $0.05$ or absolute z-scores above $1.96$ to evaluate results from either a meta-analysis or a larger body of work. One such popular method is the p-curve~\citep{Simonsohn2013PCurve, Simonshon2015pcurve.025} that can be used to assess the evidential value of a set of studies by examining the shape of the histogram of significant p-values.
Meanwhile, z-curve~\citep{brunner2020estimating, bartos2022zcurve} works with significant z-scores and is methodologically closely related to our proposal. Although there is only a sparse description of the formal assumptions underlying z-curve, the method appears to rely on the same assumptions we make in Section~\ref{sec:assumptions}, albeit with Assumption~\ref{assum:class_priors} replaced by a different convex class of SNR distributions for $\Fold{G}$, defined as,
\begin{equation}
\mathcal{G}^{\text{z-curve}} := \cb{ \text{all distributions supported on }\, \{0,1,2,3,4,5,6\}}.
    \label{eq:z_curve_prior}
\end{equation}
Z-curve assesses uncertainty using the bootstrap; we show in Section~\ref{subsec:bootstrap} that the bootstrap can sometimes fail to provide the desired frequentist coverage in our setting. One notable difference between our work and that of~\citet{jagerleek2014falsepositive}, as well as p-curve and z-curve, is that the former focus on global properties of the entire collection of studies. While our framework also enables confidence intervals for such global estimands, we focus on posterior estimands conditional on the observed values of the absolute z-scores.

Finally, we remark that some authors, e.g.,~\citet{andrews2019identification, HungFithian2020}, use both z-scores~$Z_i$ from initial studies and z-scores~$Z_i'$ from replication studies---for instance from the Reproducibility Project: Psychology (RP:P)~\citep{opensciencecollaboration2015estimating}---to learn properties of the publication record. Since we do not have access to replication studies, we instead work under a notion of idealized replications (Section~\ref{sec:repl_posterior_estimands}).

\subsection{Empirical Bayes (EB) and confidence intervals for EB}
\label{subsec:eb_review}

A common starting point for an empirical Bayes (EB) analysis~\citep{robbins1956empirical, efron2019bayes} is to posit that we observe independent observations $X_i$, each with its own unknown parameter $\nu_i$, generated via,
\begin{equation}
    \label{eq:eb_model}
    \nu_i \sim H,\;\; X_i \sim p(\cdot \mid \nu_i),\;\;i=1,\dotsc,n,
\end{equation}
where $H$ is the unknown prior and $p(\cdot \mid \nu_i)$ is a known likelihood. Our earlier normal-SNR model in~\eqref{eq:snr_prior}--\eqref{eq:z_normal}, $\mu_i \sim G$,  $Z_i \sim \mathrm{N}(\mu_i,1)$, arises as the special case of~\eqref{eq:eb_model} obtained by taking $\nu_i = \mu_i$, $X_i = Z_i$, $H = G$, and $p(\cdot \mid \nu_i) = \mathrm{N}(\cdot \mid \mu_i,1)$.\footnote{We use different letters to highlight the generality of the framework and to avoid notational clashes later on. In Section~\ref{sec:methodology}, we will apply this framework with the truncated normal likelihood.
} 

The idea of an EB analysis is that an oracle Bayesian that knows the prior $H$, can automatically take optimal decisions. By contrast, an empirical Bayesian does not have knowledge of $H$, but can use the parallel observations $X_1,\ldots,X_n$ to learn about properties of $H$ and to then mimic decisions of the oracle Bayesian. This imitation is often accomplished by first estimating $H$ as \smash{$\widehat{H}$} and then pretending \smash{$\widehat{H}$} is the true prior.

A recent line of work, relevant to our analysis, proposes using EB methods to complement standard analyses of randomized controlled trials (RCTs)~\citep{zwet2021statistical, zwet2022proposal, vanzwet2024new, Sherry2025oncology}. The key idea is to first learn the distribution of SNRs across RCTs from a context-relevant corpus (e.g., the Cochrane Database of Systematic Reviews), and then use this estimated distribution to contextualize downstream inferences, such as the analysis of a future RCT.

Returning to EB methods more broadly, the predominant practice is to ignore uncertainty in EB estimates (and as mentioned above, to treat \smash{$\widehat{H}$} as the ``true prior''). This is often unwarranted; any uncertainty in estimating $H$ propagates into uncertainty for downstream EB estimands. A common rationale for ignoring uncertainty is that EB analyses often have a large sample size, with $n$ in~\eqref{eq:eb_model} being in the thousands or tens of thousands. Even so, substantial uncertainty may remain since estimating $H$ is a difficult deconvolution problem. In recent work,~\citet{ignatiadis2022confidence} address this shortcoming of common EB analyses by developing a  general framework for constructing confidence intervals for EB estimands.  

The framework of~\citet{ignatiadis2022confidence}, henceforth referred to as IW, requires three ingredients:
\begin{enumerate}[label=(\roman*),noitemsep,leftmargin=*]
\item The known likelihood $p(\cdot \mid \nu)$ from the EB model in~\eqref{eq:eb_model};
\item A known, convex class of priors $\mathcal{H}$ such that $H \in \mathcal{H}$;\footnote{
\label{footnote:convexhull}
In principle, any convex class $\mathcal{H}$ works. In practice, we must be able to efficiently discretize it, typically in the form $\mathcal{H} \approx \mathrm{ConvexHull}(H_1,\ldots,H_K)$ where $H_1,\ldots,H_K$ is a finite dictionary of distributions. See Supplement~\ref{sec:discr_G} for how we discretize the three SNR distributions $\mathcal{G}$ in Assumption~\ref{assum:class_priors}.}
\item A pre-specified estimand $T(H)$ that is a function of the unknown prior $H$ and is either a linear functional of $H$, i.e., $T(H) = \int \psi(\nu)\, H(\dd\nu) $ for known $\psi$, or a ratio functional of $H$, i.e., $T(H)=N(H)/D(H)$, where both $N$ and $D$ are linear functionals.
\end{enumerate}
The class of estimands allowed in (iii) accommodates all estimands that we consider below in Section~\ref{sec:estimands_results}. The class of ratio functionals includes posterior functionals of the form $T(H)=\EE[H]{t(\nu) \mid X=x} = \int t(\nu) p (x \mid \nu) H(\dd\nu) / \int p (x \mid \nu) H(\dd\nu)$, where $t$ is a known function and $x$ is fixed. Such posterior functionals are of particular interest in EB analyses because they are the device through which an empirical Bayesian mimics the oracle Bayesian.

Given the three ingredients above, IW develop 
confidence intervals called $F$-Localization intervals with finite-sample simultaneous coverage and shorter intervals with asymptotic pointwise coverage called AMARI (Affine Minimax Anderson–Rubin
Intervals). 
Earlier work has addressed inference in EB for specific combinations of likelihood, prior class, and estimand (see \citet{ignatiadis2022confidence} for an overview), but IW provide a unified approach that applies broadly. The statistical difficulty of the underlying problems---reflected in their minimax rates---depends on all three ingredients and can range from nearly parametric to ill-posed deconvolution, or even partial identification. This heterogeneity makes it challenging for generic approaches such as the bootstrap to perform reliably across settings. As we explain further below, this generality makes the methods of IW an attractive starting point for our purposes. However, the methods of IW do not account for publication bias; we show how to extend IW's framework to the selective setting.

\subsection{Truncation models and (empirical) Bayesian inference}
\label{subsec:two_trunc_models}

In our paper, we take an EB approach to the truncation problem. As emphasized by~\citet{yekutieli2012adjusted}, the role of Bayesian inference for selection problems depends on the truncation mechanism. The most commonly considered truncation mechanism, which we call ``End truncation,'' proceeds as follows.\footnote{The idea applies to general likelihoods, but we discuss it in the setting with selection on absolute z-scores.}
\begin{equation}
\mbox{\textbf{End truncation:~~~~~}}
 \mu_i \sim \gprior, \ \ \ |Z_i| \sim |\mathrm{N}(\mu_i,1)|,  \ \  \ \text{observe } |Z_i| \text{ only if } |Z_i| \in \selection.\mbox{~~~~~~~~}
\label{eq:end_truncation} \tag{A}
\end{equation}
Under the above mechanism, if we know the data-generating $G$, then Bayesian inference does not need to adjust for selection~\citep{dawid1994selection, senn2008note}, since conditioning on the observed data already renders the selection event redundant. 

There is another common truncation mechanisms that we call ``per-unit truncation.''
Let $|\TruncNormal(\mu, 1; \selection)|$ be the $|\mathrm{N}(\mu, 1)|$ distribution truncated to $\selection$ with density function:
\begin{equation}
\label{eq:selective_likelihood_normal}
p_{\selection}(z \mid \mu) = \frac{\varphi^{\text{fold}}(z;\mu)}{\Phi(\selection; \mu)} \text{ for } z \in \selection,\;\;\; p_{\selection}(z \mid \mu)=0 \text{ for } z \notin \selection.
\end{equation}
Above, $\varphi^{\text{fold}}(z;\mu) := \varphi(z-\mu)+\varphi(z+\mu)$ is the density of the folded normal distribution, where $\varphi(\cdot)$ is the standard normal density
and $\Phi(\selection; \mu) := \int_{\selection} \varphi^{\text{fold}}(z;\mu) \dd z$. Per-unit truncation proceeds as follows.
\begin{equation}
\mbox{\textbf{Per-unit truncation:~~~~~}}\mu_i \sim \gprior, \ \ \ \abs{Z_i} \sim |\TruncNormal(\mu_i,1 ; \selection)|.\mbox{~~~~~~~~~~~~~~~~~~~~~~~~~~~~~~~~~~~~~}    
\label{eq:per_unit_truncation} \tag{B}
\end{equation}
An important insight of~\citet{yekutieli2012adjusted} is that under the truncation model in~\eqref{eq:per_unit_truncation}, Bayesian inference must also adjust for selection, as the selection cannot be expressed as an event on the sample space $(\mu, |Z|)$. While our Assumption~\ref{assum:publication_bias} implies model~\eqref{eq:end_truncation}, in Section~\ref{sec:methodology} we show that model~\eqref{eq:per_unit_truncation} plays a key role in our methodology by enabling the use of IW's framework. The two truncation mechanisms in~\eqref{eq:end_truncation} and~\eqref{eq:per_unit_truncation} appear under various names in the literature; we summarize some of these in Table~\ref{tab:terminology_truncation_model}.

The two models~\eqref{eq:end_truncation} and~\eqref{eq:per_unit_truncation} may be contrasted through the form of their marginal density:
\begin{equation}
  \label{eq:two_marginal_densities}
f^{A}_{G}(z) = \frac{\int  \varphi^{\text{fold}}(z;\mu) G(\dd\mu)}{\int_{\selection} \int \varphi^{\text{fold}}(z;\mu) G(\dd\mu) \dd z} \ind(z\in \selection),\,\;\; f^B_G(z) = \int  \frac{\varphi^{\text{fold}}(z;\mu) \ind(z\in \selection)}{\Phi(\selection; \mu)}\, G(\dd \mu).
\end{equation}
We emphasize that these models are different statistically and conceptually. Model~\eqref{eq:per_unit_truncation} would be justified if each scientific team, upon deciding on their hypothesis and experiment of interest, kept repeating the exact same experiment until they obtained a statistically significant result which is subsequently published. In contrast, Model~\eqref{eq:end_truncation} is justified under a model of scientific discovery wherein each scientific team performs a single experiment and publishes the result if it is statistically significant, 
otherwise the whole hypothesis and experiment are discarded, and a new hypothesis is pursued.

How does EB interact with these truncation models? On one hand, one strand of the literature including, e.g.,~\citet{efron2011tweedie, hwang2013empirical}, considers the case wherein we observe all samples without any truncation and can use these to estimate the prior. However, afterwards we are only interested in inference or estimation of parameters selected as in~\eqref{eq:end_truncation}. In this case, we can proceed using the usual EB rule, without further adjustment for selection. Papers that consider estimation of the prior from only truncated samples include~\citet{park2010estimation, greenshtein2022generalized, greenshtein2024consistent} as well as works that focus on zero-truncation for the Poisson likelihood~\citep{bohning2006equivalence,efron2019bayes}. See~\citet{rasines2022empirical} for a review of EB and selective inference.

\section{Methodology: EB inference with selective tilting}
\label{sec:methodology}

Our high-level strategy to inference proceeds in two main steps.\\

\newlength{\panelht}
\setlength{\panelht}{4.2cm}
\begin{figure}
  \centering
 \subcaptionbox{\strut\label{fig:selectionProcess}}[0.49\textwidth]{
   \begin{minipage}[t][\panelht][t]{\linewidth}
      \begin{adjustbox}{max width=\linewidth, valign=t}
        \begin{tikzpicture}[node distance=1.5cm, auto, baseline=(current bounding box.north)]
        \node (start) {$(\mu_1, \abs{Z_1}, D_1), \dotsc, (\mu_{n_{\text{all}}}, \abs{Z_{n_{\text{all}}}}, D_{n_{\text{all}}}) \sim \mathbb{P}$};
        \node (published) [below of=start] {$\{\abs{Z_{i_1}}, \dotsc, |Z_{i_{n_{\text{published}}}}|\} \subset \{\abs{Z_1}, \dotsc, \abs{Z_{n_{\text{all}}}}\}$};
        \node (final) [below of=published] {$\{\abs{Z_{j_1}}, \dotsc, |Z_{j_{n_{\text{trun}}}}|\} \subset \{\abs{Z_{i_1}}, \dotsc, |Z_{i_{n_{\text{published}}}}|\}$};

        \draw[->] (start) -- (published) coordinate[midway] (mid1);
        \draw[->] (published) -- (final) coordinate[midway] (mid2);

        \node (condition1) [right=0.5cm of mid1] {Selection given $D=1$};
        \node (condition2) [right=0.5cm of mid2] {Selection where $\lvert Z\rvert \in \selection$};
      \end{tikzpicture}
    \end{adjustbox}
    \end{minipage}
  }\hfill
 \subcaptionbox{\strut\label{fig:selection_hist}}[0.49\textwidth]{
    \begin{minipage}[t][\panelht][t]{\linewidth}
      \includegraphics[width=\linewidth]{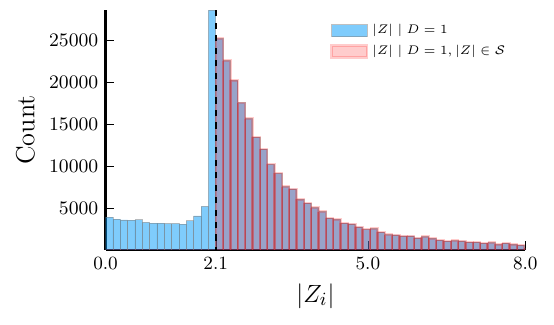}
 \end{minipage}
  }
  \caption{Selection process and resulting distributions: (a) publication selection ($D=1$) then analyst truncation ($\lvert Z\rvert\in\selection$); (b) empirical absolute $z$-score distribution under the analyst’s truncation.}
\end{figure}

\noindent \textbf{Step 1: Further filtering of z-scores.} Recall from model~\eqref{eq:publication_bias_model} that we only get to observe absolute z-scores with $D_i=1$ (first arrow in Fig.~\ref{fig:selectionProcess} and blue histogram in Fig.~\ref{fig:selection_hist}). We further discard all absolute z-scores with $\abs{Z_i} \notin \mathcal{S}$ (second arrow in Fig.~\ref{fig:selectionProcess}, pink histogram in Fig.~\ref{fig:selection_hist}). 
This leaves us
with absolute z-scores $\abs{Z_1},\ldots,\abs{Z_{n_\mathrm{trun}}}$ that satisfy both $D_i=1$ and $\abs{Z_i} \in \selection$,
labeled without loss of generality consecutively as $1,\ldots,  n_{\mathrm{trun}}$. For the MEDLINE dataset, we have $n_{\mathrm{trun}}=247,447$.\\

\noindent \textbf{Step 2: Empirical Bayes (EB) confidence intervals.} Using Assumption~\ref{assum:publication_bias}, we may treat $\abs{Z_1},\ldots,\abs{Z_{n_\mathrm{trun}}}$ above as iid samples from the distribution of $|Z| \mid (\abs{Z} \in \mathcal{S})$, that is, as samples generated via the end-truncation mechanism~\eqref{eq:end_truncation}.  
Hence our next goal is to apply the methods of~\citet{ignatiadis2022confidence} that we reviewed in Section~\ref{subsec:eb_review} to form confidence intervals for EB estimands of interest.\\

\noindent A key technical challenge in implementing Step 2 is that the techniques of IW do not directly apply to model~\eqref{eq:end_truncation} as it is not of the general form specified in~\eqref{eq:eb_model}. Briefly, IW crucially rely on the linearity of the map from the prior distribution to the marginal distribution,
$
H \mapsto f_H(\cdot) := \int p(\cdot \mid \nu) H(\dd\nu)
$. By contrast, the mapping of the prior to the marginal density under model~\eqref{eq:end_truncation}, $G \mapsto f_G^A$, shown in~\eqref{eq:two_marginal_densities}, is not linear. Below we introduce a selective tilting procedure that enables us to extend the methods of IW to model~\eqref{eq:end_truncation}. We describe an observational equivalence result underlying our construction in Section~\ref{subsec:simple_tilting} and then we describe our proposed inference procedures in Section~\ref{subsec:description_inference}.

\subsection{Selective tilting}
\label{subsec:simple_tilting}
Recall the three ingredients for EB confidence intervals from Section~\ref{subsec:eb_review}: (i) the likelihood $p(\cdot \mid \nu)$, (ii) the convex class of priors $\mathcal{H}$, and (iii) the estimand $T(H)$, where $H$ is the unknown prior. 

Next, consider the following thought experiment. Suppose the truncation mechanism would follow per-unit truncation as in~\eqref{eq:per_unit_truncation} instead of end truncation as in~\eqref{eq:end_truncation}. Then the approach of IW (reviewed in Section~\ref{subsec:eb_review}) would be directly applicable without any modifications by setting \smash{$\nu_i = \mu_i$}, \smash{$X_i = \abs{Z_i}$, \smash{$H = G$}}, and \smash{$p(\cdot \mid \nu_i) = |\mathrm{TruncN}(\cdot \mid \mu_i,1; \selection)|$}, where \smash{$|\mathrm{TruncN}(\cdot \mid \mu_i,1; \selection)|$} represents the likelihood of \smash{$|\mathrm{TruncN}(\mu_i,1; \selection)|$} as defined in~\eqref{eq:selective_likelihood_normal}. 

The key idea is that we can actually apply IW under~\eqref{eq:per_unit_truncation} and recover confidence intervals under model~\eqref{eq:end_truncation}. Suppose the analyst specifies model~\eqref{eq:end_truncation} with $G \in \mathcal{G}$ (as in Assumption~\ref{assum:class_priors}) and the estimand $T(G)$ of scientific interest (see Section~\ref{sec:estimands_results}). Then we show that there exist maps $\Tilt{\cdot}$ that act on priors, classes of priors, and functionals (Fig.~\ref{fig:tilting_schematic}), such that 
if we apply IW with (i) $p(\cdot \mid \nu_i) = |\mathrm{TruncN}(\cdot \mid \mu_i,1; \selection)|$, (ii) $\mathcal{H}=\Tilt{\mathcal{G}}$ and (iii) estimand $\Tilt{T}(\Tilt{G})$ , then we recover valid confidence intervals for $T(G)$.\\

\noindent \textbf{Tilting of priors.} We first define the tilting operation, as defined for priors,
\begin{equation}
  \label{eq:tilt_G}
\Tilt{G}(\dd\mu) := \frac{\Phi(\selection; \mu) G(\dd\mu)}{ \int \Phi(\selection; \mu) G(\dd\mu)}. 
\end{equation}
As an example, Fig.~\ref{fig:tilted} shows  $G=\mathrm{N}(0,2)$ and $\Tilt{G}$. Meanwhile, Fig.~\ref{fig:normalizer} shows how $\Phi(\selection; \mu)$ changes with $\mu$. The tilted $\Tilt{G}$ 
places less mass near $0$ at which point the selection probability $\Phi(\selection; \mu)$ is  small. 

The tilting operation in~\eqref{eq:tilt_G} induces the following observational equivalence.

\begin{theo}[Observational equivalence]
Fix a prior $G$. 
Then, the marginal distribution of $|Z|$ under model~\eqref{eq:end_truncation} with prior $G$ is equal to the marginal distribution of $|Z|$ under model~\eqref{eq:per_unit_truncation} with prior $\Tilt{G}$. Stated in terms of marginal densities~\eqref{eq:two_marginal_densities}, $f_G^A(\cdot)= f_{\Tilt{G}}^B(\cdot)$. 

\label{theo:observational_equivalence}
\end{theo}
\noindent A related equivalence has been derived in~\citet{bohning2006equivalence} and~\citet[Remarks D and G]{efron2019bayes} for zero-truncated Poisson samples (see Supplement~\ref{sec:butterflies}). 

\begin{figure}
    \centering
     \begin{subfigure}[t]{0.32\textwidth}
        \centering
        \caption{}
        \includegraphics[height=\panelht, width=\linewidth, keepaspectratio]{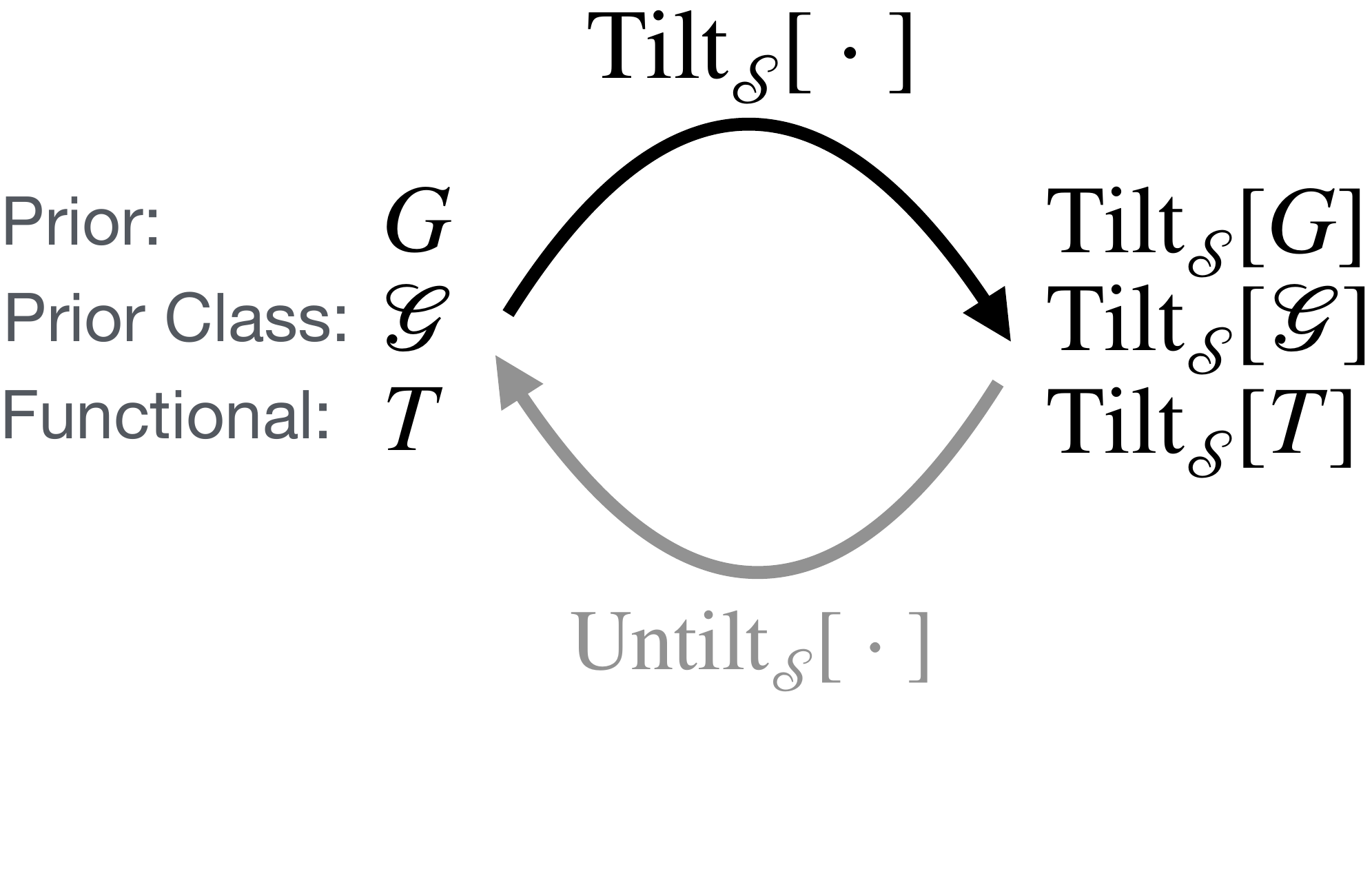}
        \label{fig:tilting_schematic}
    \end{subfigure}
    \begin{subfigure}[t]{0.32\textwidth}
        \centering
        \caption{}
        \includegraphics[height=\panelht, width=\linewidth, keepaspectratio]{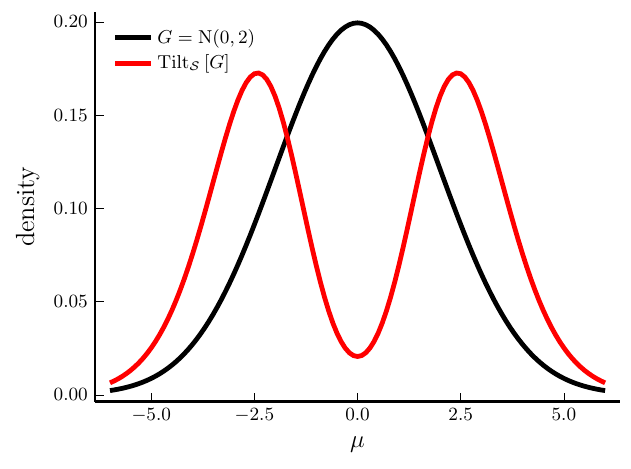}
        \label{fig:tilted}
    \end{subfigure}
    \hfill
    \begin{subfigure}[t]{0.32\textwidth}
        \centering
        \caption{}
        \includegraphics[height=\panelht, width=\linewidth, keepaspectratio]{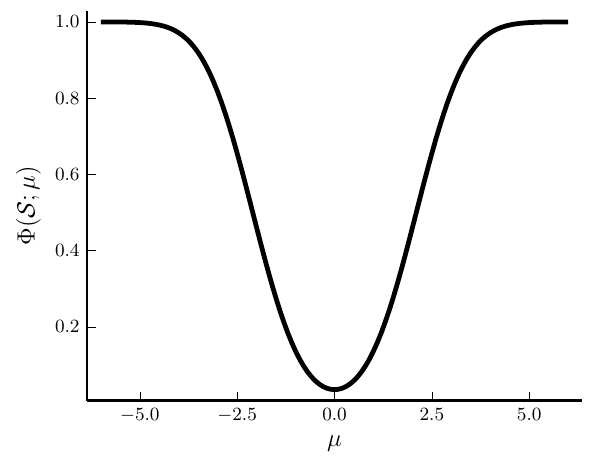}
        \label{fig:normalizer}
    \end{subfigure}
    \caption{Schematic demonstration of selective tilting and example of a $\Tilt{G}$ where $G = \mathrm{N}(0,2)$: (a) illustration of the mapping $\Tilt{\cdot}$ and $\Untilt{\cdot}$; (b) The density of $G = \mathrm{N}(0,2)$ and $\Tilt{G}$; (c) the corresponding $\Phi(\selection; \mu)$, note $\Phi(\selection; 0) \neq 0$}.
    \label{fig:tilting_illustrations}
\end{figure}

\begin{rema}[Untilting]
\label{rema:untilt}
Tilting can be reversed using an untilt operation, defined as \\
$\Untilt{\tG}(\dd\mu) := \Phi(\selection; \mu)^{-1} \tG(\dd\mu)/\int \Phi(\selection; \mu)^{-1} \tG(\dd\mu)$ for any $\tG \in \Tilt{\mathcal{G}}$.
Analogously to Theorem~\ref{theo:observational_equivalence}, we have that the marginal distribution of $|Z|$ under model \eqref{eq:per_unit_truncation} with prior $\tG$ is equal to the marginal distribution of $|Z|$ under model~\eqref{eq:end_truncation} with prior $\Untilt{\tG}$.
\end{rema}

\noindent \textbf{Tilting of convex classes of priors.} Let $\mathcal{G}$ be a class of priors. We define the tilted class of priors as,
\begin{equation}
\Tilt{\mathcal{G}} := \cb{ \Tilt{G}\,:\; G \in \mathcal{G}}.
\end{equation}
Recall that the methods of IW require a convex class of priors. The following proposition establishes that the $\Tilt{\cdot}$ operation maintains convexity of the input class of priors.

\begin{prop}
\label{prop:convex_equivalence}
Suppose $\mathcal{G}$ is a convex class of priors, that is $\lambda G_1 + (1-\lambda)G_2 \in \mathcal{G}$ for any $\lambda \in [0,1]$ and $G_1,G_2 \in \mathcal{G}$. Then $\Tilt{\mathcal{G}}$ is also a convex class of priors.
\end{prop}

\noindent \textbf{Tilting of functionals.} Finally, we explain how to tilt functionals. Let $T(\cdot)$ be a ratio functional, that is, $T(\cdot) = N(\cdot)/D(\cdot): \mathcal{G} \to \mathbb R$, where $N(\cdot)$ and $D(\cdot)$ are linear functionals with $N(G)= \int \nu(\mu)G(\dd\mu)$ and $D(G) = \int \delta(\mu)G(\dd\mu)$ for some known $\nu(\cdot)$, $\delta(\cdot)$. Then $\Tilt{T}: \Tilt{\mathcal{G}} \to \RR$ is the ratio functional defined via
\begin{equation}
    \Tilt{T}(\tG) :=\frac{ \int \nu(\mu)\Phi(\selection; \mu)^{-1}\tG(\dd\mu)}{\int \delta(\mu) \Phi(\selection; \mu)^{-1}\tG(\dd\mu)},\;\; \tG \in \Tilt{\mathcal{G}}.
\label{eq:functional_tilt}
\end{equation}
In case $T(\cdot)$ is a linear functional, i.e., $\delta(\cdot)\equiv 1$, we can also apply the above transformation. In this case, tilting turns the linear functional into a ratio functional.  We have the following identity.

\begin{prop}[Functional equivalence]
Let $T(\cdot)$ be a ratio functional. Then, 
\begin{equation*}
    \Tilt{T}(\Tilt{G}) = T(G) \;\text{ for all }\; G \in \mathcal{G}. \label{eq:linear_ratio_representation}
\end{equation*}
\label{prop:functional_equivalence}
\end{prop}
\noindent \textbf{Inference after tilting.} To recap, here's what we have achieved. By Proposition~\ref{prop:functional_equivalence}, we can write our estimand $T(G)$ as $\Tilt{T}(\Tilt{G})$. Next, Theorem~\ref{theo:observational_equivalence} implies that we can pretend our truncated $|Z_i|$ came from model~\eqref{eq:per_unit_truncation} with the tilted prior $\Tilt{G}$.  Therefore, to get a confidence interval for $T(G)$ under model~\eqref{eq:end_truncation} with prior $G$, it suffices to develop a confidence interval for $\Tilt{T}(\Tilt{G})$ under model~\eqref{eq:per_unit_truncation} with prior $\Tilt{G}$. The following theorem formalizes this result.
\begin{theo}
\label{theo:ci_reduction}
Let $|Z_1|,\ldots,|Z_n|$ be nonnegative observations and let $\mathcal{I} \equiv \mathcal{I}(\abs{Z_1},\ldots,\abs{Z_n})$ be a random interval depending on the data. Also let $T(\cdot)$ be a ratio functional. Then,
$$
\mathbb P_G^A[ T(G) \in \mathcal{I}] = \mathbb P_{\Tilt{G}}^B[ \Tilt{T}(\Tilt{G}) \in \mathcal{I}],
$$
where the notation $\mathbb P_G^A$, resp.  $\mathbb P_G^B$ refers to $|Z_1|,\ldots,|Z_n|$ independently generated from~\eqref{eq:end_truncation} with prior $G$, resp. from~\eqref{eq:per_unit_truncation} with prior $\Tilt{G}$.
\end{theo}
The reduction of this theorem applies to any confidence interval procedure $\mathcal{I}$. In our case, our interest is driven by the fact that the general framework of IW (and their coverage theorems) directly apply to any ratio functional in model~\eqref{eq:per_unit_truncation} with a convex class of priors. If we start with $G \in \mathcal{G}$ for convex $\mathcal{G}$, then Proposition~\ref{prop:convex_equivalence} also implies that $\Tilt{G} \in \Tilt{\mathcal{G}}$, another convex class of priors. 

All equivalences above extend to convex hulls of finite dictionaries of priors (as in Footnote~\ref{footnote:convexhull}). See Supplement~\ref{subsec:tilting_computation} on these equivalences and their computational consequences.

\subsection{Description of inferential approaches}
\label{subsec:description_inference}

Given the reduction of Theorem~\ref{theo:ci_reduction}, we now give a brief description of the two methods of IW, $F$-Localization and AMARI, in our setting.\\

\noindent{\textbf{$F$-Localization.}}
Let \smash{$F_{\Tilt{G}}^B(\cdot)$} be the marginal CDF of $\abs{Z}$ under model~\eqref{eq:per_unit_truncation} with prior \smash{$\Tilt{G}$}.
An $F$-Localization is defined as a level $1-\alpha$ confidence set for this marginal distribution (for $\alpha \in (0,1)$). One specific construction of an $F$-Localization uses the Kolmogorov-Smirnov ball around the empirical distribution of truncated samples,
\begin{equation}
\mathcal{F}_{n_\text{trun}}^{\text{DKW}}(\alpha) := \left\{\text{F distribution}\;:\; \sup_{t\in \selection}|F(t) - \hat{F}_{n_\text{trun}}(t)| \leq \sqrt{\frac{\log(2/\alpha)}{2n_\text{trun}}}\right\}
\label{eq:KS-ball},
\end{equation}
where 
$
\hat{F}_{n_\text{trun}}(t) := n_\text{trun}^{-1}\sum_{i=1}^{n_\text{trun}}\ind(\abs{Z_i} \leq t, \abs{Z_i} \in \selection, D_i=1).
$
By the Dvoretzky-Kiefer-Wolfowitz  inequality~\citep{Massart1990DKW}, we have that \smash{$\mathbb P_{\Tilt{G}}^{B}[F_{\Tilt{G}}^B \in \mathcal{F}_{n_\text{trun}}^{\text{DKW}}(\alpha)] \geq 1-\alpha$}, i.e., \smash{$\mathcal{F}_{n_\text{trun}}^{\text{DKW}}(\alpha)$} is an $F$-Localization.
We can then find the smallest possible value of \smash{$\Tilt{T}(\tG)$} among all priors \smash{$\tG \in \Tilt{\mathcal{G}}$} consistent with the $F$-Localization,
\begin{equation}
\hat{T}_\alpha^- := \text{inf}\left \{\Tilt{T}(\tG)\;:\;  \tG \in \Tilt{\mathcal{G}},\, F_{\tG}^B \in \mathcal{F}_{n_\text{trun}}^{\text{DKW}}(\alpha)\right\},
\label{eq:DKW_F_loc_lower}
\end{equation}
and analogously for the upper bound $\hat{T}_\alpha^+$. Using Theorem~\ref{theo:ci_reduction}, it follows that $\mathbb P_G^A[T(G) \in [\hat{T}_\alpha^-,\hat{T}_\alpha^+] \geq 1-\alpha$. Moreover, the probability statement is simultaneous over all possible functionals $T(\cdot)$ we may be interested in.\footnote{This is important for our application, since we report confidence intervals for a lot of different estimands.} We refer to Supplement~\ref{subsec:floc_comp} for computational details.

\noindent{\textbf{AMARI.}}
The second construction from IW is AMARI (Affine Minimax Anderson Rubin Intervals), which aims to construct pointwise confidence intervals for a specific empirical Bayes estimand rather than achieving simultaneous coverage over all estimands, typically resulting in shorter intervals. We refer to~\citet{ignatiadis2022confidence} for details and to Supplement~\ref{subsec:amari_comp} for a brief sketch.

\begin{rema}[Selective tilting for other EB approaches]
\label{rema:tilting_for_other_approaches}
Although we have explained how selective tilting enables us to use the methods of IW, the idea applies more broadly. For example, the nonparametric maximum likelihood estimator (NPMLE) in the EB model~\eqref{eq:eb_model} can be computed via convex optimization~\citep{koenker2014convex}. If instead we seek to compute the NPMLE under model~\eqref{eq:end_truncation}, then selective tilting permits us to compute the NPMLE under model~\eqref{eq:per_unit_truncation}, and to then untilt. For the zero-truncated Poisson problem, this approach has been suggested by~\citet{bohning2006equivalence} and, without details, by~\citet[Section 2]{greenshtein2022generalized} in a model involving post-stratification.
\end{rema}

\section{Inference for estimands of interest in MEDLINE}
\label{sec:estimands_results}

We are ready to turn to our main objective, the analysis of the MEDLINE dataset. We proceed as follows: in each case, we describe the estimand of scientific interest, and then we report confidence intervals for it. By default, these are 95\% $F$-Localization intervals with simultaneous coverage. In our main figure (Fig.~\ref{fig:medline_full_analysis}) for each estimand we thus show three confidence bands, each one corresponding to the three nested assumptions for the SNR distribution (Assumption~\ref{assum:class_priors}). For certain targeted estimands, we  also report AMARI which has asymptotic pointwise coverage. In particular, the confidence intervals from the introduction are constructed using AMARI with $\mathcal{G}^{\mathrm{unm}}$.

In view of our modeling assumptions in Section~\ref{sec:assumptions}, we focus on estimands that are a function of $\Fold{G}$ and are thus identifiable by Theorem~\ref{theo:Identifiability}. Formally:

\begin{prop}[Identifiability of estimands]
    All estimands below are functions of $\Fold{G}$ only,  thus are identifiable.
\label{prop:IdentifiabilityAll}
\end{prop}

\begin{figure}[p]
\centering
\begin{adjustbox}{max totalsize={\textwidth}{0.9\textheight}, center}
\begin{minipage}{\textwidth}
    \begin{subfigure}[t]{0.45\textwidth}
        \centering
        \caption{Marginal density}
        \includegraphics[width=\linewidth]{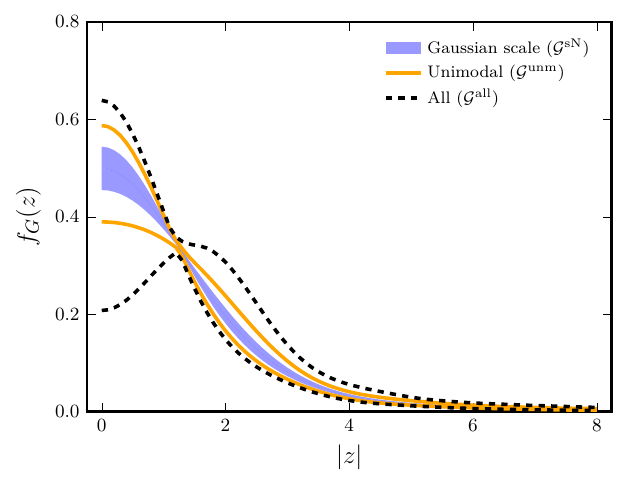}
        \label{fig:marginal_density_unnormalized_full}
    \end{subfigure}
    \hfill
    \begin{subfigure}[t]{0.45\textwidth}
        \centering
        \caption{Normalized marginal density}
        \includegraphics[width=\linewidth]{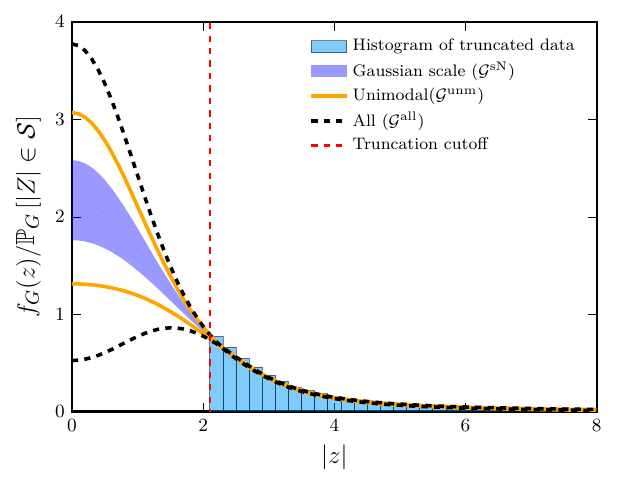}
        \label{fig:marginal_density_normalized_full}
    \end{subfigure}
    \hfill
    \begin{subfigure}[t]{0.45\textwidth}
        \centering
        \caption{Binned density of power}
        \includegraphics[width=\linewidth]{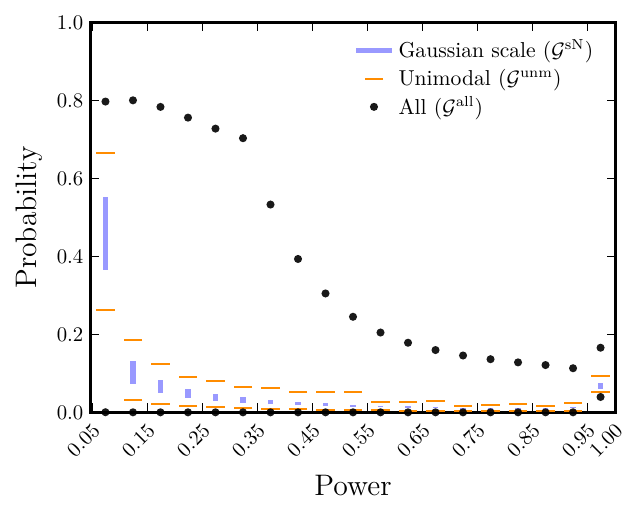}
        \label{fig:binned_power_full}
    \end{subfigure}
    \hfill
    \begin{subfigure}[t]{0.45\textwidth}
        \centering
        \caption{Probability of same sign}
        \includegraphics[width=\linewidth]{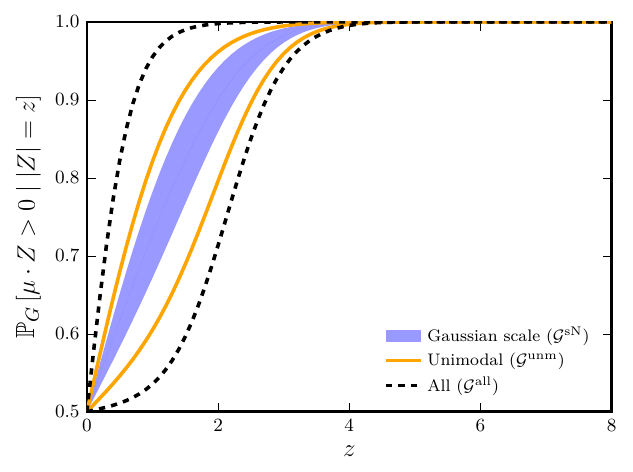}
        \label{fig:same_sign_full}
    \end{subfigure}
    \hfill
    \begin{subfigure}[t]{0.45\textwidth}
        \centering
        \caption{Symmetrized posterior mean}
        \includegraphics[width=\linewidth]{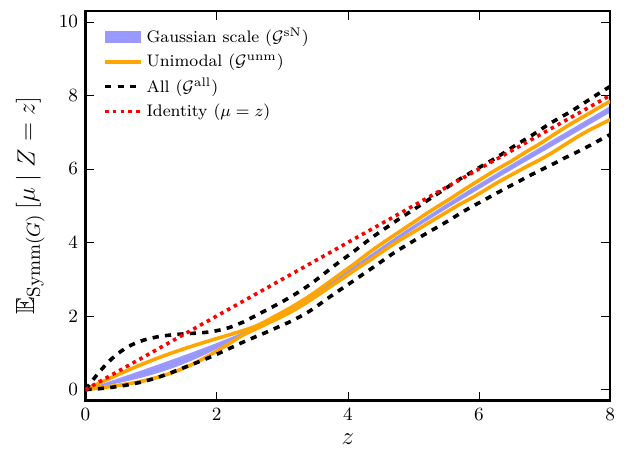}
        \label{fig:symmetrized_posterior_full}
    \end{subfigure}
    \hfill
    \begin{subfigure}[t]{0.45\textwidth}
        \centering
        \caption{Replication probability}
        \includegraphics[width=\linewidth]{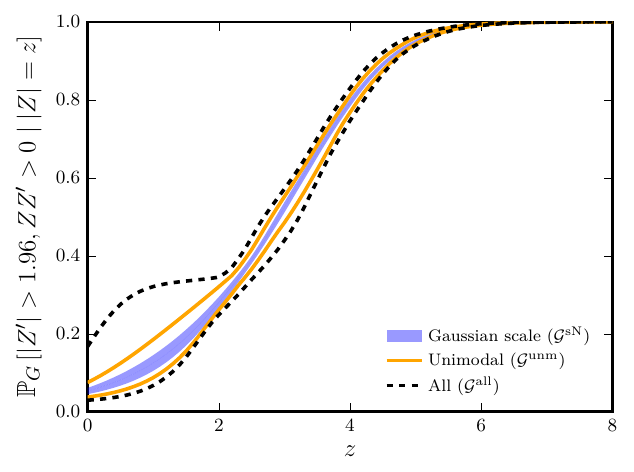}
        \label{fig:repl_prob_full}
    \end{subfigure}
    \hfill
    \begin{subfigure}[t]{0.45\textwidth}
        \centering
        \caption{Future coverage probability}
        \includegraphics[width=\linewidth]{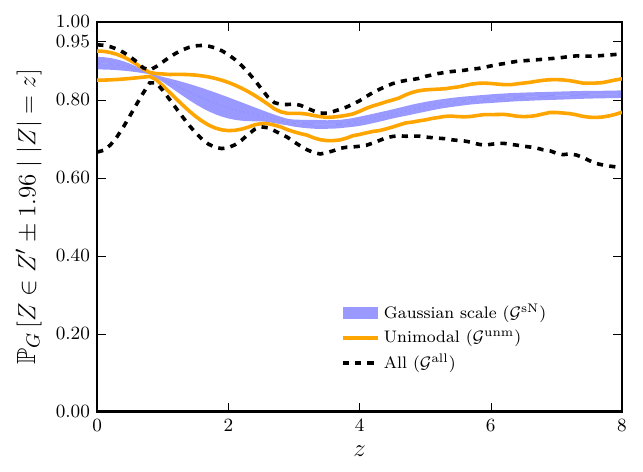}
        \label{fig:future_coverage_prob_full}
    \end{subfigure}
    \hfill
    \begin{subfigure}[t]{0.45\textwidth}
        \centering
        \caption{Effect size replication probability}
        \includegraphics[width=\linewidth]{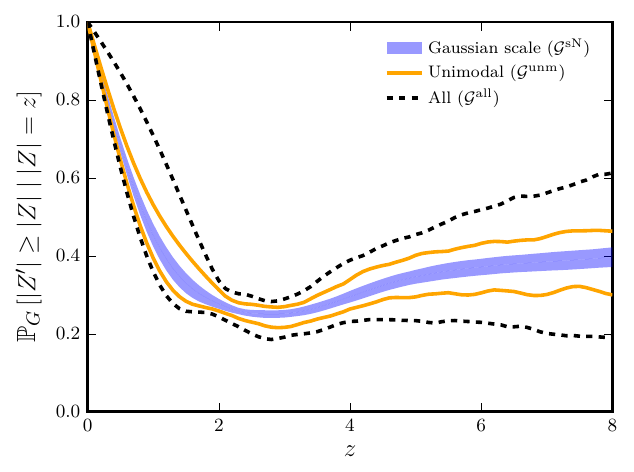}
        \label{fig:effect_repl_prob_full}
    \end{subfigure}
    
\end{minipage}
\end{adjustbox}
\caption{95\% Confidence interval analyses for MEDLINE (2000-2018): Each panel presents one estimand of interest, accompanied by 95\% confidence intervals under different assumptions for the SNR distribution.}
\label{fig:medline_full_analysis}
\end{figure}

\subsection{Marginal densities and power}
\label{subsec:marginal_and_power}
Our first estimands inform us about the distribution of z-scores in MEDLINE abstracts. We discuss estimands for posterior inference afterwards.\\

\noindent \textbf{Marginal density.}
Let $f_G(z)$ be the marginal density of $\abs{Z}$ under \eqref{eq:publication_bias_model}, 
$$
f_G(z) := \int\varphi^{\text{fold}}(z;\mu)G(\dd\mu),\;\; z \geq 0.
$$
The marginal density characterizes the distribution of absolute z-scores in the complete population of studies (both published and unpublished). This estimand for $z \notin \mathcal{S}$ makes it clear that we are dealing with an extrapolation problem.

As shown in Fig.~\ref{fig:marginal_density_unnormalized_full}, we see that different prior class assumptions noticeably impact confidence interval widths that widen under increasingly flexible prior classes. For $\mathcal{G}^{\mathrm{all}}$ that makes almost no assumption at all, the resulting intervals are wide since we only work with $|Z_i| \geq 2.1$ and need to infer the density at  $z$ near $0$. By contrast, the structure afforded by $\mathcal{G}^{\mathrm{sN}}$ and $\mathcal{G}^{\mathrm{unm}}$ means that we can more accurately conduct this extrapolation task, at the cost of stronger structural assumptions. 

Beyond the issue of extrapolation, we observe that 
that there is also substantial uncertainty for $z \in \mathcal{S}$. The reason is that the available samples only indirectly provide information on the normalizing constant in the denominator of $f_G^A(z)$ in \eqref{eq:two_marginal_densities}; we do not know how many samples are actually truncated. To address this, we define the normalized marginal density
\begin{equation}
\begin{aligned}
f_G^{\mathrm{n}}(z):=\frac{f_G(z)}{\PP[G]{\abs{Z} \in \selection}} \;\text{ with }\;\PP[G]{\abs{Z} \in \selection} =  \int_\selection \int \varphi^{\text{fold}}(z;\mu)G(\dd\mu) \dd z.
\end{aligned}
\label{eq:marginal_density_normalized}
\end{equation}
In words, this estimand is the marginal density normalized such that it integrates to $1$ over $\selection$. In this way, the samples we observe (after the preprocessing of Fig.~\ref{fig:selectionProcess}) follow precisely this normalized density over $\selection$. To also interpret the marginal density outside $\selection$, we extend the definition to all $z \geq 0$. In Fig.~\ref{fig:marginal_density_normalized_full}, we accompany the confidence intervals for the normalized density with the histogram of the truncated observations. We observe that for $z \in \selection$ our intervals track the histogram closely with little uncertainty. As we move away from $z \in \selection$ and toward $z\approx 0$, our confidence intervals become wider. 
The substantial uncertainty near $0$ even with our large sample size underscores the fundamental difficulty of extrapolating the density of $\abs{Z}$.\\

\begin{table}
\centering
\caption{95\% Confidence intervals for proportion of studies with at least 80\% power under different prior classes on MEDLINE (2000-2018) data.}
\label{tab:power_above_80_medline}
\begin{tabular}{lcc}
\toprule
\textbf{Prior}               & \textbf{FLOC} & \textbf{AMARI} \\
\midrule
$\mathcal{G}^{\mathrm{sN}}$  & (0.091, 0.108)  & (0.106, 0.109) \\
$\mathcal{G}^{\mathrm{unm}}$ & (0.077, 0.125)  & (0.103, 0.126) \\
$\mathcal{G}^{\mathrm{all}}$ & (0.047, 0.209)  & (0.049, 0.212) \\
\bottomrule
\end{tabular}
\end{table}

\noindent \textbf{Power-based estimands.}
Denote the power function for a two-sided z-test as 
$$\beta(\mu) := \PP{\abs{Z} \geq 1.96 \mid \mu} =  1- \Phi(1.96-\mu)+\Phi(-1.96-\mu),$$
where $\Phi(\cdot)$ is the standard normal CDF. The power measures the probability of a two-sided z-test rejecting the null hypothesis $\mathrm{H}_0:\mu=0$ at level $\alpha=0.05$ when the true SNR is $\mu$. We let $\beta(0)\stackrel{\cdot}{=}0.05$.
To better understand the power distribution across studies in MEDLINE, we define the binned density of power as,
$$
\mathrm{Pow}_G(\mathcal{I}) := \PP[G]{ \beta(\mu) \in \mathcal{I}} = \int \ind( \beta(\mu) \in \mathcal{I} )G(\dd\mu),
$$
where we set $\mathcal{I}$ by partitioning $[0.05, 1.00]$ into intervals of width $0.05$. Results are shown in Fig.~\ref{fig:binned_power_full}. For instance, the confidence interval for $\mathcal{I}=[0.05,0.1]$ using $\mathcal{G}^{\mathrm{sN}}$ shows that about 37\%-55\% of results reported in MEDLINE abstracts have underlying power below $10\%$. 

While Fig.~\ref{fig:binned_power_full} looks at this more fine-grained partitioning of the density of power, we also consider this quantity for $\mathcal{I}=[0.8,1.0]$. In this case the estimand $\PP[G]{\beta(\mu) \geq 0.8}$ represents the proportion of studies with at least 80\% power. In the confidence intervals reported in Table~\ref{tab:power_above_80_medline}, we see that only between 4.7\%-20.9\% of studies (using F-Localization and $\mathcal{G}^{\mathrm{all}}$) have at least 80\% power. In interpreting this interval, we remind the reader that our analysis applies to z-scores that could have been in an abstract.

We next start discussing posterior estimands.

\subsection{Posterior sign-agreement probability}
\label{sec:prob_same_sign}

We define the sign-agreement probability as 
$$\PP[G]{\mu \cdot Z > 0 \mid \abs{Z}=z},$$
which represents the probability that the z-score and the true SNR share the same sign, conditional on a particular value of the observed absolute z-score. The reason we condition on $\abs{Z}$ instead of $Z$ is that only the former estimand is identified in our setting.
This quantity addresses a fundamental question in scientific interpretation: when a study reports a positive effect, how confident can we be that the true effect is indeed positive? This quantity relates to the Type S error advocated by~\citet{gelman2000types, gelman2014typem} (which is equal to $\PP[G]{\mu \cdot Z > 0 \mid \abs{Z}> 1.96}$ in our notation), and is also closely related to the local false sign rate of~\citet{stephens2017false}.
A probability just above 0.5 indicates that the sign is highly uncertain, which may occur e.g., for $z \approx 0$. As $|z|$ increases, this probability approaches one due to increasing strength in evidence. The rate of increase, however, depends on $\Fold{G}$. Therefore, by examining how the sign-agreement probability changes with $\abs{z}$, we can evaluate directional reliability across the spectrum of observed z-scores.
From Fig.~\ref{fig:same_sign_full}, the probability of sign-agreement is quite high for studies that just reach the 0.05 level of significance. The lower end of the confidence interval (using $\mathcal{G}^{\mathrm{sN}}$) is above 80\%.

Our claim made in the introduction that the posterior probability that $\mu_i$ has the same sign as $Z_i$ is at least $90.7\%$ (where $Z_i = -2.22$ refers to the z-score from the abstract of~\citet{decolonization2019medline}) is derived from the confidence interval $[90.7\%, 97.5\%]$ for the above using AMARI and $\mathcal{G}^{\mathrm{unm}}$. See Table~\ref{tab:example_inference} for all CIs for this specific estimand.

\subsection{The symmetrized posterior mean}
\label{sec:symm_posterior_mean}

An important question is how we should use our MEDLINE analysis to develop better shrinkage estimators. Specifically, suppose we seek to estimate $\mu$ from a study based on its z-score $Z$. Of note, for this study, we assume that its sign is reliable. The natural estimator is given by the posterior mean $\hat{\mu} := \EE[G]{\mu \mid Z}$. The following optimization problem over all possible functions $\delta:\RR \to \RR$,
\begin{equation}
\underset{\delta:\RR \to\RR}{\textnormal{minimize}}\quad\EE[G]{\p{\delta(Z)-\mu}^2},
\label{eq:posterior_mean_computation}
\end{equation}
is solved by the posterior mean function $\delta_{G}(z):= \EE[G]{\mu \mid Z=z}$. However, the posterior mean is not identifiable in our setting, because it relies on sign-information of $G$. Recalling from Theorem~\ref{theo:Identifiability} that $\Symm{G}$ is identified,  we define the symmetrized posterior mean:
\begin{equation}
\delta_{G}^{\mathrm{Symm}}(z) := \EE[ \Symm{G}]{\mu \mid Z=z}.
\end{equation}
In words, \smash{$\delta_{G}^{\mathrm{Symm}}(z)$} is the posterior mean if the prior had been the symmetrized version $\Symm{G}$ of $G$ instead of $G$ itself.\footnote{
We note that the symmetrized posterior mean is not equal to \smash{$\EE[G]{|\mu| \mid |Z|=|z|}\mathrm{sign}(z)$}. To wit, \smash{$\EE[G]{|\mu| \mid |Z|=|z|} = \EE[\Symm{G}]{|\mu| \mid |Z|=|z|} \geq |\delta_{G}^{\mathrm{Symm}}(z)|$}, where the inequality follows from Jensen's inequality and is strict as long as $G$ is not a point mass at $0$.
} Our proposal is to use \smash{$\delta_{G}^{\mathrm{Symm}}(z)$} for estimation of $\mu$ when only $\Fold{G}$ (equivalently $\Symm{G}$) is identified. While the symmetrized posterior mean has been previously used by~\citet{vanzwet2024evaluating} without being given an explicit name, here we provide a decision theoretic foundation.

\begin{prop}
Consider the following two optimization problems over functions $\delta:\RR \to \RR$ in lieu of~\eqref{eq:posterior_mean_computation}:

\begin{align*}
\text{(a)}\quad & \underset{\delta:\mathbb{R} \to \mathbb{R}}{\textnormal{minimize}}\quad \mathbb{E}_G\left[(\delta(Z)-\mu)^2\right] \quad \textnormal{s.t.} \quad \delta(\cdot) \textnormal{ is odd, i.e., } \delta(-z)=-\delta(z) \textnormal{ for all } z,\\[4pt]
\text{(b)}\quad & \underset{\delta:\mathbb{R} \to \mathbb{R}}{\textnormal{minimize}}\quad \sup\left\{ \mathbb{E}_{\tG}\left[(\delta(Z)-\mu)^2\right] \,: \,\tG \,\textnormal{ with }\, \mathrm{Symm}\,[\tG]=\mathrm{Symm}\,[G] \right\}.
\end{align*}
Then the solution of both optimization problems is given by $\delta_G^{\mathrm{Symm}}(\cdot)$, even if the original prior $G$ is not symmetric.
\label{prop:interpretation_symm_postmean}
\end{prop}

The meaning of part (a) of the proposition is that the symmetrized posterior mean minimizes mean squared error among all functions that are constrained to satisfy the natural equivariance requirement of being odd (and so treating $z$ and $-z$ in a symmetric way). The result provides a rationale for caring about \smash{$\delta_G^{\mathrm{Symm}}(\cdot)$} even when $G$ is not symmetric. We also briefly refer to~\citet{jaffe2025constrained} who study mean squared error optimal denoising subject to various constraints (that are however different from the sign-equivariance constraint we impose). The interpretation of part (b) is that it minimizes the worst-case risk over all possible priors that induce the same distribution for $\abs{\mu}$ as $G$ does. 

Fig.~\ref{fig:symmetrized_posterior_full} shows the confidence intervals for the symmetrized posterior mean. We only show the result for $z \geq 0$, since by symmetry the result for $z\leq0$ is the same with flipped signs. We observe that all intervals suggest stronger shrinkage for moderate $z$, with smaller shrinkage for larger $z$. At $z  = -2.22$, our confidence interval for the symmetrized posterior mean is $[-1.49, -1.39]$ (using AMARI and $\mathcal{G}^{\mathrm{unm}}$), indicating noticeable shrinkage towards zero. We refer to Table~\ref{tab:example_inference} for the CIs for the symmetrized posterior mean at $z  = -2.22$.

\subsection{Posterior estimands of idealized replications}
\label{sec:repl_posterior_estimands}

We now consider hypothetical idealized replications under~\eqref{eq:snr_prior} and~\eqref{eq:z_normal} according to,
$$
\mu \sim G,\;\; Z,Z' \mid \mu \simiid \mathrm{N}(\mu, 1).
$$
We only observe $\abs{Z}$ (or $Z$), but imagine we have a perfect replication that would yield $Z'$. We can then use our framework to ask questions about posterior probabilities involving the unobserved replication $Z'$. In this sense, the following estimands seek to mimic what would happen under an actual replication study without actually running it. For interpretation, it is important to keep in mind that there is no such thing as perfect replication. For instance, the results on replication probabilities below, show that  \emph{even if it were possible to do a perfect replication}, the probability of replicating is much smaller than one might anticipate.

We focus on three posterior estimands that are motivated by the three reported measures of the~\citet{opensciencecollaboration2015estimating} and attempt to provide a rigorous definition for them for the empirical Bayes setting we consider. \citet{HungFithian2020} also define estimands motivated by~\citet{opensciencecollaboration2015estimating} that are purely frequentist in nature, but require access to actual replications.\\

\noindent{\textbf{Replication probability.}} We define the replication probability as 
$$\PP[G]{|Z'| > 1.96, ZZ'>0\mid |Z|=z}$$ 
to describe the probability that a z-score from a replication achieves statistical significance 
and has the same sign as the original z-score. As before, for identifiability reasons, although our counterfactual question pertains to the actual z-score $Z$ and its replication $Z'$, we condition on the absolute value of the observed z-score. 

The posterior probability that an exact replication study will produce a z-score $Z'$ that achieves statistical significance in the same direction is small for most values of the original study z-score $Z$, as shown in Fig.~\ref{fig:repl_prob_full}. For studies with absolute z-scores just above 1.96, we obtain intervals centered around 25\% (using $\mathcal{G}^{\mathrm{sN}}$). Thus, even under perfect replication, the probability that a just-significant result replicates is small. The~\citet{opensciencecollaboration2015estimating} reports that only 36\% of replications produced statistically significant effects in the same direction as the original studies.~\Citet{erik2022replprob} also estimates the same quantity based on the Cochrane database, where they report a replication probability of about 29\% for studies that just achieve statistical significance. If we measure it at $\abs{z}  = 2.22$ (as in the study of~\citet{decolonization2019medline}), we have a confidence interval of $[30.8\%, 33.4\%]$ (using AMARI and $\mathcal{G}^{\mathrm{unm}}$).  We refer to Table~\ref{tab:example_inference} for all CIs for this specific estimand.\\

\noindent{\textbf{Future coverage probability.}}
The future coverage probability is defined as
$$
\PP[G]{ Z \in Z' \pm 1.96       \mid \abs{Z}=z},
$$
and measures the probability that the 95\% confidence interval from a replication study contains the initial z-score $Z$, again conditional on the absolute value of the initial z-score.\footnote{A related estimand is the conditional coverage of the usual 95\% confidence interval for the true SNR $\mu$, i.e., $\PP[G]{ \mu \in Z \pm 1.96  \mid \abs{Z}=z}$, previously considered e.g., by~\citet{zwet2021statistical}. This estimand can also be handled by our methods but is omitted for brevity.}

Overall, the future coverage probability in Fig.~\ref{fig:future_coverage_prob_full} is around 80\% using $\mathcal{G}^{\mathrm{sN}}$, with a dip for moderate absolute z-scores around $2$ to $4$. At absolute z-scores of around $1.96$, this probability is much higher than the replication probability. A related estimand is reported by~\citet{opensciencecollaboration2015estimating}, where 47\% of 95\% confidence intervals of the replicated effect size estimates contain the original estimate. This number is quite a bit lower than what our confidence intervals indicate; a potential explanation lies in the difference between real replications and the idealized replications we consider.
\\

\noindent{\textbf{Effect size replication probability.}}
\label{sec:effect_repl_prob}
We define the probability that the replication absolute z-score will be larger than the original absolute z-score as
$$ \PP[G]{   |Z'| \geq |Z| \mid \abs{Z}=z}.$$
This estimand is closely related to Type M errors~\citep{gelman2014typem} and the general issue of exaggeration of significant point estimates~\citep{vanzwet2021significance}, i.e., a low effect size replication probability could indicate exaggeration in original estimates. 

As shown in Fig.~\ref{fig:effect_repl_prob_full}, the effect size replication probability diminishes rapidly as $\abs{z}$ increases and reaches a trough at absolute z-scores that just reached 0.05 statistical significance, it then increases slowly as $\abs{z}$ increases. Specifically, for studies with absolute z-scores just above 1.96, the probability that an exact replication produces a z-score of greater magnitude is about 25\% under $\mathcal{G}^{\mathrm{sN}}$. Such a low probability indicates that their original estimates are likely exaggerated. The~\citet{opensciencecollaboration2015estimating} estimates that only 17\% of the replicated effect size estimates are greater than the original studies.

\begin{table}
\centering
\caption{Confidence intervals for each estimand under different priors on MEDLINE (2000-2018) data.  
CIs for \(\omega_1\) and \(\omega_2\) are at the 97.5\% level;  
CI for \(\omega\) is at the 95\% level.}
\label{tab:omega_medline_full}
\begin{tabular}{lccccc}
\toprule
 & & \multicolumn{2}{c}{\(\omega_{2}\) (97.5\%)} & \multicolumn{2}{c}{\(\omega\) (95\%)} \\
\cmidrule(lr){3-4} \cmidrule(lr){5-6}
\textbf{Prior} & \(\omega_{1}\) (97.5\%) & \textbf{FLOC} & \textbf{AMARI} & \textbf{FLOC} & \textbf{AMARI} \\
\midrule
$\mathcal{G}^{\mathrm{sN}}$  & \multirow{3}{*}{(5.46, 5.58)} & (2.46, 3.35) & (2.44, 2.76) & (13.40, 18.67) & (13.26, 15.37) \\
$\mathcal{G}^{\mathrm{unm}}$ & & (2.02, 3.79) & (2.00, 2.77) & (11.02, 21.11) & (10.93, 15.43) \\
$\mathcal{G}^{\mathrm{all}}$ & & (1.26, 4.43) & (1.23, 3.14) & (6.86, 24.67) & (6.72, 17.51) \\
\bottomrule
\end{tabular}
\end{table}

\subsection{Risk ratio of publication of significant vs. nonsignificant results}
\label{sec:publication_prob}
Following~\citet{hedges1992modeling}, we consider the risk ratio of publication
of a significant vs.\ a non-significant result:
$$
\omega := \frac{\PP{D=1 \mid \abs{Z} \geq 1.96}}{\PP{D=1 \mid \abs{Z} < 1.96}} = \frac{\PP{\abs{Z} \geq 1.96 \mid D=1} \big  / \PP[G]{\abs{Z} \geq 1.96}}{\PP{\abs{Z} < 1.96 \mid D=1} \big / \PP[G]{\abs{Z} < 1.96}  }.
$$
We can break it up as $\omega =  \omega_1 \cdot \omega_2$, where $\smash{\omega_1=\PP{\abs{Z}\ge1.96\mid D=1}/\PP{\abs{Z}<1.96\mid D=1}}$ and $\smash{\omega_2=\PP[G]{\abs{Z}<1.96}/\PP[G]{\abs{Z}\ge1.96}}$.
The value of $\omega$ informs us on the extent of publication bias: $\omega > 1$ indicates that significant results are more likely to be published than non-significant ones, providing evidence of publication bias.

To obtain a 95\% confidence interval for $\omega$, we first derive two 97.5\% confidence intervals for $\omega_1$ and $\omega_2$, then the product yields the desired 95\% interval by Bonferroni adjustment. Specifically the interval for $\omega_2$ can be constructed by methods in Section~\ref{sec:methodology}, and we use the classical Wald's interval for $\omega_1$. We refer to Supplement~\ref{sec:omega_detail} for details.

The results are summarized in Table~\ref{tab:omega_medline_full}. Across all prior classes, the intervals for the risk ratio $\omega$ lie entirely above 1, providing evidence that statistically significant results have a higher probability of publication than non-significant findings. This result is not surprising and consistent with the visual evidence of heaping around $\abs{z} = 1.96$ observed in Fig.~\ref{fig:hist_med}.

\section{Further numerical results and analyses}
\label{sec:additional_results}
In addition to the analysis of the MEDLINE dataset, we conduct a variety of supplementary analyses addressing three practical questions. First, how sensitive are our conclusions to the sample size? Second, what is the price of robustness when we account for potential publication bias versus ignoring it? Third, how do our proposed confidence intervals compare to z-curve 2.0~\citep{bartos2022zcurve} for an estimand for which both methods are applicable, and how do our proposed methods perform compared to the bootstrap in our setting?

\subsection{Sensitivity to Sample Scope: Single-Year Analysis}
\label{subsec:single_year_analysis}
We replicate our procedure in the subset of MEDLINE studies published in 2018 (recall that in our analysis above, we pooled the years 2000-2018). There are two reasons why such an analysis is of interest. First, the publishing pattern could change over the years. Second, comparing the confidence intervals derived from this subset to those obtained from the full dataset (2000–2018) provides an empirical assessment of how interval width scales with sample size. 
We show and discuss the results in Supplement~\ref{subsec:2018_analysis} (Fig.~\ref{fig:medline_2018_CIs}, Tables~\ref{tab:power_above_80_medline2018} and~\ref{tab:omega_medline_2018}).
The resulting intervals provide no evidence of a different publishing pattern in 2018 relative to 2000-2018, but are substantially wider, illustrating the precision gains from aggregating more years of data.

\subsection{Cochrane robustness analysis}
\label{subsec:Cochrane_analysis}
In Supplement~\ref{subsec: supp_Cochrane_analysis} (Figs.~\ref{fig:Cochrane_CIs_1} and~\ref{fig:Cochrane_CIs_2}, Tables~\ref{tab:power_above_80_Cochrane_combined} and~\ref{tab:omega_cochrane}), we analyze the Cochrane Database of Systematic Reviews (CDSR), 
which is the leading database for systematic reviews in healthcare that contains results of more than 20,000 RCTs.
The CDSR is often perceived to exhibit comparatively small publication bias~\citep{zwet2021statistical, Schwabe2021Cochrane}. We first consider two settings: one employing our proposed selective tilting procedure to account for potential publication bias, and another without any truncation adjustment. This contrast quantifies the methodological cost of truncation, that is, the degree to which interval widths increase to ensure robustness against selective reporting, when applied to a corpus where such adjustments may be less necessary. Typically, the intervals constructed with truncation are much wider, partly because of the reduced sample size after truncation. We use the subset of the Cochrane database analyzed in~\citet{zwet2021statistical}, which contains 23,551 z-scores; after truncation, 6,119 samples remain. To isolate the effect of sample size reduction, we introduce a third setting where we analyze the data without truncation but on a random subset of size 6,119. This comparison shows that sample size reduction explains only part of the increased interval width; the remaining inflation is driven by extrapolation toward the truncated region, which is most apparent for the marginal density. Specifically, intervals constructed with truncation widens substantially as we move towards $z\approx 0$, whereas in the settings without truncation, the interval widths stay roughly constant across $z$. 
In addition, our intervals for $\omega$ based on Cochrane include 1 across all prior classes, with upper bounds ranging from 2 to 3, which is consistent with empirical findings in literature that the Cochrane database suffers from some, but limited publication bias. 

\subsection{Simulation study: z-curve 2.0 and the bootstrap} 
\label{subsec:bootstrap}

In Supplement~\ref{subsec:simulation_detail} (Fig.~\ref{fig:simulation_results}), we conduct simulation studies to assess frequentist coverage of our proposed methods and compare them with bootstrap-based methods, including z-curve, under both their setting and ours. The z-curve method,  as mentioned in Section~\ref{subsec:selective_inf_studies}, relies on assumptions similar to ours, but uses a different prior class $\mathcal{G}^{\text{z-curve}}$ as in~\eqref{eq:z_curve_prior} and the truncation set $\selection = [1.96, 6]$ (instead of our choice of $[2.1, \infty]$). It fits truncated folded z-scores using the expectation-maximization (EM) algorithm and constructs confidence intervals via the multinomial bootstrap~\citep{efron1979bootstrap}. Since $\mathcal{G}^{\text{z-curve}}$ is convex, we can apply our methods under the z-curve specifications. In addition, we consider a variant of z-curve that replaces the EM algorithm with the interior point (IP) method of~\citet{koenker2014convex} for computing the NPMLE (see Remark~\ref{rema:tilting_for_other_approaches}) while still using the bootstrap for inference. In our simulations we consider different estimands, and vary the specification to either match z-curve ($\mathcal{G}^{\text{z-curve}}$, $\selection=[1.96,6]$), or our modeling assumptions ($\mathcal{G}^{\mathrm{sN}}$, $\selection=[2.1, \infty)$). 

Across all settings, $F$-Localization provides near 100\% point-wise coverage rate due to its simultaneous coverage guarantee, while AMARI attains coverage at or above the nominal level with substantially narrower intervals. Coverage for AMARI can be larger than 95\% by design of the method in~\citet{ignatiadis2022confidence} to account for worst-case bias, which means that if bias is smaller than worst-case, the coverage can be larger than nominal. In contrast, bootstrap-based methods produce notably shorter intervals but exhibit substantial undercoverage, particularly outside the truncation region where extrapolation is needed.  Meanwhile, z-curve has lower coverage than the corresponding bootstrap method that uses IP instead of EM; the EM implementation of z-curve often does not converge to the maximizer of the marginal likelihood.

Although these simulations illustrate that bootstrap-based methods may fail even when the ground-truth prior is well-specified within the assumed convex class, we do not dismiss the bootstrap as a viable alternative in practice. Rather, our results underscore the absence of general theoretical guarantees for the bootstrap in nonparametric EB settings (see~\citet{ignatiadis2022rejoinder} for further discussion and~\citet{karlis2018confidence} for a rare result on the validity of the bootstrap in the Poisson EB problem) and demonstrate the potential of noticeable undercoverage. Caution is therefore warranted when relying on bootstrap-based uncertainty quantification in these problems.
\paragraph{Reproducibility.} All results in this paper are fully third-party reproducible with the code we provide under the following Github repository:\\ 
\url{https://github.com/huNterrchen/selective-eb-confidence-intervals-paper}

\paragraph{Acknowledgments.} We would like to thank Sowon Jeong and Jerry Zhu for helpful feedback on an earlier version of this manuscript. NI would like to thank Sifan Liu, Snigdha Panigrahi, and Asaf Weinstein for helpful conversations about empirical Bayes and selective inference. NI gratefully acknowledges support from the U.S. National Science Foundation (DMS-2443410). Part of the computing for this project was conducted on UChicago's Data Science Institute cluster.

\bibliographystyle{abbrvnat}
\bibliography{truncated_eb}

\appendix

\setcounter{equation}{0} 
\renewcommand{\theequation}{S\arabic{equation}}

\section{More details on preprocessing:  MEDLINE abstracts}
\label{sec:medline}

\citet{georgescu2018algorithmic} extract statistical measures representing ratios (e.g., odds ratios, hazard ratios) and their associated confidence intervals from MEDLINE abstracts. From the dataset constructed by~\citet{BarnetteMEDLINE2019} using the same extraction algorithm, we retain only abstracts from 2000-2018 that report a 95\% confidence interval. When multiple confidence intervals are present in a single abstract, we select one of them at random. For each reported ratio estimate $\hat{\theta}_i$ with a two-sided 95\% confidence interval $[L_i,U_i]$, we reconstruct the corresponding z-score as follows:  we compute the standard error as \smash{$\text{SE}_i = (\log(U_i) -\log(L_i))/(2\cdot1.96)$}, and the z-score as \smash{$Z_i = (\log(U_i) +\log(L_i))/(2\cdot\text{SE}_i)$}. After filtering, the final data set consists of $326,060$ z-scores.

\section{Discussion points from Jager and Leek (2014)}
\label{sec:jager_leek}
We summarize concerns raised by discussion articles of~\citet{jagerleek2014falsepositive} and examine the extent to which our methodology addresses them.

\begin{itemize}[noitemsep, leftmargin=*]
\item \citet{Benjamini2014discussion} draw a distinction between confidence intervals and p-values reported in abstracts. \citet{jagerleek2014falsepositive} work with p-values, while we only work with z-scores converted from reported confidence intervals (as described in Footnote~\ref{footnote:ci_to_z}). \citet{Ioannidis2014discussion} and \citet{Benjamini2014discussion} note that only extracting p-values from abstract both ignores vast amount of published studies that only reported confidence intervals, and results in selecting systematically lower p-values.
\item \citet{Gelman2014discussion} emphasize the shortcomings of defining false discoveries with respect to a point null hypothesis ($\mu_i=0$) and instead advocate for the importance of estimating Type S (sign) and Type M (magnitude) errors~\citep{gelman2014typem, gelman2000types}. By working in terms of z-scores we are able to go beyond the point null, and indeed, as discussed in Footnote~\ref{footnote:point_null}, most of our estimands do not depend at all on the existence of exact null effects. The sign-agreement probability estimand in Section~\ref{sec:prob_same_sign} relates to Type S errors. The effect size replication probability estimand (Section~\ref{sec:effect_repl_prob}) and the symmetrized posterior mean (Section~\ref{sec:symm_posterior_mean}) are closely connected to Type M errors.
\item \citet{Ioannidis2014discussion} questions the assumption in~\citet{jagerleek2014falsepositive} that the distribution of alternative p-values follows a Beta distribution. By considering three broad convex classes of priors in Assumption~\ref{assum:class_priors}, we are able to assess sensitivity to specific distributional assumptions on the prior.
\item \citet{Ioannidis2014discussion} and~\citet{Goodman2014discussion} note that the analysis of~\citet{jagerleek2014falsepositive} provides estimates for effects that are reported in abstracts and that such effects are not representative of the science-wise record. For this reason, it is important to interpret $G$ in~\eqref{eq:snr_prior} as pertaining specifically to effects that could have appeared in abstracts.
\item \citet{Goodman2014discussion} notes that p-values in abstracts can be highly correlated. We seek to avoid this pitfall by randomly choosing one p-value per abstract (see Supplement~\ref{sec:medline}). Moreover, \citet{Goodman2014discussion} notes that some p-values in abstracts are not for the main findings. This concern also applies to our analysis, since~\citet{BarnetteMEDLINE2019} do not extract information about which findings are primary. Future work could redo the extraction using, e.g., large language models, enabling subsetting to primary effects.
\item Throughout, we take the reported z-scores and standard errors at face value. In particular, we do not assess whether studies lack external validity, are confounded~\citep{schuemie2014discussion}, or ignore some sources of uncertainty~\citep{Cox2014discussion, Gelman2014discussion}. In this sense, our reported results should be interpreted as optimistic upper bounds on the replicability of results appearing in abstracts.
\end{itemize}
\beginsupptables
\section{Existing terminology for truncation mechanisms}
In table~\ref{tab:terminology_truncation_model}, We summarize some of the existing terminologies for the two truncation mechanisms we considered in Section~\ref{subsec:two_trunc_models}.
\begin{table}
\caption{Terminology used in the literature to refer to two alternative truncation mechanisms.}
\begin{tabular}{lll}
\toprule
Model~\eqref{eq:end_truncation} & Model~\eqref{eq:per_unit_truncation} & Reference\\ 
\toprule
\makecell[l]{Truncated mixture of\\untruncated densities} & 
\makecell[l]{Mixture of\\truncated densities} & 
\citet{bohning2006equivalence}\\
\midrule
\makecell[l]{Random-parameter\\truncated sampling model} & 
\makecell[l]{Fixed-parameter\\truncated sampling model} & 
\makecell[l]{\citet{yekutieli2012adjusted},\\\citet{rasines2022empirical}}\\ 
\midrule
Joint selection & Conditional selection & \citet{woody2022optimal}\\
\midrule
End truncation & Per-unit truncation & Here \\
\bottomrule
\end{tabular}
\label{tab:terminology_truncation_model}
\end{table}
\section{Further computational considerations}
\label{sec:comp_details}
\subsection{Discretization of prior classes:}
\label{sec:discr_G}
The prior classes we considered in Assumption~\ref{assum:class_priors} are infinite-dimensional, so our actual implementation approximates these by convex hulls of finite dictionaries. Specifically, for each class of random effects, the following finite dictionaries of distributions are used to allow tractable computation:
\begin{itemize}
    \item $\mathcal{G}^{\mathrm{sN}}$: We construct a finite mixture model where each component is a normal distribution with mean zero and variance $\sigma^2$, and $\sigma$ varies over a discretized grid of positive values. Specifically, we define a geometrically spaced grid $\left \{\sigma_i  \right \}_{i=1}^K$ where $\sigma_i = \sigma_\text{min} \cdot \gamma^{i-1}$ for $i=1,...,K$, with a constant factor $\gamma > 1$. The grid starts at a predetermined $\sigma_\text{min}$ and increases until it meets or exceeds a fixed $\sigma_\text{max}$. The number of grid points $K$ is the smallest integer such that $\sigma_\text{min} \cdot \gamma^{K-1} \geq \sigma_\text{max}$, i.e., $K=\lceil \log(\sigma_\text{max}/\sigma_\text{min})/\log(\gamma) \rceil$. In our setting, we set $\sigma_\text{min} = 0.001$, $\sigma_\text{max} = 100$ with $\gamma = 1.2$ to ensure wide coverage while keeping computational cost low. 
    \item $\mathcal{G}^{\mathrm{unm}}$: To approximate the class of all distributions with unimodal densities centered at zero, we rely on Khinchin's theorem, which states that any unimodal distribution is a scale mixture of symmetric uniform distributions. Each component of the scale mixture is a uniform distribution on $[-a, a]$, where $a$ is in a discretized grid of positive values. Following the same discretization approach as for $\mathcal{G}^{\mathrm{sN}}$, we generate a geometrically spaced grid $\left \{a_i  \right \}_{i=1}^K$ that starts from $a_\text{min}$, with each subsequent value scaled by a constant factor $\gamma > 1$ till it it reaches or surpasses $a_\text{max}$. The number of grid points $K$ is determined same way as for $\mathcal{G}^{\mathrm{sN}}$. We choose $a_\text{min} = 0.001$, $a_\text{max} = 100$ with $\gamma = 1.2$ so that our grid is sufficiently fine and wide for good approximation.
    \item $\mathcal{G}^{\mathrm{all}}$: The class of all distributions on $\mathbb R$ with a density can be approximated arbitrarily well by finite normal mixtures, which are known to be universal density approximators. Therefore, we create a normal location-scale mixture consisting of two types of components: (1) a location component, comprising normal distribution with means $\mu$ varying over a positive grid and a small fixed standard deviation $std$ ,and (2) a scale component, corresponding to the zero-mean normal scale mixture defined for $\mathcal{G}^{\mathrm{sN}}$. \\
    The location grid is defined as an equispaced set over [$\mu_\text{min}$, $\mu_\text{max}$] with spacing $\Delta\mu = std/4$, ensuring a sufficiently dense grid. We pick $\mu_\text{min}=0$, $\mu_\text{max}=12$, $std = 0.05$ for our setting. The scale component reuses the same $\left \{\sigma_i  \right \}_{i=1}^K$ grid as in $\mathcal{G}^{\mathrm{sN}}$.
\end{itemize}

\subsection{Details on the computational aspect of tilting operation}
\label{subsec:tilting_computation}
Following our computational approximation of $\mathcal{G}$ by convex hulls of finite dictionaries, we now extend the selective tilting equivalences of Section~\ref{sec:methodology} to this setting. This ensures that our inference procedure remains valid when implemented using convex hull representations. 
\subsubsection{Discretization of the tilted prior classes}
\label{sec:discr_tilted_G}
From above, we follow the computational strategy from~\citet{ignatiadis2022confidence} that fix a dictionary of priors $G_1,\dotsc,G_K$ and then take as the convex class the convex hull of the dictionary, that is,
$\mathcal{G} = \mathrm{ConvexHull}(G_1,\dots,G_K)$. The upshot is that any $G \in \mathcal{G}$ can be parameterized as $G = \sum_{j=1}^K \pi_j G_j$ where $(\pi_1,\dots,\pi_K)$ lie on the probability simplex ($\pi_j \geq 0$ and $\sum_{j=1}^K \pi_j = 1$), a constraint that can be handled by convex programming solvers.

Moreover, we can show that if we start with such a dictionary for $\mathcal{G}$, then we can lift all our computations on the dictionary $\Tilt{G_1},\dotsc,\Tilt{G_K}$ and in particular, we have the following.
\setcounter{prop}{0}
\renewcommand{\theprop}{S\arabic{prop}}
\begin{prop}
\label{prop:convex_hull_equivalence}
Suppose $\mathcal{G}= \mathrm{ConvexHull}(G_1,\dots,G_K)$, then
$$
\Tilt{\mathcal{G}} = \mathrm{ConvexHull}(\Tilt{G_1},\dotsc,\Tilt{G_K}). 
$$
\end{prop}

Some care is needed, however, because the bijection between $\Tilt{\mathcal{G}}$ and $\mathcal{G}$ is in fact not linear, that is, in general,
$$
\Tilt{\sum_{j=1}^K \pi_j G_{j}} \neq \sum_{j=1}^K \pi_j \Tilt{G_j}.
$$

\subsubsection{Establishing the mapping between \texorpdfstring{$\mathcal{G}$}{G} and \texorpdfstring{$\Tilt{\mathcal{G}}$}{Tilted G}}
\label{subsec: mapping_convex_hull}
Our confidence interval builds on the techniques of IW, which involve solving a convex optimization problem over a discretized convex class of priors. Specifically, Consider model~\eqref{eq:end_truncation} with prior $G \in \mathcal{G}$, we can leverage the observational equivalence in Theorem~\ref{theo:observational_equivalence} by working under model~\eqref{eq:per_unit_truncation} with priors $\Tilt{G} \in \Tilt{\mathcal{G}}$ to utilize the techniques. Based on Proposition~\ref{prop:functional_equivalence}, all estimands of interest in this paper can be reparametrized with some $\Tilt{G}$. By doing so, we lift all our computations on the tilted space $\Tilt{\mathcal{G}}$, where the optimizer returns a $\Tilt{G} = \sum_{i=1}^{K}\widetilde\pi_i \Tilt{G_i} \in \Tilt{\mathcal{G}}$. However, to acquire the confidence interval for each estimand with some prior $G = \sum_{i=1}^{K} \pi_iG_i$, we need to establish the mapping between $G$ and $\Tilt{G}$ so that we can properly apply the probabilities $\left \{\widetilde\pi_i  \right \}_{i=1}^K$ from the optimizer.
\paragraph{Mapping between $G = \sum_{i=1}^{K} \pi_iG_i$ and $\Tilt{G} = \sum_{i=1}^{K}\widetilde\pi_i \Tilt{G_i}$.} \mbox{}\\
By the definition of $\Tilt{G}$ in~\eqref{eq:tilt_G}:
\begin{align*}
\Tilt{G} &= \frac{\sum_{i=1}^{n} \Phi(\selection; \mu)\pi_iG_i }{ \sum_{j=1}^{n}\pi_j\int \Phi(\selection; \mu) G_j(d\mu)}=\sum_{i=1}^{n}\frac{\pi_i\PP[G_i]{\abs{Z} \in \selection}}{ \sum_{j=1}^{n}\pi_j\PP[G_j]{\abs{Z} \in \selection}} \frac{\Phi(\selection; \mu)G_i}{\PP[G_i]{\abs{Z} \in \selection}} \\
&=\sum_{i=1}^{n}\widetilde\pi_i \Tilt{G_i},
\end{align*}
where $\PP[G_i]{\abs{Z} \in \selection} =\int \Phi(\selection; \mu) G_i(d\mu)$ and $\widetilde\pi_i =\frac{\pi_i\PP[G_i]{\abs{Z} \in \selection}}{ \sum_{j=1}^{n}\pi_j\PP[G_j]{\abs{Z} \in \selection}}$.

\noindent So the final mapping between $\pi_i$ and $\widetilde\pi_i$ is as follows:
$$
\widetilde\pi_i =\frac{\pi_i\PP[G_i]{\abs{Z} \in \selection}}{ \sum_{j=1}^{n}\pi_j\PP[G_j]{\abs{Z} \in \selection}},\;\;\; 
\pi_i= \frac{\widetilde\pi_i/\PP[G_i]{\abs{Z} \in \selection}}{\sum_{j=1}^{n} \frac{\widetilde\pi_j}{\PP[G_j]{\abs{Z} \in \selection}}} .
$$

\noindent Based on this observation, we can rewrite every quantity of interest in the results sections in terms of both $\pi_i$ and $\widetilde\pi_i$ as follows:
\paragraph{Extended marginal density.}
\begin{align*}
\frac{ f_G(z)}{\PP[G]{\abs{Z} \in \selection} } &= \frac{\sum_{i=1}^{K} \pi_if_{G_i}(z)}{\sum_{j=1}^{K}\pi_j\PP[G_j]{\abs{Z} \in \selection}} \\
&=\sum_{i=1}^{K} \widetilde\pi_i\frac{f_{G_i}(z)}{\PP[G_i]{\abs{Z} \in \selection}}.
\end{align*}

\paragraph{Marginal density.}
\begin{align*}
f_G(z) &= \sum_{i=1}^{K} \pi_if_{G_i}(z)\\
&=  \frac{\sum_{i=1}^{K}\widetilde\pi_i\frac{f_{G_i}(z)}{\PP[G_i]{\abs{Z} \in \selection}}}{\sum_{j=1}^{K} \frac{\widetilde\pi_j}{\PP[G_j]{\abs{Z} \in \selection}}}. \\
\end{align*}

\paragraph{Power related quantities.}
Denote power function as $\beta(\mu) = 1- \Phi(1.96-\mu)+\Phi(-1.96-\mu)$ for a two-sided z-test, for any $B \subset [0.05,1]$, we consider the quantity $\PP[G]{\beta(\mu) \in B}$:
\begin{align*}
\PP[G]{\beta(\mu)  \in B} &= \sum_{i=1}^{K}\pi_i\PP[G_i]{\beta(\mu)  \in B} \\
&= \sum_{i=1}^{K}\frac{\widetilde\pi_i/\PP[G_i]{\abs{Z} \in \selection}}{\sum_{j=1}^{K} \frac{\widetilde\pi_j}{\PP[G_j]{\abs{Z} \in \selection}}}\PP[G_i]{\beta(\mu) \in B}.
\end{align*}

\paragraph{General Posterior quantity.}
Consider a general posterior functionals of $\mu$, denoted by $\theta_G(z) =\mathbb{E}[h(\mu) \mid \abs{Z}=z] $ in the following form:
$$
\theta_G(z) = \frac{\int h(\mu)\varphi^\text{fold}(z;\mu) G(\dd\mu) }{f_G(z)}.
$$
Then,
\begin{align*}
\theta_G(z) &=\frac{\sum_{i=1}^{K}\pi_i\int h(\mu)\varphi^\text{fold}(z;\mu) G_i\dd\mu}{\sum_{j=1}^{K} \pi_jf_{G_j}(z)}\\ 
&=  \sum_{i=1}^{K}\frac{\widetilde\pi_i\frac{f_{G_i}(z)}{\PP[G_i]{\abs{Z} \in \selection}}}{\sum_{j=1}^{K} \widetilde\pi_j\frac{f_{G_j}(z)}{\PP[G_j]{\abs{Z} \in \selection}}}\theta_{G_i}(z).\\
\end{align*}

\subsubsection{Exact equivalence of functionals on convex hulls}
\label{subsec:functional_equivalence_convex_hull}
We now extend the Proposition~\ref{prop:functional_equivalence} in Section~\ref{subsec:simple_tilting} to the scenario where we take $\mathcal{G}$ and thus $\Tilt{\mathcal{G}}$ as the convex hull. Let $T(\cdot) = N(\cdot)/D(\cdot): \mathcal{G} \to \mathbb R$ , where $N(\cdot)$ and $D(\cdot)$ are linear functionals, that is, $N(G)= \int \nu(\mu)G(\dd\mu)$ and $D(G) = \int \delta(\mu)G(\dd\mu)$ for some known $\nu(\cdot)$, $\delta(\cdot)$. By the definition of $\Tilt{T}$ in~\eqref{eq:functional_tilt}, we have the following.

\begin{prop}[Functional equivalence on convex hulls]
Suppose \\
$\mathcal{G}= \mathrm{ConvexHull}(G_1,\dots,G_K)$
and $\Tilt{\mathcal{G}} = \mathrm{ConvexHull}(\Tilt{G_1},\dotsc,\Tilt{G_K})$, then
\begin{equation*}
    \Tilt{T}(\Tilt{G}) = T(G) \;\text{ for all }\; G \in \mathcal{G}.
\end{equation*}
\label{prop:functional_equivalence_convexhull}
\end{prop}
\noindent Computationally, for linear and ratio functionals of $G$, evaluating them on the tilted space $\Tilt{\mathcal{G}}$ is quantitatively the same as the evaluating them with the corresponding $G$. 

\subsection{Computational details for \texorpdfstring{$F$-Localization}{F-Localization}}
\label{subsec:floc_comp}
Here we address how to compute these intervals in practice. Our overall approach is to build on the discretization and computation strategies of IW. We first discretize $\mathcal{S}$ as $s_1,\ldots,s_L$ and 
replace the supremum over $t \in \mathcal{S}$ in~\eqref{eq:KS-ball} by a maximum over $t \in \cb{s_1,\ldots,s_L}$.\footnote{This step is conservative and can only make our intervals longer. By default we take about $1,000$ grid points starting from the smallest up to the largest among truncated folded observations, filling in the interior grid points as sample quantiles at intermediate equispaced levels.} Second, we discretize
$\mathcal{G}$ as the convex hull of finite dictionary $\left \{G_1,...,G_K\right\}$. In Supplement~\ref{sec:discr_G} we provide this discretization for each of the classes in Assumption~\ref{assum:class_priors}. By Proposition~\ref{prop:convex_hull_equivalence} in Supplement~\ref{sec:discr_tilted_G}, we can write any $\tG \in \Tilt{\mathrm{ConvexHull}(G_1,\ldots,G_K)}$ as $\tG = \sum_{j=1}^{K} \tilde{\pi}_j \Tilt{G_j}$.

Suppose moreover that $T(\cdot)$ is the ratio functional $N(\cdot)/D(\cdot)$ with numerator $N(G) = \int \nu(\mu)G(\mu)$ and $D(G)=\int \delta(\mu)G(\dd\mu)$. With the discretization above, as explained in IW, we can solve the discretized version of the optimization problem in~\eqref{eq:DKW_F_loc_lower}
as a simple linear program using the~\citet{CharnesCooper1962} transformation. It takes the following form in our setting:
\[
\begin{aligned}
\underset{\zeta \geq 0,\{\tilde{\pi}_j\}_{j=1}^K \geq 0}{\text{minimize}}\quad
& \sum_{j=1}^K \tilde{\pi}_j \int \nu(\mu)\,\Phi(\selection;\mu)^{-1}\,\Tilt{G_j} (\dd\mu) \\[2pt]
\text{subject to}\quad
& \sum_{j=1}^K \tilde{\pi}_j = \zeta, \;\; \sum_{j=1}^K \tilde{\pi}_j \int\delta(\mu)\,\Phi(\selection;\mu)^{-1}\,\Tilt{G_j}(\dd\mu) = 1,\\
& \left\lvert\sum_{j=1}^K\tilde{\pi}_j\int_0^{s_l} f_{\Tilt{G_j}}^B(z)\dd z - \zeta\hat{F}_{n_\text{trun}}(s_{l})\right\rvert \leq \sqrt{\frac{\log(2/\alpha)}{2n_\text{trun}}}\; \text{for $l = 1 \ldots L$}.
\end{aligned}
\]

\subsection{Computational details for AMARI}
\label{subsec:amari_comp}
The basic idea of AMARI is to form confidence sets by test inversion. Specifically, suppose we seek to test $H_0: T(G)=c$ for some $c \in \RR$. Then the returned confidence interval consists of all $c$ that are not rejected. Now, for fixed $c$, using~\eqref{eq:functional_tilt}, Proposition~\ref{prop:functional_equivalence}, by rearranging (following~\citet{Fieller1940} and \citet{AndersonRubin1949}, $H_0$ is equivalent to the following null hypothesis
$$
H_0\,:\, \int \nu(\mu)\Phi(\selection; \mu)^{-1}\Tilt{G}(\dd\mu) - c \int \delta(\mu) \Phi(\selection; \mu)^{-1}\Tilt{G}(\dd\mu) = 0.
$$
The upshot is that we are testing whether a linear functional of $\Tilt{G}$ is equal to $0$. 
Then, IW build on classical ideas of~\citet{donoho1994statistical} on affine minimax estimators for linear functionals over convex spaces alongside ideas from bias-aware inference~\citep{Imbens2004bias, Armstrong2018,Imbens2019bias} to test the above hypotheses for all $c \in \RR$.

\subsection{Details on inference method for publication probability}
\label{sec:omega_detail}
Building on the definition of the risk ratio $\omega = \omega_1 \cdot \omega_2$ introduced in Section~\ref{sec:publication_prob}, we now detail the inference procedures. We can estimate $\omega_1$ by using all published absolute z-scores (before truncation to $\selection$),
$\{|Z_{i_1}|\;, \dotsc,\; |Z_{i_{n_{\text{published}}}}|\}$ (recall Fig.~\ref{fig:selectionProcess}) as follows:
$$ \widehat{\omega}_1 = \frac{ \#\cb{i \in \cb{i_1, \dotsc, i_{n_{\text{published}}}}\,:\, \abs{Z_i} \geq 1.96}}{\#\cb{i \in \cb{i_1, \dotsc, i_{n_{\text{published}}}}\,:\, \abs{Z_i} < 1.96}}.$$

Let us go back to our motivating example of the MEDLINE abstracts. We empirically estimate the numerator of $\omega_1$ using z-scores that exceed the 1.96 threshold. Denoting this proportion as $\widehat{p} = \frac{\#\cb{i \in \cb{i_1, \dotsc, i_{n_{\text{published}}}}\,:\, \abs{Z_i} \geq 1.96}}{n_\text{published}}$, we get $\widehat{\omega}_1 = \frac{\widehat{p}}{1-\widehat{p}} $. To construct a $97.5\%$ for $\PP{\abs{Z} \geq 1.96 \mid D=1}$ using the Wald's interval: $$\widehat{p}\pm z_{\alpha/2}\sqrt{\frac{\widehat{p}(1-\widehat{p})}{n}},\; \alpha = 0.025.$$ As $\omega_1$ is a strictly monotonic increasing function of $\PP{\abs{Z} \geq 1.96 \mid D=1}$, the corresponding $97.5\%$ confidence interval for $\omega_1$ is obtained by applying the same functional transformation to the confidence bounds of $\PP{\abs{Z} \geq 1.96 \mid D=1}$. 

We can conduct inference for $\omega_2$ as in Section~\ref{sec:methodology} as it is the ratio 
of two linear functionals of $G$. By leveraging the methodologies in Section~\ref{sec:methodology}, we can obtain a $97.5\%$ confidence interval for $\PP[G]{\abs{Z} < 1.96}$ under various prior classes. Since $\omega_2$ similarly represents a strictly monotonic transformation of this measure, its $97.5\%$ confidence interval is obtained by applying the same functional transformation to the confidence bounds of $\PP[G]{\abs{Z} < 1.96}$. Consequently, we compute the product of the lower bounds and the product of the upper bounds from the $97.5\%$ confidence intervals for $\omega_1$ and $\omega_2$. By the Bonferroni adjustment, this yields a  $95\%$ confidence interval for $\omega$.

The confidence interval for $\omega_1$ is substantially narrower than that for $\omega_2$. This is because $\omega_1$ is estimated directly from the full sample of published z-scores, thus its uncertainty is driven only by sampling variability. In contrast, $\omega_2$ depends on the unknown SNR distribution $G$, which we infer from selectively reported samples using the method described in Section~\ref{sec:methodology}. This inference problem is inherently more challenging, as it involves both deconvolution and extrapolation. As a result, most of the uncertainty in $\omega$ is driven by uncertainty in estimating $\omega_2$.

\section{Proofs}
For $\mu\sim G$, We denote the distribution of $-\mu$ by $G^-$. So we have $\Symm{G}(\dd\mu)=\frac{1}{2}G(\dd\mu)+\frac{1}{2}G^-(\dd\mu)$. 
\subsection{Proof of Proposition~\ref{prop:iff_cond_a1}}
\begin{proof}
($\implies$):
   Suppose there exists a constant $a \in (0,1]$ such that $\pi(z) = \mathbb{P}[D=1 \mid \abs{Z} =z] = a$, almost everywhere on $\selection$. For any measurable set $A \subset \selection$,
    \begin{align*}
        \mathbb{P}[\abs{Z} \in A \mid \abs{Z} \in \selection, D=1] &= \frac{\mathbb{P}[\abs{Z} \in A, \abs{Z} \in \selection, D=1]}{\mathbb{P}[\abs{Z} \in \selection, D=1]} \\
        &=\frac{\mathbb{P}[\abs{Z} \in A, D=1]}{\mathbb{P}[\abs{Z} \in \selection, D=1]} \\
        &=\frac{\mathbb{P}[\abs{Z} \in A]\mathbb{P}[D=1 \mid \abs{Z} \in A]}{\mathbb{P}[\abs{Z} \in \selection]\mathbb{P}[D=1 \mid \abs{Z} \in \selection]} \\
        &=\frac{\mathbb{P}[\abs{Z} \in A]a}{\mathbb{P}[\abs{Z} \in \selection]a} \\
        &=\frac{\mathbb{P}[\abs{Z} \in A]}{\mathbb{P}[\abs{Z} \in \selection]} \\
        &=\mathbb{P}[\abs{Z} \in A \mid \abs{Z} \in \selection].
    \end{align*}
The second equality holds since $\abs{Z} \in A$ implies $\abs{Z} \in \selection$. The fourth equality utilizes the fact that $ \mathbb{P}[D=1 \mid \abs{Z} \in A] = \mathbb{E}[\pi(z) \mid \abs{Z} \in A]  = \mathbb{E}[ a\mid \abs{Z} \in A] = a$. And similarly for $\mathbb{P}[D=1 \mid \abs{Z} \in \selection]$.

Since this equality holds for any $A \subset \selection$, we have $\cb{\abs{Z} \;\mid \;\p{\abs{Z} \in \selection},\;D=1}\;\;\;\; \stackrel{\mathcal{D}}{=}\;\;\;\; \cb{\abs{Z} \;\mid \; \p{\abs{Z} \in \selection}}$.

($\impliedby$): Suppose for any measurable set $A \subset \selection$, $\mathbb{P}[\abs{Z} \in A \mid \abs{Z} \in \selection, D=1] = \mathbb{P}[\abs{Z} \in A \mid \abs{Z} \in \selection]$. By decomposition above, this implies that $\mathbb{P}[D=1 \mid \abs{Z} \in A] = \mathbb{P}[D=1 \mid \abs{Z} \in \selection], \forall A \subset \selection$. Denote $\mathbb{P}[D=1 \mid \abs{Z} \in \selection] = c$, for some $c \in (0,1]$:
$$
\mathbb{P}[D=1 \mid \abs{Z} \in A] = \frac{\int_{A}\pi(z)f_{G}(z)\dd z}{\int_{A}f_{G}(z)\dd z} = c,
$$
which means $\int_{A}(\pi(z)-c)f_{G}(z)\dd z = 0, \forall A \subset \selection$. Hence $\cb{z \in \selection: \pi(z)\neq c}$ must be a measure-0 set, so $\pi(z)= c$ a.e. on $\selection$.    
\end{proof}

\subsection{Proof of Theorem~\ref{theo:Identifiability}}
\begin{proof}
By Assumption~\ref{assum:publication_bias} and Proposition~\ref{prop:iff_cond_a1},
    the density of $\cb{\abs{Z} \;\mid \;\p{\abs{Z} \in \selection},\;D=1}$ is the same as that of $\cb{\abs{Z} \;\mid \;\p{\abs{Z} \in \selection}}$, and is given by
    \begin{align*}
    f_{\mathcal{S}, G} (z) &:= \frac{\pi(z)  f_G(z) }{\int_\mathcal{S} \pi(z)  f_G(z) \dd z } \stackrel{(*)}{=}\frac{1}{C_{\selection, G}}\int_{\mu >0}\varphi^{\text{fold}}(z;\mu)\Fold{G}(\dd\mu) \qquad\text{for} \,\,z \in \mathcal{S}.  
    \end{align*}
    where ($*$) follows from Lemma~\ref{lemm:fold-marginal}, and $$C_{\selection, G}=\int_\mathcal{S} \int \varphi^{\text{fold}}(z;\mu)\Fold{G}(\dd\mu)\dd z$$ is the normalizing constant.
    Now suppose there exists $H$ such that
    \[
    {f}_{\mathcal{S},H}(z) =\frac{1}{C_{\selection, H}} \int_{\mu >0}\varphi^{\text{fold}}(z;\mu)\Fold{H}(\dd\mu) = f_{\mathcal{S}, G}(z) \qquad \forall \,z\in \mathcal{S}.
    \]
    We want to show that this implies $\Fold{G} = \Fold{H}$, hence the map: $G\mapsto f_{\selection, G}$ is injective and $\Fold{G}$ is identified by the observable $\cb{\abs{Z} \;\mid \;\p{\abs{Z} \in \selection},\;D=1}$. We proceed as follows: (1) first we consider the analytic continuation of $f_{\selection, G}$ and $f_{\selection, H}$ to $\RR$, and show that they also agree on $\RR$; (2) we then show that the agreement of the Fourier transforms of the analytic continuation, $f^*_{\mathcal{S},H}={f}^*_{\mathcal{S},G}$, implies $\Symm{\Fold{H}}=\Symm{\Fold{G}}$; (3) since $\Symm{F} = \frac{1}{2}(F+F^-)$ for any $F\in \mathcal{P}(\RR_{\geq 0})$, it then follows that $\Fold{H} = {\Fold{G}}$ as desired.   
    
    (1): we first note that $\varphi^{\text{fold}}$ (and hence $f_{\mathcal{S}, G}$ and $f_{\mathcal{S}, H}$) is analytic on $\RR$. Since $\mathcal{S}$ has nonzero Lebesgue measure by Assumption~\ref{assum:publication_bias}, $\mathcal{S}$ is uncountable, and there exists an accumulation point of $\RR$ in $\mathcal{S}$. Indeed, write $\RR=\cup_{i=1}^\infty A_i$, where $A_i$'s are disjoint left half-open bounded intervals, so $\mathcal{S}= \mathcal{S}\,\cap\, \cup_{i=1}^\infty A_i = \cup_{i=1}^\infty S_i$, where $S_i=\mathcal{S}\, \cap\, A_i$ is bounded. Then at least one of $S_i$'s is uncountable (otherwise $\mathcal{S}$ is countable). Consider an arbitrary infinite sequence in $S_i$. By Bolzano–Weierstrass theorem, there exists a convergent subsequence, and hence an accumulation point exists in $S_i \subset \mathcal{S}$. This allows us to conclude $f_{\mathcal{S},G}(z)=f_{\mathcal{S},H}(z)$ for all $z\in\RR$ by invoking the identity theorem. 
    
    (2): The previous step justifies computing Fourier transforms on all of $\RR$:
    \begin{equation*}
    \begin{aligned}
        f^*_{\mathcal{S}, G}(t)&=\int_\RR f^*_{\mathcal{S},G}(z)\dd z
        =\int_\RR e^{izt}   \int_{\mu >0} \frac{1}{C_{\selection, G}}\varphi^{\text{fold}}(z;\mu) \Fold{G}(\dd\mu) \dd z \\
        &= \frac{1}{C_{\selection, G}}\int_{\mu > 0} \Bigg[\int_\RR\varphi^{\text{fold}}(z;\mu) e^{izt}\dd z \Bigg]\Fold{G}(\dd\mu) \\
        &= \frac{1}{C_{\selection, G}}\int_{\mu > 0} \Bigg[\int_\RR\ \Big(\varphi(z;\mu)+\varphi(-z;\mu)\Big) e^{izt}\dd z \Bigg]\Fold{G}(\dd\mu)\\
        &= \frac{e^{-t^2/2}}{C_{\selection, G}}\int_{\mu > 0} \Big[ e^{i\mu t} + e^{-i\mu t}  \Big]\Fold{G}(\dd\mu)
        = \frac{2e^{-t^2/2}}{C_{\selection, G}} \int_{\mu>0} \cos(\mu t) \Fold{G}(\dd\mu) \\
        &= \frac{e^{-t^2/2}}{C_{\selection, G}} \Bigg[\int_{\mu>0} \cos(\mu t) \Fold{G}(\dd\mu) + \int_{\mu\leq 0} \cos(\mu t) \Fold{G}^-(\dd\mu)\Bigg] \\
        &= \frac{e^{-t^2/2}}{C_{\selection, G}} \Bigg[\int_{\RR} \cos(\mu t) \Symm{\Fold{G}}(\dd\mu) + i\int_\RR \sin(\mu t) \Symm{\Fold{G}}(\dd\mu) \Bigg] \\
        &= \frac{e^{-t^2/2}}{C_{\selection, G}} \int_{\RR} e^{i\mu t} \Symm{\Fold{G}}(\dd\mu) =: e^{-t^2/2} \cdot {\Symm{\Fold{G}/C_{\selection, G}}}^*(t). 
    \end{aligned}
    \end{equation*}
    Similarly,
    \begin{align*}
        {f}^*_{\mathcal{S},H}(t)&=\frac{2e^{-t^2/2}} {C_{\selection, H}} \int_{\mu>0} \cos(\mu t) \Fold{H}(\dd\mu) \\
    &=\frac{e^{-t^2/2}}{C_{\selection, H}} \int_{\RR} e^{i\mu t}\Symm{ \Fold{H}}(\dd\mu)=:e^{-t^2/2} \cdot{\Symm{\Fold{H}/C_{\selection, H}}}^*(t). 
    \end{align*}
    So
    \[
    {\Symm{\Fold{G}/C_{\selection, G}}}^*(t) = {\Symm{\Fold{H}/C_{\selection, H}}}^*(t) \quad \forall \, t \in \RR.
    \]
    It then follows that
    \[
    \frac{1}{C_{\selection, G}}\Symm{\Fold{G}} = \frac{1}{C_{\selection, H}} \Symm{\Fold{H}}.
    \]
    The fact the $\int_\RR \Symm{\Fold{G}} (\dd\mu) = \int_\RR \Symm{\Fold{H}} (\dd\mu)= 1$ implies $C_{\selection, G}=C_{\selection, H}$. So $\Symm{\Fold{G}}=\Symm{\Fold{H}}$.
    
\end{proof}

\subsection{Proof of Theorem~\ref{theo:observational_equivalence}}
\begin{proof}
Suppose $z \in \selection$, for otherwise, $f_G^A(z) = f_G^B(z) = 0$. Then:
$$ 
\begin{aligned}
f_{\Tilt{G}}^B(z) &=  \int  \frac{\varphi^{\text{fold}}(z;\mu)}{\Phi({\selection;\mu})}\Tilt{G}(\dd\mu) \\
& =  \frac{\int \varphi^{\text{fold}}(z;\mu)G(\dd\mu)}{\int \Phi({\selection;\mu})G(\dd\mu)} \\
& = f_G^A(z).
\end{aligned}
$$
\end{proof}

\subsection{Proof of Remark~\ref{rema:untilt}}
\begin{proof}
Suppose $z \in \selection$, then:
$$ 
\begin{aligned}
f_{\Untilt{\tG}}^A(z) &=  \frac{\int  \varphi^{\text{fold}}(z;\mu) \Untilt{\tG}(\dd\mu)}{\int_\selection\int \varphi^{\text{fold}}(z;\mu) \Untilt{\tG}(\dd\mu)\dd z} \\
&= \frac{\int  \varphi^{\text{fold}}(z;\mu) \Phi(\selection; \mu)^{-1}\tG(\dd\mu)/\int \Phi(\selection; \mu)^{-1}\tG(\dd\mu)}{\int_\selection\int \varphi^{\text{fold}}(z;\mu) \Phi(\selection; \mu)^{-1}\tG(\dd\mu)\dd z/\int \Phi(\selection; \mu)^{-1}\tG(\dd\mu)} \\
&= \frac{\int  \varphi^{\text{fold}}(z;\mu) \Phi(\selection; \mu)^{-1}\tG(\dd\mu)}{\int_\selection\int \varphi^{\text{fold}}(z;\mu) \Phi(\selection; \mu)^{-1}\tG(\dd\mu)\dd z} \\
&= \frac{\int  \varphi^{\text{fold}}(z;\mu) \Phi(\selection; \mu)^{-1}\tG(\dd\mu)}{\int \Phi(\selection; \mu)^{-1} \int_\selection \varphi^{\text{fold}}(z;\mu) \dd z\tG(\dd\mu)} \\
&= \frac{\int  \varphi^{\text{fold}}(z;\mu) \Phi(\selection; \mu)^{-1}\tG(\dd\mu)}{\int \Phi(\selection; \mu)^{-1} \Phi(\selection; \mu)\tG(\dd\mu)} \\
&= \frac{\int  \varphi^{\text{fold}}(z;\mu) \Phi(\selection; \mu)^{-1}\tG(\dd\mu)}{\int \tG(\dd\mu)} \\
&= \int \frac{\varphi^{\text{fold}}(z;\mu)}{\Phi(\selection; \mu)}\tG(\dd\mu) \\
&= f_{\tG}^B(z).
\end{aligned}
$$
The fourth equality applies Fubini's theorem, and the fifth equality is by the definition of $\Phi(\selection; \mu)$.
\end{proof}

\subsection{Proof of Proposition~\ref{prop:convex_equivalence}}
\begin{proof}
Take $H_1, H_2 \in \Tilt{\mathcal{G}}$ and $\lambda \in (0,1)$. We need to show that $\lambda H_1 + (1-\lambda)H_2 \in \Tilt{\mathcal{G}}$. To do so, it suffices to find $G \in \mathcal{G}$ such that $\lambda H_1 + (1-\lambda)H_2  = \Tilt{G}$.
Let $H_1 = \Tilt{G_1}, H_2 = \Tilt{G_2}$ for some $G_1,G_2 \in \mathcal{G}$. Then:
$$
\begin{aligned}
    \lambda H_1 + (1-\lambda)H_2&= \frac{\lambda\Phi(\selection; \mu) G_1}{ \int \Phi(\selection; \mu) G_1(\dd\mu)} + \frac{(1-\lambda)\Phi(\selection; \mu) G_2}{ \int \Phi(\selection; \mu) G_2(\dd\mu)} \\
    &= \Phi(\selection; \mu)\left(\frac{\lambda G_1}{ \int \Phi(\selection; \mu) G_1(\dd\mu)} + \frac{(1-\lambda) G_2}{ \int \Phi(\selection; \mu) G_2(\dd\mu)}\right)    
\end{aligned}
$$
Consider $G = \frac{1}{\beta}\left(\frac{\lambda G_1}{ \int \Phi(\selection; \mu) G_1(\dd\mu)} + \frac{(1-\lambda) G_2}{ \int \Phi(\selection; \mu) G_2(\dd\mu)}\right)$, where $\beta = \frac{\lambda}{\int \Phi(\selection; \mu) G_1(\dd\mu)}+\frac{1-\lambda}{\int \Phi(\selection; \mu) G_2(\dd\mu)}$. Since $G = \theta G_1+(1-\theta) G_2$  and $\theta =\frac{1}{\beta}\frac{\lambda}{\int \Phi(\selection; \mu) G_1(\dd\mu)} \in (0, 1)$, we have $G \in \mathcal{G}$ under the assumption that $\mathcal{G}$ is a convex class. Observe that: 
$$
\int \Phi(\selection; \mu) G(\dd\mu) = \frac{1}{\beta}\int \Phi(\selection; \mu)\left(\frac{\lambda G_1}{ \int \Phi(\selection; \mu) G_1(\dd\mu)} + \frac{(1-\lambda) G_2}{ \int \Phi(\selection; \mu) G_2(\dd\mu)}\right)(\dd\mu) = \frac{1}{\beta}.
$$
Hence,
$$
\begin{aligned}
    \Tilt{G}&= \frac{\Phi(\selection; \mu) G}{ \int \Phi(\selection; \mu) G(\dd\mu)} \\
    &= \Phi(\selection; \mu)\frac{1}{\beta}\left(\frac{\lambda G_1}{ \int \Phi(\selection; \mu) G_1(\dd\mu)} + \frac{(1-\lambda) G_2}{ \int \Phi(\selection; \mu) G_2(\dd\mu)}\right)\beta   \\
    &=\lambda H_1 + (1-\lambda)H_2.
\end{aligned}
$$
So our choice of $G$ satisfies all the requirements.
\end{proof}

\subsection{Proof of Proposition~\ref{prop:functional_equivalence}}
\begin{proof}
    Simply plug in the definition of $\Tilt{G}$:
    \begin{align*}
       \Tilt{T}(\Tilt{G}) &= \frac{ \int \nu(\mu)\Phi(\selection; \mu)^{-1}\Tilt{G}(\dd\mu)}{\int \delta(\mu) \Phi(\selection; \mu)^{-1}\Tilt{G}(\dd\mu)}  \\
        &= \frac{ \int \nu(\mu)G(\dd\mu)/\int \Phi(\selection; \mu)G(\dd\mu)}{\int \delta(\mu)G(\dd\mu)/\int \Phi(\selection; \mu)G(\dd\mu)} \\
        &=\frac{\int \nu(\mu)G(\dd\mu)}{\int \delta(\mu)G(\dd\mu)} = T(G).
    \end{align*}
\end{proof}

\subsection{Proof of Theorem~\ref{theo:ci_reduction}}
\begin{proof}
    By Theorem~\ref{theo:observational_equivalence}, we have that $f_G^A(\cdot) = f_{\Tilt{G}}^B(\cdot)$. Since $\abs{Z_1},\ldots, \abs{Z_n}$ are i.i.d, we have that the joint distribution of $(\abs{Z_1},\ldots, \abs{Z_n})$ is the same under both models:
    $$
    \mathbb{P}_G^A[\abs{Z_1},\ldots, \abs{Z_n}] =  \mathbb{P}_{\Tilt{G}}^B[\abs{Z_1},\ldots, \abs{Z_n}].
    $$
    From Proposition~\ref{prop:functional_equivalence}, since $\Tilt{T}(\Tilt{G}) = T(G)$, we have the following equivalence of events:
    $$
    \left\{T(G) \in \mathcal{I}(\abs{Z_1},\ldots, \abs{Z_n}) \right \} = \left\{\Tilt{T}(\Tilt{G}) \in \mathcal{I}(\abs{Z_1},\ldots, \abs{Z_n}) \right \}.
    $$
    Therefore 
    $$\mathbb P_G^A[ T(G) \in \mathcal{I}(\abs{Z_1},\ldots, \abs{Z_n})] = \mathbb P_{\Tilt{G}}^B[ \Tilt{T}(\Tilt{G}) \in \mathcal{I}(\abs{Z_1},\ldots, \abs{Z_n})].$$
\end{proof}

\subsection{Proof of Proposition~\ref{prop:IdentifiabilityAll}}
To prove Proposition~\ref{prop:IdentifiabilityAll}, we check each estimand separately:
\subsubsection{Marginal density}
For the marginal density $f_G(z)$, we have the following lemma:
\begin{lemm}[Representation via $\Fold{G}$]
Under \eqref{eq:publication_bias_model}, the marginal density of $\abs{Z}$ is given by:
$$f_G(z)=\int_{u >0}(\varphi(z;-u) + \varphi(z;u))\Fold{G}(\dd u).$$
So $f_G$ depends only on $\Fold{G}$ the distribution of $|\mu|$.
\label{lemm:fold-marginal}
\end{lemm}
\begin{proof}
    Denote $f_G(z)$ the marginal density of $\abs{Z}$ under \eqref{eq:publication_bias_model}, and $G^{\text{fold}}$ the distribution of $\abs{\mu}$ induced by $G$. Let $\varphi(z;\mu)$ be the normal density with mean $\mu$ and variance $1$ evaluated at $z$. We have that for $z\geq0$:
\begin{align*}
f_G(z) &= \int(\varphi(z;\mu) + \varphi(-z;\mu))G(\dd\mu) \\
&= \int_{\mu \geq0}(\varphi(z;\mu) + \varphi(-z;\mu))G(\dd\mu) + \int_{\mu <0}(\varphi(z;\mu) + \varphi(-z;\mu))G(\dd\mu) \\
&\text{Substitute u = $-\mu$ for the second integral} \\
&=  \int_{\mu \geq0}(\varphi(z;\mu) + \varphi(-z;\mu))G(\dd\mu) + \int_{u >0}(\varphi(z;-u) + \varphi(-z;-u))G(\dd(-u)) \\
&= \int_{\mu \geq0}(\varphi(z;\mu) + \varphi(z;-\mu))G(\dd\mu) + \int_{u >0}(\varphi(z;-u) + \varphi(z;u))G(\dd(-u)) \\
&=\int_{u >0}(\varphi(z;-u) + \varphi(z;u))G^{\text{fold}}(\dd u).
\end{align*}
The fourth equality applies the fact that $\varphi(z;-\mu) = \varphi(-z;\mu)$.
Hence, the marginal density of $\abs{Z}$ is a function of the distribution of $\abs{\mu}$ only.
\end{proof}
Based on this result, let us consider the normalized marginal density $f_G(z)/\PP[G]{\abs{Z} \in \selection}$. Since $\PP[G]{\abs{Z} \in \selection} = \int_\selection f_G(z)\dd z$, it follows that $\PP[G]{\abs{Z} \in \selection}$ is a functional of $G^\text{fold}$, so does the normalized marginal density following Lemma~\ref{lemm:fold-marginal}.

\subsubsection{Power-based estimands}
The power function exhibits symmetry: $\beta(\mu) = \beta(-\mu)$, so $\beta(\mu)$ is a function of $\abs{\mu}$ only. Consequently, for any $B \subset [0.05,1]$, $\PP[G]{\beta(\mu) \in B}$ depends only on the distribution of $\abs{\mu}$.

This property implies that the binned power density, the Cumulative distribution of power, and the proportion of sufficiently powered studies we introduced earlier must also be functions of $\abs{\mu}$.

\subsubsection{Probability of same sign}
\label{sec:same_sign_proof}
Recall that following from model~\eqref{eq:publication_bias_model}, we can write each z-score as $Z=\mu+\epsilon$, where $\epsilon\sim \mathrm{N}(0,1)$ and $\mu \indep \epsilon$.
To show that $\PP[G]{\mu \cdot Z > 0 \mid \abs{Z}=z}$ depends only on the distribution of $\abs{\mu}$, note that it is determined by the joint distribution of $(\mu \cdot Z, |Z|)$. Thus, it suffices to establish that this joint distribution depends solely on the distribution of $\abs{\mu}$. To see this, we start by defining $\sign(\cdot ):= \ind\{\cdot\geq 0\} - \ind\{\cdot < 0\}$. By considering the different cases of the sign of $\mu$, we have $-\sign(\mu)\cdot \epsilon \sim \mathrm{N}(0,1)$,
so $-\sign(\mu)\cdot \epsilon \stackrel{\mathcal{D}}{=} \epsilon \indep \mu$. It then follows that
\begin{align*}
    (\mu\cdot Z, |Z|) &= (\mu^2+\mu\cdot \epsilon, |\mu+\epsilon|) \\
    &\stackrel{\mathcal{D}}{=} (|\mu|^2-\mu\cdot \sign(\mu)\cdot \epsilon, |\mu-\sign(\mu)\cdot \epsilon|) \\
    &=(|\mu|^2-|\mu|\cdot \epsilon, ||\mu|-\epsilon|).
\end{align*}
To justify the last equality: if $\mu \geq0$, then $|\mu-\sign(\mu)\cdot \epsilon| = |\mu - \epsilon| = ||\mu|-\epsilon|$; otherwise, $|\mu-\sign(\mu)\cdot \epsilon| = |\mu + \epsilon| = |-\mu-\epsilon| = ||\mu|-\epsilon|$. Since the distribution of $\epsilon$ is known, the joint distribution only depends on that of $\abs{\mu}$. 

Our computation of this estimand relies on the following result. Denote $f_G^Z(z)$ the marginal density of signed $Z$:

\begin{prop}[Sign–agreement probability decomposition]
\label{prop:sign-agreement}
For any $z > 0$ with $f_G^{Z}(z) + f_G^{Z}(-z)>0$, 
\begin{align*}
    \PP[G]{\mu \cdot Z > 0 \mid \abs{Z}=z}  = \frac{\PP[G]{\mu >0\mid Z = z}f_G^{Z}(z) + \PP[G]{\mu <0\mid Z = -z}f_G^{Z}(-z)}{f_G^{Z}(z) + f_G^{Z}(-z)}.
\end{align*}
\end{prop}
\begin{proof}
Given $|Z|=z$, the event \(\{\mu \cdot Z>0\}\) is a disjoint union of
\(\{Z=z,\;\mu>0\}\) and \(\{Z=-z,\;\mu<0\}\).
By the law of total probability,
\begin{align*}
\PP[G]{\mu \cdot Z > 0 \mid \abs{Z}=z} &= \PP[G]{\mu >0,  Z =z \mid \abs{Z}=z} +  \PP[G]{\mu <0,  Z =-z \mid \abs{Z}=z} \\
& = \PP[G]{\mu >0\mid Z = z, \abs{Z}=z}\PP[G]{Z =z \mid \abs{Z}=z} \\
&+ \PP[G]{\mu <0\mid Z = -z, \abs{Z}=z}\PP[G]{Z =-z \mid \abs{Z}=z} \\
&=  \PP[G]{\mu >0\mid Z = z}\frac{f_G^{Z}(z)}{f_G^{Z}(z) + f_G^{Z}(-z)} \\
&+ \PP[G]{\mu <0\mid Z = -z}\frac{f_G^{Z}( -z)}{f_G^{Z}(z) + f_G^{Z}(-z)} \\
&= \frac{\PP[G]{\mu >0\mid Z = z}f_G^{Z}( z) + \PP[G]{\mu <0\mid Z = -z}f_G^{Z}(-z)}{f_G^{Z}(z) + f_G^{Z}(-z)}.
\end{align*}
The second equality applied the chain rule, and the third equality holds since $Z = z$ implies $|Z|=z$, so $\PP[G]{\mu >0\mid Z = z, \abs{Z}=z} = \PP[G]{\mu >0\mid Z = z}$, and likewise for $-z$.

\end{proof}
At $z = 0$ we define the sign-agreement probability by convention:
we define by the continuous extension so that it takes the value of
$1/2 (\PP[G]{\mu > 0 \mid Z=0} + \PP[G]{\mu < 0 \mid Z=0}).$

\subsubsection{Replication probability}
Consider an exact replication study that has the same underlying parameter of interest $\mu$ as the original study and the scientific procedure. Denote the z-score from such a replication study by $Z'$, which can be written as $Z' = \mu + \epsilon'$, where $\epsilon' \sim \mathrm{N}(0,1)$ is independent of $\epsilon$.
This estimand $\PP[G]{ \abs{Z'}>1.96, ZZ'>0 \cond \abs{Z}=z}$ is a function of the joint distribution of $(ZZ', |Z'|, |Z|)$. To see that it depends only on the distribution of $|\mu|$, we use similar argument as in previous sections and conclude that
\begin{align*}
    (ZZ', |Z|, |Z'|) &= (\mu^2 +\epsilon\mu+\epsilon'\mu+\epsilon\epsilon', |\mu+\epsilon'|, |\mu+\epsilon|) \\
    &\stackrel{\mathcal{D}}{=} (|\mu|^2 +\epsilon \cdot  \sign(\mu) \cdot |\mu|+\epsilon'\cdot \sign(\mu) \cdot |\mu|+\epsilon\epsilon', ||\mu|-\epsilon'|, ||\mu|-\epsilon|) \\
    &\stackrel{\mathcal{D}}{=} (|\mu|^2 +\epsilon |\mu|+\epsilon'|\mu|+\epsilon\epsilon', ||\mu|-\epsilon'|, ||\mu|-\epsilon|).
\end{align*}

To facilitate computation, we introduced the following result:
\begin{prop}[Replication probability decomposition]
\label{prop:rep-prob}
For any $z> 0$ with with $f_G^{Z}(z) + f_G^{Z}(-z)>0$,
\begin{align*}
    &\PP[G]{ \abs{Z'}>1.96, ZZ'>0 \cond \abs{Z}=z}  \\
    &\qquad =\frac{\int (1-\Phi(1.96-\mu ))\varphi(z;\mu)G(\dd\mu) + \int \Phi(-1.96-\mu )\varphi(-z;\mu)G(\dd\mu)}{f_G^Z(z) + f_G^Z(-z)}.
\end{align*}
\end{prop}

\begin{proof}
Given $|Z|=z$, the event \(\{\abs{Z'}>1.96, ZZ'>0\}\) is a disjoint union of
\(\{Z=z,\;Z'>1.96\}\) and \(\{Z=-z,\;Z'<-1.96\}\).
By the law of total probability,
    \begin{align*}
    &\PP[G]{ \abs{Z'}>1.96, ZZ'>0 \cond \abs{Z}=z}  \\
    &\qquad = \PP[G]{Z=z,\;Z'>1.96\cond \abs{Z}=z}+ \PP[G]{Z=-z,\;Z'<-1.96\cond \abs{Z}=z} \\
    &\qquad =  \PP[G]{Z'>1.96\cond Z=z, \abs{Z}=z}\PP[G]{Z =z \mid \abs{Z}=z} \\
    &\qquad + \PP[G]{Z'<-1.96\cond Z=-z, \abs{Z}=z}\PP[G]{Z =-z \mid \abs{Z}=z} \\
    &\qquad = \frac{\PP[G]{Z'>1.96\cond Z=z}f_G^{Z}( z) + \PP[G]{Z'<-1.96\cond Z=-z}f_G^{Z}(-z)}{f_G^{Z}(z) + f_G^{Z}(-z)} \\
    &\qquad =\frac{\frac{\PP[G]{Z'>1.96,\; Z=z}}{f_G^{Z}( z)}f_G^{Z}( z) + \frac{\PP[G]{Z'<-1.96,\; Z=-z}}{f_G^{Z}(- z)}f_G^{Z}(-z)}{f_G^{Z}(z) + f_G^{Z}(-z)} \\
    &\qquad = \frac{\int\PP[G]{Z'>1.96\cond \mu}f(z \cond \mu)G(\dd\mu)+ \int \PP[G]{Z'<-1.96\cond \mu}f(-z \cond \mu)G(\dd\mu)}{f_G^{Z}(z) + f_G^{Z}(-z)} \\
    &\qquad =\frac{\int (1-\Phi(1.96-\mu ))\varphi(z;\mu)G(\dd\mu) + \int \Phi(-1.96-\mu )\varphi(-z;\mu)G(\dd\mu)}{f_G^{Z}(z) + f_G^{Z}(-z)}.
\end{align*}
The second equality applies the chain rule, and the third equality holds since $Z = z$ implies $|Z|=z$, so $\PP[G]{Z'>1.96\cond Z=z, \abs{Z}=z} = \PP[G]{Z'>1.96\mid Z = z}$, and likewise for $-z$. The fifth equality uses the fact that $Z$ is independent of $Z'$ conditional on $\mu$. Last equality holds as $Z \cond \mu \sim \mathrm{N}(\mu, 1)$ and $Z' \cond \mu \sim \mathrm{N}(\mu, 1)$.
\end{proof}

At $z = 0$ we define the replication probability by convention: we define by the continuous extension so that it takes the value of
$1/2 (\PP[G]{Z' > 1.96 \mid Z=0} + \PP[G]{Z' < -1.96 \mid Z=0})$.
\subsubsection{Future coverage probability}
To see that $\PP[G]{ Z \in Z' \pm 1.96       \mid \abs{Z}=z}$ only depends on $\Fold{G}$, observe that \(\{ Z \in Z' \pm 1.96, \abs{Z}=z\}\) = \(\{\abs{Z-Z'}\leq 1.96, \abs{Z}=z\}\). So the estimand depends on the joint distribution of $(\abs{Z-Z'}, \abs{Z})$, and by similar argument in Section~\ref{sec:same_sign_proof}:
\begin{align*}
    (\abs{Z-Z'}, \abs{Z}) &= (\abs{\mu+\epsilon -\mu-\epsilon'}, |\mu+\epsilon|) \\
    &= (|\epsilon -\epsilon'|, |\mu+\epsilon|) \\
    &\stackrel{\mathcal{D}}{=}(|\epsilon -\epsilon'|, ||\mu|-\epsilon|).
\end{align*}
Since the distribution of $\epsilon$ and $\epsilon'$ are known, the joint distribution only depends on $\abs{\mu}$.

Computationally, we rely on the following decomposition:
\begin{prop}[Future coverage probability decomposition]
For any $z\ge 0$ with with $f_G^{Z}(z) + f_G^{Z}(-z)>0$,
\begin{align*}
    &\PP[G]{ Z \in Z' \pm 1.96       \mid \abs{Z}=z}  \\
    &\qquad =\frac{\int (\Phi(z+1.96-\mu )-\Phi(z-1.96-\mu ))\varphi(z;\mu)G(\dd\mu)}{f_G^{Z}(z) + f_G^{Z}(-z)} \\
    &\qquad + \frac{ \int (\Phi(-z+1.96-\mu )-\Phi(-z-1.96-\mu ))\varphi(-z;\mu)G(\dd\mu)}{f_G^{Z}(z) + f_G^{Z}(-z)}.
\end{align*}
\end{prop}
\begin{proof}
    Given $|Z|=z$, the event \(\{Z \in Z' \pm 1.96\}\) is a disjoint union of
\(\{Z \in Z' \pm 1.96, Z=z\}\) and \(\{Z \in Z' \pm 1.96, Z=-z\}\).
By the law of total probability,
    \begin{align*}
    &\PP[G]{ Z \in Z' \pm 1.96 \cond \abs{Z}=z}  \\
    &\qquad = \PP[G]{Z \in Z' \pm 1.96, Z=z\cond \abs{Z}=z}+ \PP[G]{Z \in Z' \pm 1.96, Z=-z\cond \abs{Z}=z} \\
    &\qquad =  \PP[G]{Z \in Z' \pm 1.96\cond Z=z, \abs{Z}=z}\PP[G]{Z =z \mid \abs{Z}=z} \\
    &\qquad + \PP[G]{Z \in Z' \pm 1.96\cond Z=-z, \abs{Z}=z}\PP[G]{Z =-z \mid \abs{Z}=z} \\
    &\qquad = \frac{\PP[G]{Z \in Z' \pm 1.96\cond Z=z}f_G^{Z}( z) + \PP[G]{Z \in Z' \pm 1.96\cond Z=-z}f_G^{Z}(-z)}{f_G^{Z}(z) + f_G^{Z}(-z)} \\
    &\qquad = \frac{\frac{1}{f_G^{Z}( z)}f_G^{Z}( z)\int\PP[G]{Z' \in z \pm 1.96\cond \mu}f(z \cond \mu)G(\dd\mu)}{f_G^{Z}(z) + f_G^{Z}(-z)} \\
    &\qquad + \frac{\frac{1}{f_G^{Z}(-z)}f_G^{Z}(-z)\int \PP[G]{Z' \in -z \pm 1.96\cond \mu}f(-z \cond \mu)G(\dd\mu)}{f_G^{Z}(z) + f_G^{Z}(-z)} \\
    &\qquad =\frac{\int (\Phi(z+1.96-\mu )-\Phi(z-1.96-\mu ))\varphi(z;\mu)G(\dd\mu)}{f_G^{Z}(z) + f_G^{Z}(-z)} \\
    &\qquad + \frac{ \int (\Phi(-z+1.96-\mu )-\Phi(-z-1.96-\mu ))\varphi(-z;\mu)G(\dd\mu)}{f_G^{Z}(z) + f_G^{Z}(-z)}.
\end{align*}

The second equality applies the chain rule, and the third equality holds since $Z = z$ implies $|Z|=z$, so $\PP[G]{Z \in Z' \pm 1.96\cond Z=z, \abs{Z}=z} = \PP[G]{Z \in Z' \pm 1.96\mid Z = z}$, and likewise for $-z$. The fourth equality uses the the fact that $Z$ is independent of $Z'$ conditional on $\mu$, and we have $ \PP[G]{z \in Z' \pm 1.96\cond \mu} = \PP[G]{Z' \in z \pm 1.96\cond \mu}$. Last equality holds as $Z \cond \mu \sim \mathrm{N}(\mu, 1)$ and $Z' \cond \mu \sim \mathrm{N}(\mu, 1)$.
\end{proof}

\subsubsection{Effect size replication probability}
Observe that $\PP[G]{   |Z'| \geq |Z| \mid \abs{Z}=z}.$ depends on the joint distribution of $(\abs{Z'}, \abs{Z})$. A similar argument yields that:
\begin{align*}
    (\abs{Z'}, \abs{Z}) &= (\abs{\mu+\epsilon'}, |\mu+\epsilon|) \\
    &\stackrel{\mathcal{D}}{=}(||\mu|-\epsilon'|, ||\mu|-\epsilon|).
\end{align*}
Hence, it only depends on the distribution of $\abs{\mu}$.

Similarly as for future coverage probability, our computation utilizes the following decomposition:
\begin{prop}[Effect size replication probability decomposition]
\begin{align*}
    &\PP[G]{|Z'| \geq |Z|\cond \abs{Z}=z}  \\
    &\qquad = \frac{\int (1-\Phi(z-\mu )+\Phi(-z-\mu ))(\varphi(z;\mu)+\varphi(-z;\mu))G(\dd\mu)}{f_G^{Z}(z) + f_G^{Z}(-z)}.
\end{align*}
\end{prop}
\begin{proof}
    \begin{align*}
    &\PP[G]{|Z'| \geq |Z|\cond \abs{Z}=z}  \\
    &\qquad = \PP[G]{|Z'| \geq z\cond \abs{Z}=z} \\
    &\qquad = \frac{\int\PP[G]{|Z'| \geq z\cond \mu}(f_G^{Z}(z) + f_G^{Z}(-z))G(\dd\mu)}{f_G^{Z}(z) + f_G^{Z}(-z)} \\
    &\qquad = \frac{\int (1-\Phi(z-\mu )+\Phi(-z-\mu ))(\varphi(z;\mu)+\varphi(-z;\mu))G(\dd\mu)}{f_G^{Z}(z) + f_G^{Z}(-z)}.
\end{align*}
The third equality uses the fact that $Z$ is independent of $Z'$ conditional on $\mu$. Last equality holds as $Z \cond \mu \sim \mathrm{N}(\mu, 1)$ and $Z' \cond \mu \sim \mathrm{N}(\mu, 1)$.
\end{proof}

\subsubsection{Publication probability}
Recall that we can break $\omega$ up as two terms, 
 $\omega = \omega_1 \cdot \omega_2$, with:
$$\omega =  \omega_1 \cdot \omega_2, \; \text{ with }\; \omega_1 =  \frac{\PP{\abs{Z} \geq 1.96 \mid D=1}}{\PP{\abs{Z} < 1.96 \mid D=1}},\;\; \omega_2 = \frac{ \PP[G]{\abs{Z} < 1.96}}{\PP[G]{\abs{Z} \geq 1.96}}.$$
where $\omega_1$ is estimated directly from the observed data. Since 
$$\omega_2 = \frac{\int_0^{1.96} f_G(z)\dd z}{\int_{1.96}^{\infty} f_G(z)\dd z},$$
it is a functional of $f_G$ and therefore depends solely on the distribution of $\abs{\mu}$. Hence, $\omega$ is identifiable.

\subsection{Proof of Proposition~\ref{prop:interpretation_symm_postmean} (a)}
\label{sec:proof_interpretation_symm_postmean}
\begin{proof}
Observe that
\begin{equation*}
\begin{aligned}
    \EE[G]{(\mu - \delta(Z))^2} - \EE[G]{(\mu - \EE[G]{\mu \mid Z})^2} &= -2\EE[G]{\mu \cdot \delta(Z)} + \EE[G]{\delta(Z)^2} + \EE[G]{\EE[G]{\mu \mid Z}^2} \\
    &=\EE[G]{(\delta(Z) - \EE[G]{\mu \mid Z})^2},
\end{aligned}
\end{equation*}
so the above optimization problem is equivalent to
$$
\underset{\delta:\RR \to\RR}{\textnormal{minimize}}\quad\EE[G]{\p{\delta(Z)-\EE[G]{\mu \mid Z}}^2}\, \text{ s.t. }\, \delta(-z)=-\delta(z)\, \text{ for all }\, z.
$$
Splitting up the integral with the constraint in mind, we have
\begin{equation*}
\begin{aligned}
  \EE[G]{(\EE[G]{\mu \mid Z} - \delta(Z))^2}
  &= \int_0^\infty (\EE[G]{\mu \mid Z=z}-\delta(z))^2 f_G^Z(z) \dd z \\
  & \qquad \qquad \text{}+ \int_{-\infty}^{0} (\EE[G]{\mu \mid Z=z}-\delta(z))^2 f_G^Z(z) \dd z  \\
  &= \int_0^\infty (\EE[G]{\mu \mid Z=z}-\delta(z))^2 f_G^Z(z) \dd z \\
  &\qquad \qquad \text{}+ \int_0^{\infty} (\EE[G]{\mu \mid Z=-z}+\delta(z))^2 f_G^Z(-z) \dd z \\
  &= \int_0^\infty \Big[ \delta(z)^2\cdot \big(f_G^Z(z)+f_G^Z(-z)\big) \\
  &\qquad \qquad \text{}+2\cdot \delta(z) \cdot \big( \EE[G]{\mu \mid Z=-z} f_G^Z(-z) -\EE[G]{\mu \mid Z=z} f_G^Z(z)\big) \\
  &\qquad \qquad \text{}+\big(\EE[G]{\mu \mid Z=z}^2f_G^Z(z) + \EE[G]{\mu \mid Z=-z}^2f_G^Z(-z) \big)\Big]\dd z.
\end{aligned}
\end{equation*}
It suffices to minimize the integrand for each fixed $z$. In this case, we have a quadratic function of $\delta(z)$, and the minimizer is
$$
\delta^*(z) = \frac{ \EE[G]{\mu \mid Z=z} f_G^Z(z) - \EE[G]{\mu \mid Z=-z} f_G^Z(-z)}{f_G^Z(z)+f_G^Z(-z)}.
$$
Expressing it as a functional of $\mathrm{Symm}[G]$, we have 
\begin{equation*}
\begin{aligned}
    \delta^*(z) &= \frac{ \EE[G]{\mu \mid Z=z} f_G^Z(z) - \EE[G]{\mu \mid Z=-z} f_G^Z(-z)}{f_G^Z(z)+f_G^Z(-z)} \\
    &= \frac{\int_\RR \mu \varphi(z;\mu) G(\dd\mu) - \int_\RR \mu \varphi(-z;\mu) G(\dd\mu)}{\int_\RR \varphi(z;\mu) G(\dd\mu) + \int_\RR  \varphi(-z;\mu) G(\dd\mu)} \\
    &= \frac{\int_\RR \mu \varphi(z;\mu) G(\dd\mu) + \int_\RR \mu \varphi(z;\mu) G^-(\dd\mu)}{\int_\RR \varphi(z;\mu) G(\dd\mu) + \int_\RR  \varphi(z;\mu) G^-(\dd\mu)} \\
    &= \frac{\int_\RR \mu\varphi(z;\mu)\mathrm{Symm}[G](\dd\mu)}{\int_\RR \varphi(z;\mu)\mathrm{Symm}[G](\dd\mu)} \\
    &= \delta_G^{\mathrm{Symm}}(z).
\end{aligned}
\end{equation*}
\end{proof}

\subsection{Proof of Proposition~\ref{prop:interpretation_symm_postmean} (b)}
\begin{proof}
    Denote $L(\delta, \mu):=\EE[Z\sim\mu]{(\delta(Z)-\mu)^2}$ and $\mathcal{A}_G:= \{\tG: \mathrm{Symm}[\tG]=\mathrm{Symm}[G]\}$. Write $\delta^*:=\delta^{\text{Symm}}_G$ and $L^*(\mu):=L(\delta^*, \mu)$. Since $\delta^*$ is an odd function, we have
    \begin{align*}
        L^*(-\mu) &= \EE[-Z\sim-\mu]{(\delta^*(-Z)+\mu)^2} = \EE[-Z\sim-\mu]{(-\delta^*(Z)+\mu)^2} \\
        &= \EE[Z\sim-\mu]{(\delta^*(Z)-\mu)^2} = \EE[Z\sim\mu]{(\delta^*(Z)-\mu)^2} = L^*(\mu),
    \end{align*}
    i.e., $L^*$ is an even function. For any $\tG \in \mathcal{A}_G$, we have 
    \begin{align*}
        \EE[\tG]{(\delta^*(Z)-\mu)^2} &= \int_\RR L^*(\mu) \tG(\dd\mu) \\
        &= \int_{\mu \geq 0} L^*(\mu) \tG(\dd\mu) + \int_{\mu < 0} L^*(\mu) \tG(\dd\mu) \\
         &= \int_{\mu \geq 0} L^*(\mu) \tG(\dd\mu) + \int_{\mu > 0} L^*(-\mu) \tG^-(\dd\mu) \\
         &= \int_{\mu \geq 0} L^*(\mu) \tG(\dd\mu) + \int_{\mu > 0} L^*(\mu) \tG^-(\dd\mu) \\
         &= \frac{1}{2}\Bigg[\int_{\mu \geq 0} L^*(\mu) \tG(\dd\mu)  + \int_{\mu \leq 0} L^*(\mu) \tG^-(\dd\mu) \Bigg] \\
         &\quad\quad  \text{ }+\frac{1}{2}\Bigg[\int_{\mu < 0} L^*(\mu) \tG(\dd\mu) + \int_{\mu > 0} L^*(\mu) \tG^-(\dd\mu) \Bigg] \\
         &= \frac{1}{2} \int L^*(\mu) \tG(\dd\mu) + \frac{1}{2} \int L^*(\mu) \tG^-(\dd\mu) \\
         &= \int L^*(\mu) \operatorname{Symm}\,[\tG]\,(\dd\mu) = \int L^*(\mu) \Symm{G}(\dd\mu) \\
         &=\EE[{\Symm{G}}]{(\delta^*(Z)-\mu)^2} .
    \end{align*}
    Hence, for any $\delta$,
    \begin{align*}
    \sup_{\tG \in \mathcal{A}_G} \Big\{\EE[\tG]{(\delta(Z)-\mu)^2}\Big\}& \geq \EE[\Symm{G}]{(\delta(Z)-\mu)^2} \\
    &\quad \text{ }\stackrel{*}{\geq}  \EE[\Symm{G}]{(\delta^*(Z)-\mu)^2}= \sup_{\tG \in \mathcal{A}_G} \Big\{\EE[\tG]{(\delta^*(Z)-\mu)^2}\Big\},
    \end{align*}
    where $(*)$ is due to the fact that $\delta^*$ is the Bayes estimator with respect to the prior $\Symm{G}$. So $\delta^*$ is $\mathcal{A}_G$-minimax (see Section~4.7.6 of \citet{berger1985}).
\end{proof}   

\subsection{Proof of Proposition~\ref{prop:convex_hull_equivalence}}
\begin{proof}
To establish the class equivalence, we prove the following inclusions: 
\begin{enumerate}
    \item $\Tilt{\mathcal{G}} \subset \mathrm{ConvexHull}(\Tilt{G_1},\dotsc,\Tilt{G_K})$
    \item $\mathrm{ConvexHull}(\Tilt{G_1},\dotsc,\Tilt{G_K}) \subset \Tilt{\mathcal{G}}$
\end{enumerate}
For the first inclusion: Take $G = \sum_{j=1}^K \pi_j G_j \in \mathcal{G}$ with $\pi_j \geq 0$ and $\sum \pi_j = 1$. Then:
$$\Tilt{G}  = \sum_{j=1}^K \pi_j \frac{\Phi(\selection; \mu) G_j}{\sum_{j=1}^K \pi_j\int \Phi(\selection; \mu) G_j(\dd\mu)} = \sum_{j=1}^K \underbrace{\left(\frac{\pi_j\int \Phi(\selection; \mu) G_j(\dd\mu)}{\sum_{j=1}^K\pi_j\int \Phi(\selection; \mu) G_j(\dd\mu)}\right)}_{\theta_j}\frac{\Phi(\selection; \mu) G_j}{\int \Phi(\selection; \mu) G_j(\dd\mu)},$$
where $(\theta_1,\dots,\theta_K)$ lie on the probability simplex. We have that $\Tilt{G} = \sum_{j=1}^K\theta_j\Tilt{G_j} \in \mathrm{ConvexHull}(\Tilt{G_1},\dotsc,\Tilt{G_K})$, and the first inclusion is established.
\medskip

\noindent For the second inclusion: Let $H = \sum_{j=1}^K \lambda_j\Tilt{G_j}$ where $(\lambda_1,\dots,\lambda_K)$ lie on the probability simplex. We have:
$$H =  \sum_{j=1}^K\lambda_j \frac{\Phi(\selection; \mu) G_j}{ \int \Phi(\selection; \mu) G_j(\dd\mu)} =\Phi(\selection; \mu)\sum_{j=1}^K \frac{\lambda_jG_j}{ \int \Phi(\selection; \mu) G_j(\dd\mu)}. $$
Consider $A = \frac{1}{\alpha}\sum_{j=1}^K\frac{\lambda_j}{\int \Phi(\selection; \mu) G_j(\dd\mu)}G_j$, where $\alpha =\sum_{j=1}^K\frac{\lambda_j}{\int \Phi(\selection; \mu) G_j(\dd\mu)}$. Since $A \in \mathcal{G}$, $\Tilt{A} \in \Tilt{\mathcal{G}}$ and:
$$
\int \Phi(\selection; \mu) A(\dd\mu) = \frac{1}{\alpha}\sum_{j=1}^K\frac{\lambda_j}{\int \Phi(\selection; \mu) G_j(\dd\mu)}\int\Phi(\selection; \mu)G_j(\dd\mu) =  \frac{1}{\alpha}.
$$
Then:
$$
\begin{aligned}
    \Tilt{A}&= \frac{\Phi(\selection; \mu) A}{ \int \Phi(\selection; \mu) A(\dd\mu)} \\
    &= \Phi(\selection; \mu)\frac{1}{\alpha}\left(\sum_{j=1}^K \frac{\lambda_jG_j}{\int \Phi(\selection; \mu)(\dd\mu)}G_j\right)\alpha \\
    &=H.
\end{aligned}
$$
\noindent Thus $H \in \Tilt{\mathcal{G}}$, both inclusions are established, proving the class equivalence.
\end{proof}

\subsection{Proof of Proposition~\ref{prop:functional_equivalence_convexhull}}
\begin{proof}
   Suppose $\Tilt{G} = \sum_{i=1}^{K}\widetilde\pi_i \Tilt{G_i} \in \Tilt{\mathcal{G}}$. Using the mapping we established in Section~\ref{subsec: mapping_convex_hull}, we have the corresponding $G = \sum_{i=1}^K \pi_i G_i$. Then we have:
    \begin{align*}
       \Tilt{T}(\Tilt{G}) &= \frac{ \int \nu(\mu)\Phi(\selection; \mu)^{-1}\Tilt{G}(\dd\mu)}{\int \delta(\mu) \Phi(\selection; \mu)^{-1}\Tilt{G}(\dd\mu)}  \\
        &= \frac{\sum_{i=1}^K\widetilde\pi_i \int \nu(\mu)\Phi(\selection; \mu)^{-1}\Tilt{G_i}(\dd\mu)}{\sum_{i=1}^K\widetilde\pi_i \int \delta(\mu)\Phi(\selection; \mu)^{-1}\Tilt{G_i}(\dd\mu)} \\
        &=\frac{\sum_{i=1}^K\frac{\widetilde\pi_i}{\int \Phi(\selection; \mu) G_i(\dd\mu) }\int \nu(\mu)G_i(\dd\mu)}{\sum_{i=1}^K\frac{\widetilde\pi_i}{\int \Phi(\selection; \mu) G_i(\dd\mu) }\int \delta(\mu)G_i(\dd\mu)} \\
        &=\frac{\sum_{i=1}^K\frac{\widetilde\pi_i}{\PP[G_i]{\abs{Z} \in \selection}}\int \nu(\mu)G_i(\dd\mu)/\sum_{j=1}^{K} \frac{\widetilde\pi_j}{\PP[G_j]{\abs{Z} \in \selection}}}{\sum_{i=1}^K\frac{\widetilde\pi_i}{\PP[G_i]{\abs{Z} \in \selection}}\int \delta(\mu)G_i(\dd\mu)/\sum_{j=1}^{K} \frac{\widetilde\pi_j}{\PP[G_j]{\abs{Z} \in \selection}}} \\
        &=\frac{\sum_{i=1}^K \pi_i \int \nu(\mu)G_i(\dd\mu)}{\sum_{i=1}^K \pi_i\int \delta(\mu)G_i(\dd\mu)}  = T(G).
    \end{align*}

The fifth equality is based on the relationship between $\pi_i$ and $\widetilde\pi_i$.

\end{proof}

\section{Supplementary analyses}
\subsection{Inference under a different truncation set}
\label{sec:alterna_trunc_set}
\beginsuppfigs
\begin{figure}
\centering
\begin{adjustbox}{max totalsize={\textwidth}{0.9\textheight}, center}
\begin{minipage}{\textwidth}
    \begin{subfigure}[t]{0.45\textwidth}
        \centering
        \caption{Marginal density}
        \includegraphics[width=\linewidth]{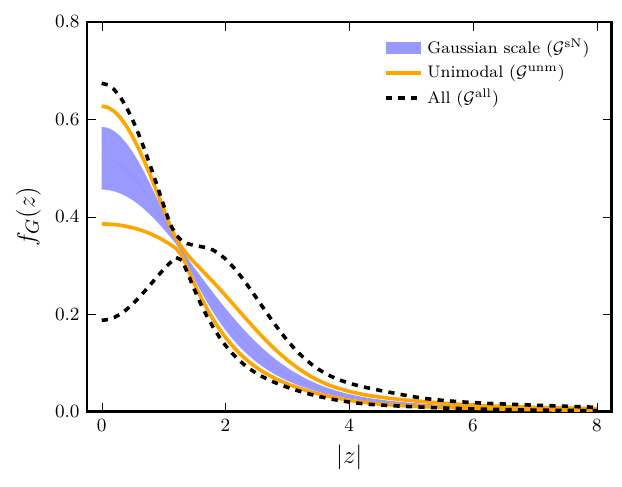}
        \label{fig:marginal_density_unnormalized_224}
    \end{subfigure}
    \hfill
    \begin{subfigure}[t]{0.45\textwidth}
        \centering
        \caption{Normalized marginal density}
        \includegraphics[width=\linewidth]{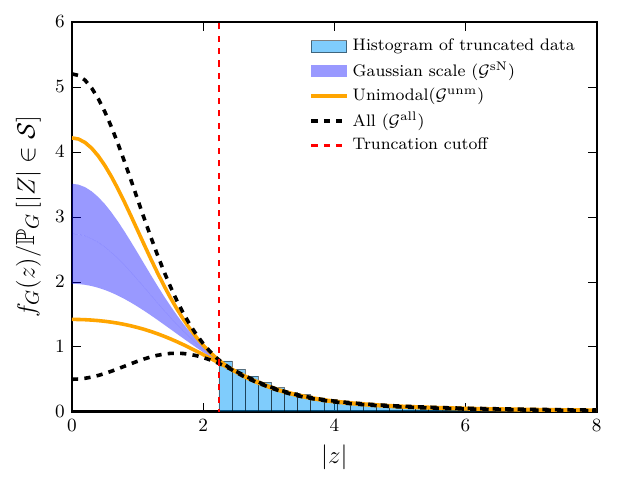}
        \label{fig:marginal_density_normalized_224}
    \end{subfigure}
    \hfill
    \begin{subfigure}[t]{0.45\textwidth}
        \centering
        \caption{Binned density of power}
        \includegraphics[width=\linewidth]{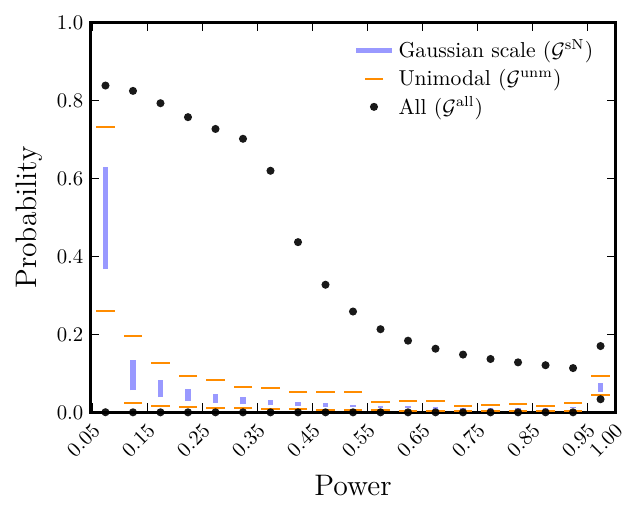}
        \label{fig:binned_power_224}
    \end{subfigure}
    \hfill
    \begin{subfigure}[t]{0.45\textwidth}
        \centering
        \caption{Probability of same sign}
        \includegraphics[width=\linewidth]{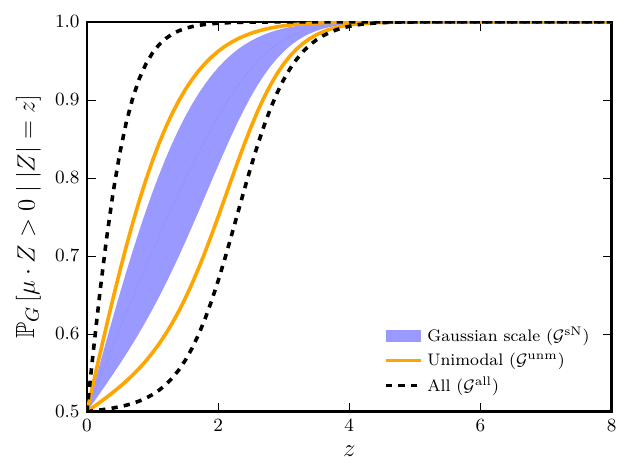}
        \label{fig:same_sign_224}
    \end{subfigure}
    \hfill
    \begin{subfigure}[t]{0.45\textwidth}
        \centering
        \caption{Symmetrized posterior mean}
        \includegraphics[width=\linewidth]{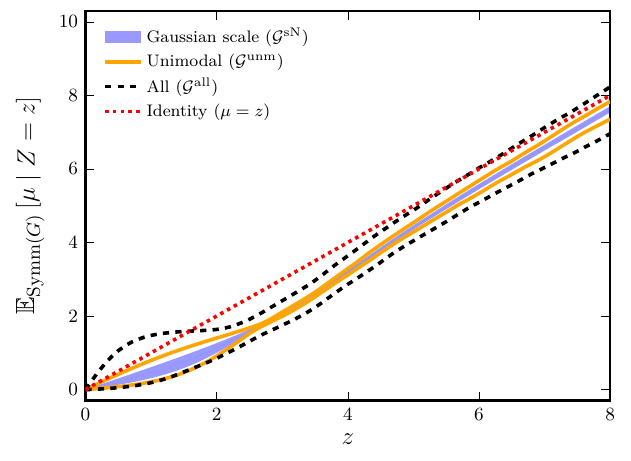}
        \label{fig:symmetrized_posterior_224}
    \end{subfigure}
    \hfill
    \begin{subfigure}[t]{0.45\textwidth}
        \centering
        \caption{Replication probability}
        \includegraphics[width=\linewidth]{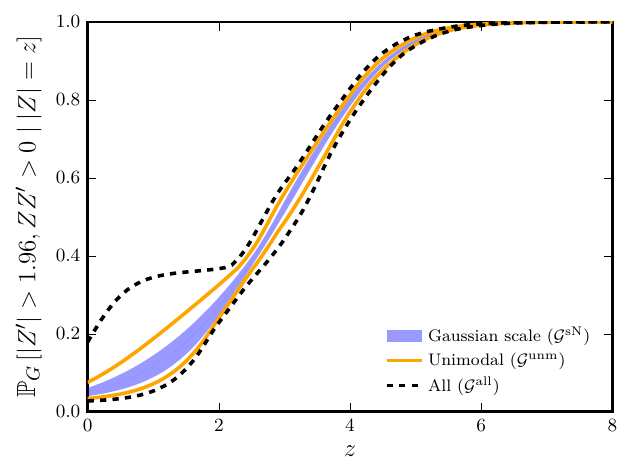}
        \label{fig:repl_prob_224}
    \end{subfigure}
    \hfill
    \begin{subfigure}[t]{0.45\textwidth}
        \centering
        \caption{Future coverage probability}
        \includegraphics[width=\linewidth]{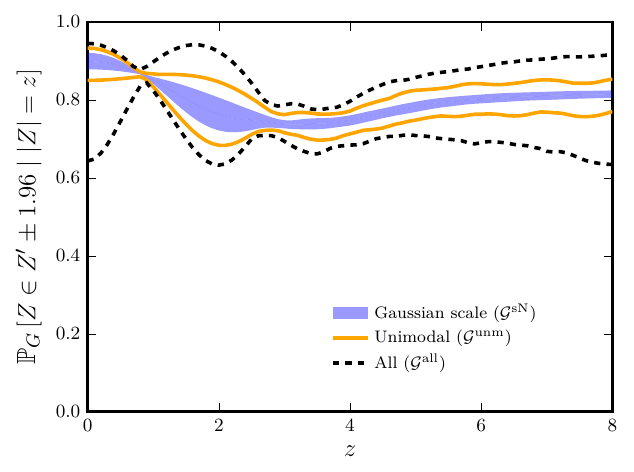}
        \label{fig:future_coverage_prob_224}
    \end{subfigure}
    \hfill
    \begin{subfigure}[t]{0.45\textwidth}
        \centering
        \caption{Effect size replication probability}
        \includegraphics[width=\linewidth]{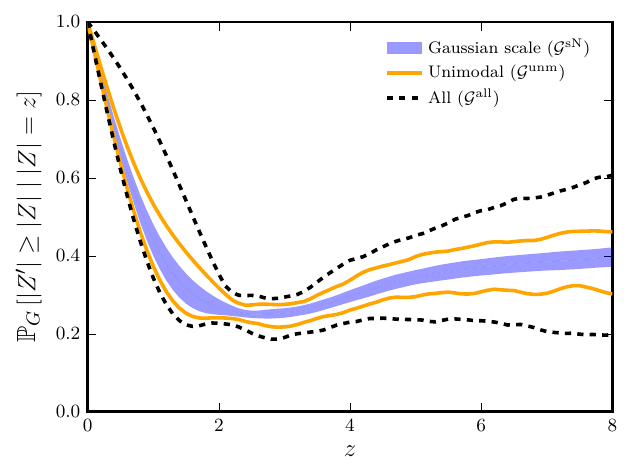}
        \label{fig:effect_repl_prob_224}
    \end{subfigure}
    
\end{minipage}
\end{adjustbox}
\caption{95\% Confidence interval analyses for MEDLINE (2000-2018) with truncation set $\selection_\text{half} = [2.24, \infty)$: Each panel presents one estimand of interest, accompanied by 95\% confidence intervals under different assumptions for the SNR distribution.}
\label{fig:medline_full_analysis_224}
\end{figure}

We replicate our inference on the MEDLINE dataset in Section~\ref{sec:estimands_results} under $\selection_\text{half} = [2.24, \infty)$, as shown in Fig.~\ref{fig:medline_full_analysis_224}. Compared to the confidence intervals constructed under $\selection = [2.1, \infty)$ in Fig.~\ref{fig:medline_full_analysis}, the intervals are very similar to each other, although the width increases slightly under $\selection_\text{half}$ due to the smaller sample size after truncation than under $\selection$.

\subsection{Inference on a future example}
\label{sec:future_example}
We return to the motivating example from the Introduction: the 2019 MEDLINE study by~\citet{decolonization2019medline}. Their
primary outcome is MRSA infection, and they report a hazard ratio of 0.70 with a
95\% confidence interval of [0.52, 0.96], associated with a p-value of 0.03. Since this result is statistically significant at the 5\% level, they conclude that decolonization is effective in reducing the risk of MRSA infection. By using only the observed absolute z-statistics from the study and ignoring all other features,  we can view the study as exchangeable with those in MEDLINE. Hence, we can use what we have learned about studies published in MEDLINE to make further inferences about this particular study.
We first transform their confidence interval for the hazard ratio into a z-score; the corresponding standard error is $\text{SE} = 0.16$, and the z-score is $z = -2.22$.

We calculate the 95\% confidence intervals using the $F$-Localization and AMARI for the sign-agreement probability, and replication probability conditional on $\abs{z}  = 2.22$, along with the symmetrized posterior mean at $z  = -2.22$ as summarized in Table~\ref{tab:example_inference}.

\begin{table}
\centering
\caption{95\% Confidence intervals for multiple estimands at $z = -2.22$ under different priors}
\label{tab:example_inference}
\begin{tabular}{lcccc}
\toprule
 & \multicolumn{2}{c}{Sign-agreement probability} & \multicolumn{2}{c}{Replication probability} \\
\cmidrule(lr){2-3} \cmidrule(lr){4-5}
\textbf{Prior} & \textbf{FLOC} & \textbf{AMARI} & \textbf{FLOC} & \textbf{AMARI} \\
\midrule
$\mathcal{G}^{\mathrm{sN}}$  & (0.894, 0.960) & (0.939, 0.961) & (0.314, 0.338) & (0.316, 0.329) \\
$\mathcal{G}^{\mathrm{unm}}$ & (0.843, 0.974) & (0.907, 0.975) & (0.307, 0.353) & (0.308, 0.334) \\
$\mathcal{G}^{\mathrm{all}}$ & (0.776, 0.999) & (0.851, 0.999) & (0.287, 0.375) & (0.297, 0.350) \\
\bottomrule
\end{tabular}

\bigskip

\begin{tabular}{lcc}
\toprule
 & \multicolumn{2}{c}{\(\EE[ \Symm{G}]{\mu \mid Z=z}\)} \\
\cmidrule(lr){2-3}
\textbf{Prior} & \textbf{FLOC} & \textbf{AMARI} \\
\midrule
$\mathcal{G}^{\mathrm{sN}}$  & (-1.43, -1.31) & (-1.44, -1.37) \\
$\mathcal{G}^{\mathrm{unm}}$ & (-1.51, -1.29) & (-1.49, -1.39) \\
$\mathcal{G}^{\mathrm{all}}$ & (-1.70, -1.14) & (-1.67, -1.10) \\
\bottomrule
\end{tabular}
\end{table}

\subsection{Single-Year analysis}
\label{subsec:2018_analysis}
\begin{figure}
\centering
\begin{adjustbox}{max totalsize={\textwidth}{0.9\textheight}, center}
\begin{minipage}{\textwidth}
    \begin{subfigure}[t]{0.45\textwidth}
        \centering
        \caption{marginal density}
        \includegraphics[width=\linewidth]{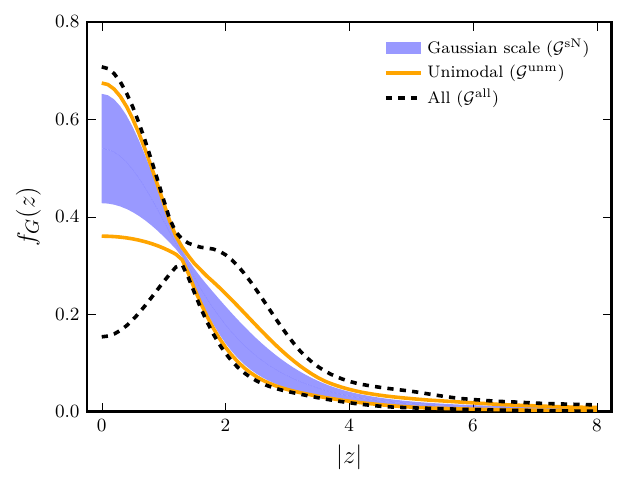}
        \label{fig:marginal_density_unnormalized_2018}
    \end{subfigure}
    \hfill
    \begin{subfigure}[t]{0.45\textwidth}
        \centering
        \caption{Normalized marginal density}
        \includegraphics[width=\linewidth]{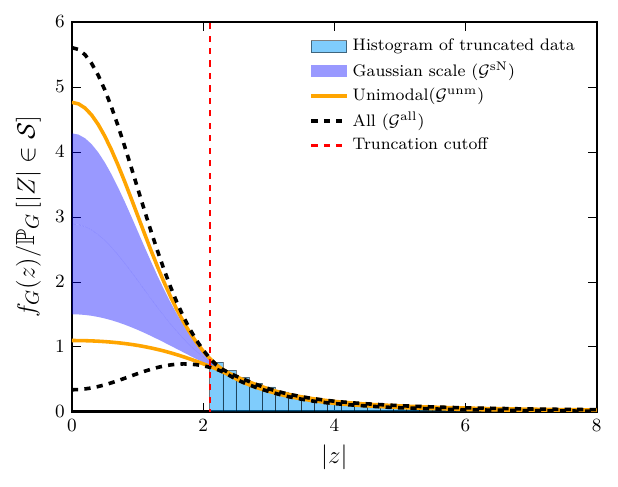}
        \label{fig:marginal_density_normalized_2018}
    \end{subfigure}
    \hfill
    \begin{subfigure}[t]{0.45\textwidth}
        \centering
        \caption{Binned density of power}
        \includegraphics[width=\linewidth]{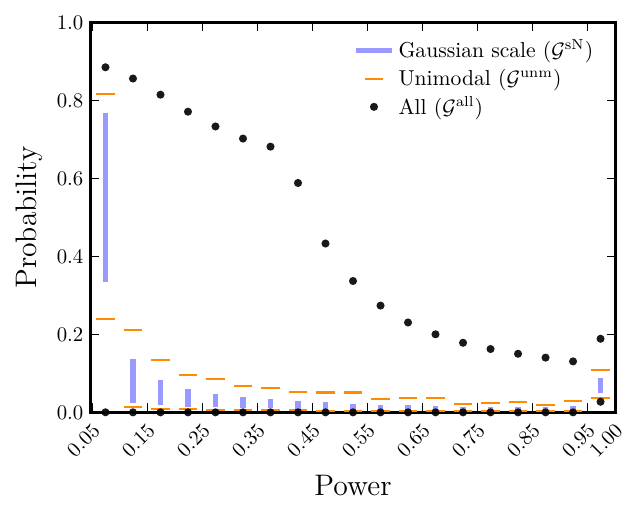}
        \label{fig:binned_power_2018}
    \end{subfigure}
    \hfill
    \begin{subfigure}[t]{0.45\textwidth}
        \centering
        \caption{Probability of same sign}
        \includegraphics[width=\linewidth]{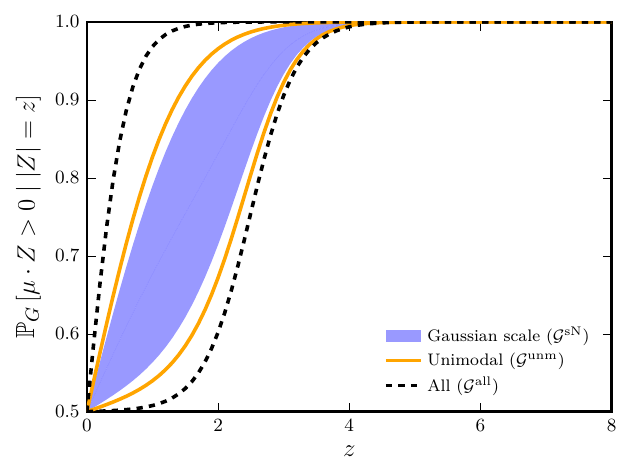}
        \label{fig:same_sign_2018}
    \end{subfigure}
    \hfill
    \begin{subfigure}[t]{0.45\textwidth}
        \centering
        \caption{Symmetrized posterior mean}
        \includegraphics[width=\linewidth]{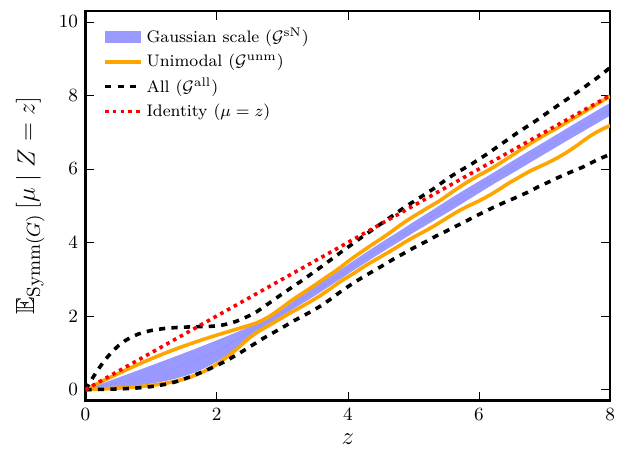}
        \label{fig:symmetrized_posterior_2018}
    \end{subfigure}
    \hfill
    \begin{subfigure}[t]{0.45\textwidth}
        \centering
        \caption{Replication probability}
        \includegraphics[width=\linewidth]{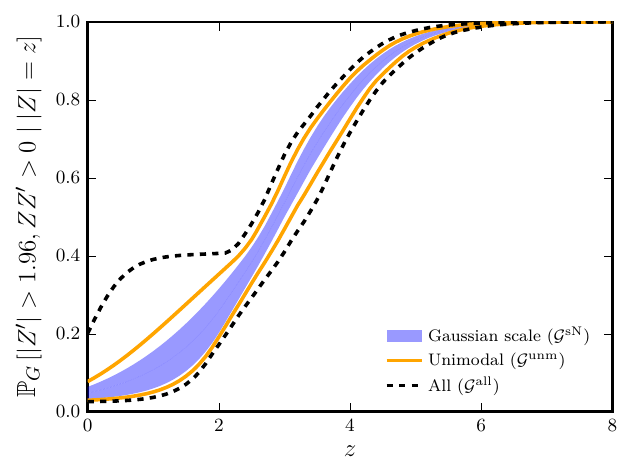}
        \label{fig:repl_prob_2018}
    \end{subfigure}
    \hfill
    \begin{subfigure}[t]{0.45\textwidth}
        \centering
        \caption{Future coverage probability}
        \includegraphics[width=\linewidth]{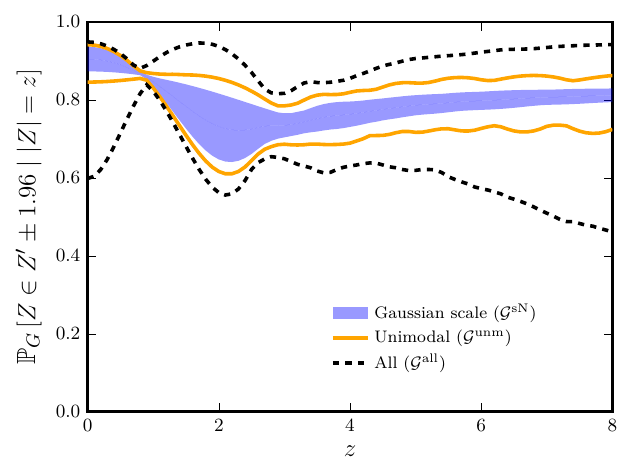}
        \label{fig:future_coverage_2018}
    \end{subfigure}
    \hfill
    \begin{subfigure}[t]{0.45\textwidth}
        \centering
        \caption{Effect size replication probability}
        \includegraphics[width=\linewidth]{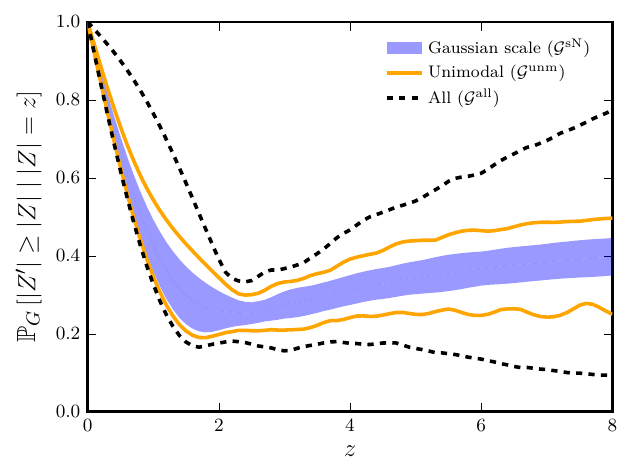}
        \label{fig:effect_size_repl_2018}
    \end{subfigure}
\end{minipage}
\end{adjustbox}
\caption{Confidence intervals analyses for MEDLINE (2018).}
\label{fig:medline_2018_CIs}
\end{figure}

\begin{table}
\centering
\caption{95\% Confidence intervals for proportion of studies with at least 80\% power under different prior classes on MEDLINE (2018).}
\label{tab:power_above_80_medline2018}
\begin{tabular}{lcc}
\toprule
\textbf{Prior}               & \textbf{FLOC} & \textbf{AMARI} \\
\midrule
$\mathcal{G}^{\mathrm{sN}}$  & (0.077, 0.127)  & (0.093, 0.130) \\
$\mathcal{G}^{\mathrm{unm}}$ & (0.073, 0.144)  & (0.084, 0.147) \\
$\mathcal{G}^{\mathrm{all}}$ & (0.034, 0.243)  & (0.033, 0.252) \\
\bottomrule
\end{tabular}
\end{table}

\begin{table}
\centering
\caption{Confidence intervals for each estimand under different priors on MEDLINE (2018).  
CIs for \(\omega_1\) and \(\omega_2\) are at the 97.5\% level;  
CI for \(\omega\) is at the 95\% level.}
\label{tab:omega_medline_2018}
\begin{tabular}{lccccc}
\toprule
 & & \multicolumn{2}{c}{\(\omega_{2}\) (97.5\%)} & \multicolumn{2}{c}{\(\omega\) (95\%)} \\
\cmidrule(lr){3-4} \cmidrule(lr){5-6}
\textbf{Prior} & \(\omega_{1}\) (97.5\%) & \textbf{FLOC} & \textbf{AMARI} & \textbf{FLOC} & \textbf{AMARI} \\
\midrule
$\mathcal{G}^{\mathrm{sN}}$  & \multirow{3}{*}{(5.64, 6.02)} & (2.15, 4.92) & (2.11, 3.66) & (12.12, 29.63) & (11.87, 22.02) \\
$\mathcal{G}^{\mathrm{unm}}$ & & (1.73, 5.39) & (1.69, 3.66) & (9.77, 32.46) & (9.54, 22.02) \\
$\mathcal{G}^{\mathrm{all}}$ & & (0.97, 6.17) & (0.93, 4.53) & (5.47, 37.14) & (5.27, 27.23) \\
\bottomrule
\end{tabular}
\end{table}
Here we demonstrate the results we obtained by focusing only on studies published in 2018 from the MEDLINE. Overall, the intervals are qualitatively similar to those from the full MEDLINE (2000-2018) data set, although substantially wider. Here we present a few interesting findings:
\begin{itemize}
    \item In Fig.~\ref{fig:marginal_density_normalized_2018}, even though our confidence intervals are much wider than the full data outside the truncation set $\selection$ due to reduced sample size. Within $\selection$, our confidence intervals track the empirical estimate of the normalized marginal density as closely as in the full data set, although with limited samples.
    \item Based on Fig.~\ref{fig:binned_power_2018} and Table~\ref{tab:power_above_80_medline2018}, most studies published in 2018 on MEDLINE exhibit low power, and we are 95\% confident that the proportion of studies with at least 80\% power is between 3.4\%-24.3\% (using $F$-Localization and $\mathcal{G}^{\mathrm{all}}$). The corresponding interval from the full data set is [4.7\%, 20.9\%], suggesting that the underlying power distribution of studies published in 2018 is similar to the population of MEDLINE studies published between 2000 to 2018. This consistency over time suggests that low power has been a persistent issue in the medical literature.
    \item We observe an even more noticeable drop in the future coverage probability around 2 than in the full dataset in Fig.~\ref{fig:future_coverage_2018}.
\end{itemize}

\begin{figure}
\setlength{\abovecaptionskip}{0pt}
\setlength{\belowcaptionskip}{0pt}
\centering
    \begin{subfigure}[t]{0.307\textwidth}
        \centering
        \caption{marginal density (with truncation)}
        \includegraphics[width=\linewidth]{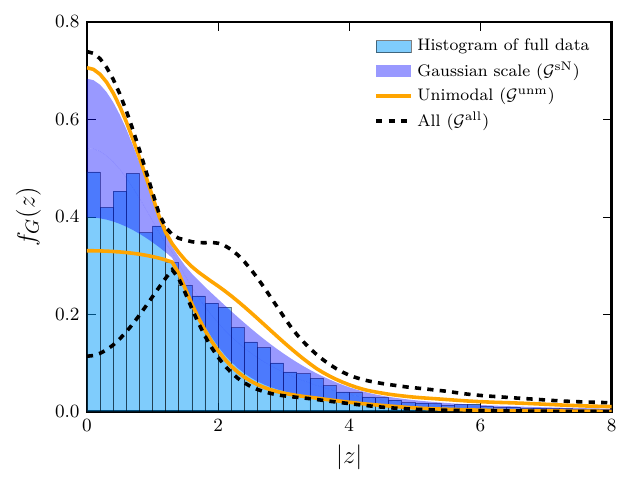}
        \label{fig:marginal_density_unnormalized_with}
    \end{subfigure}
    \hspace{0.01\linewidth}
    \begin{subfigure}[t]{0.307\textwidth}
        \centering
        \caption{marginal density (without truncation)}
        \includegraphics[width=\linewidth]{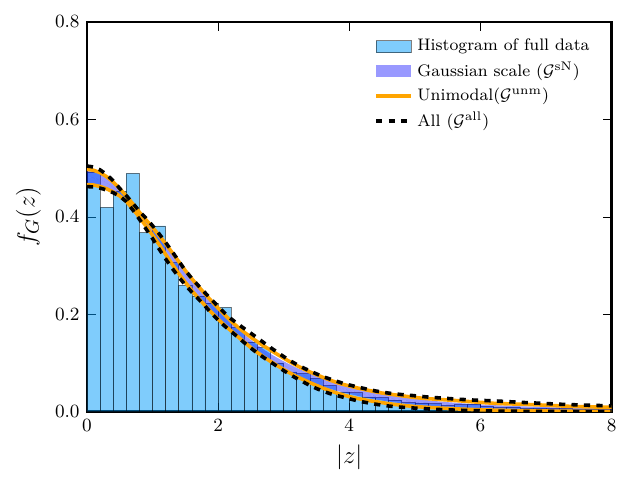}
        \label{fig:marginal_density_unnormalized_without}
    \end{subfigure}
    \hspace{0.01\linewidth}
    \begin{subfigure}[t]{0.307\textwidth}
        \centering
        \caption{marginal density (with sub-sampling)}
        \includegraphics[width=\linewidth]{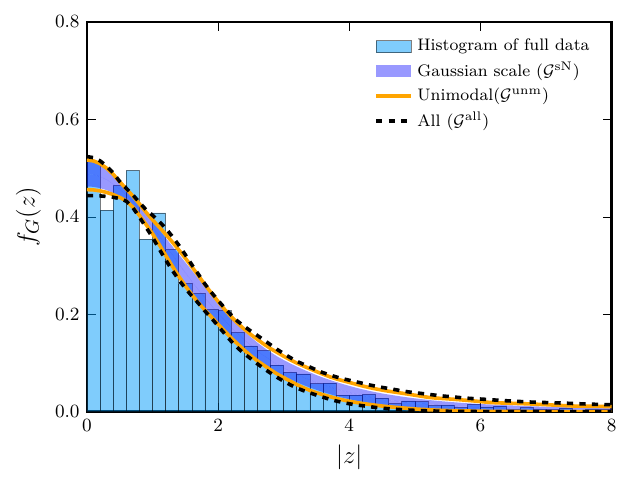}
        \label{fig:marginal_density_unnormalized_subsample}
    \end{subfigure}
    \\[0.005ex]
    \begin{subfigure}[t]{0.307\textwidth}
        \centering
        \caption{Normalized marginal density (with truncation)}
        \includegraphics[width=\linewidth]{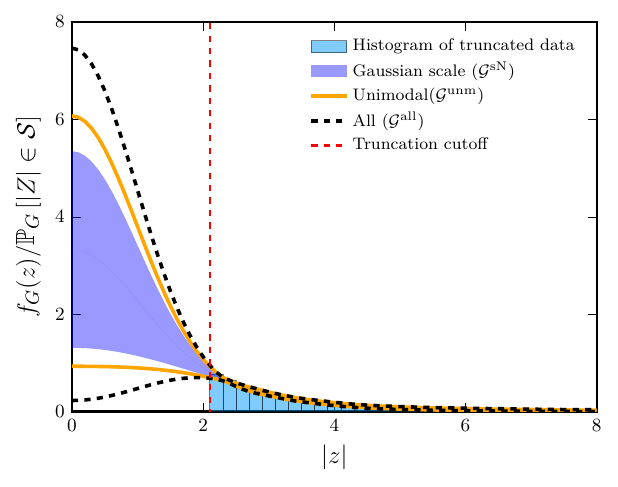}
        \label{fig:marginal_density_normalized_with}
    \end{subfigure}
    \hspace{0.01\linewidth}
    \begin{subfigure}[t]{0.307\textwidth}
        \centering
        \caption{Normalized marginal density (without truncation)}
        \includegraphics[width=\linewidth]{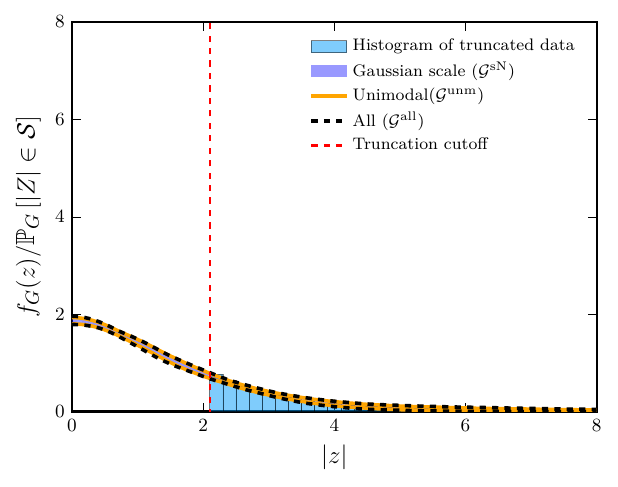}
        \label{fig:marginal_density_normalized_without}
    \end{subfigure}
    \hspace{0.01\linewidth}
    \begin{subfigure}[t]{0.307\textwidth}
        \centering
        \caption{Normalized marginal density (with sub-sampling)}
        \includegraphics[width=\linewidth]{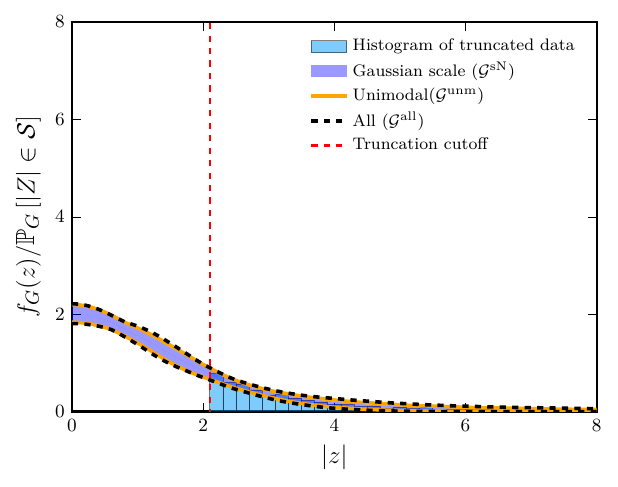}
        \label{fig:marginal_density_normalized_subsample}
    \end{subfigure}
    \\[0.005ex]
    \begin{subfigure}[t]{0.307\textwidth}
        \centering
        \caption{Binned density of power (with truncation)}
        \includegraphics[width=\linewidth]{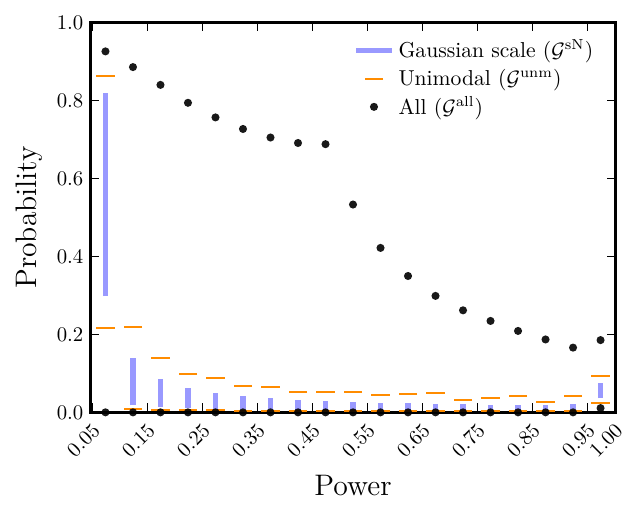}
        \label{fig:binned_power_with}
    \end{subfigure}
    \hspace{0.01\linewidth}
    \begin{subfigure}[t]{0.307\textwidth}
        \centering
        \caption{Binned density of power (without truncation)}
        \includegraphics[width=\linewidth]{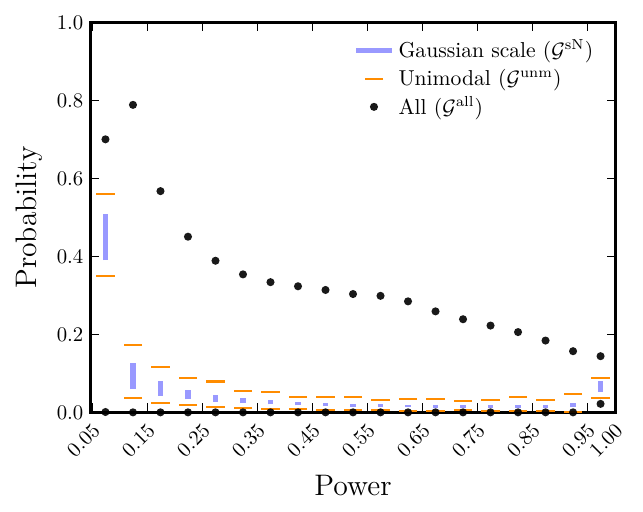}
        \label{fig:binned_power_without}
    \end{subfigure}
    \hspace{0.01\linewidth}
    \begin{subfigure}[t]{0.307\textwidth}
        \centering
        \caption{Binned density of power (with sub-sampling)}
        \includegraphics[width=\linewidth]{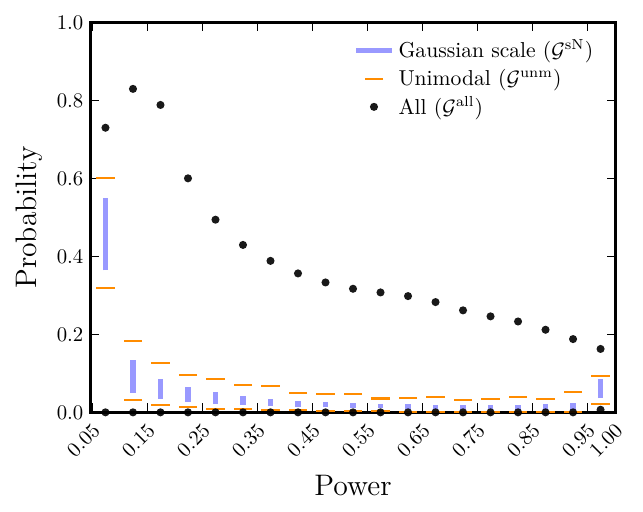}
        \label{fig:binned_power_subsample}
    \end{subfigure}
    \\[0.005ex]
    \begin{subfigure}[t]{0.307\textwidth}
        \centering
        \caption{Probability of same sign (with truncation)}
        \includegraphics[width=\linewidth]{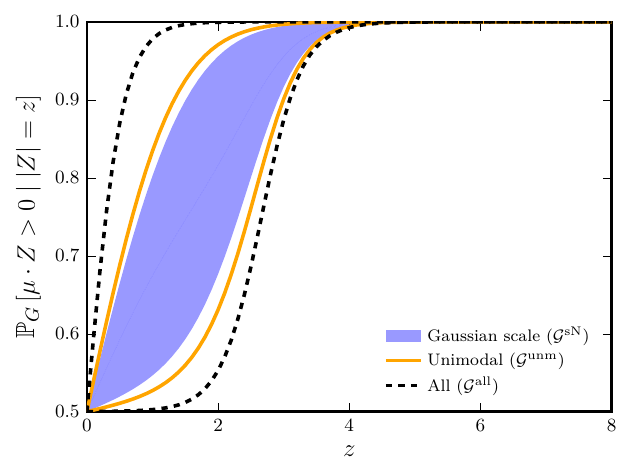}
        \label{fig:same_sign_with}
    \end{subfigure}
    \hspace{0.01\linewidth}
    \begin{subfigure}[t]{0.307\textwidth}
        \centering
        \caption{Probability of same sign (without truncation)}
        \includegraphics[width=\linewidth]{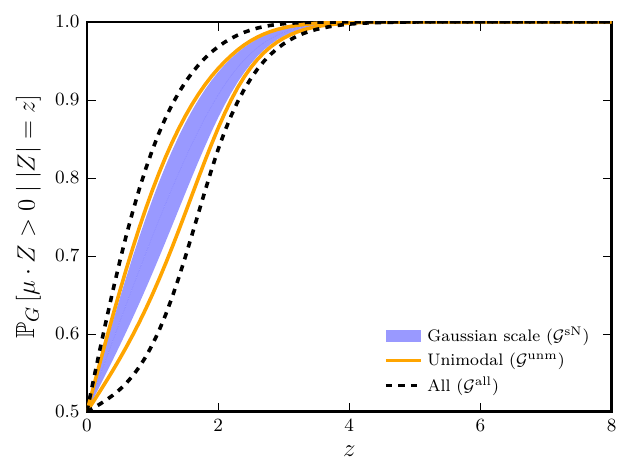}
        \label{fig:same_sign_without}
    \end{subfigure}
    \hspace{0.01\linewidth}
    \begin{subfigure}[t]{0.307\textwidth}
        \centering
        \caption{Probability of same sign (with sub-sampling)}
        \includegraphics[width=\linewidth]{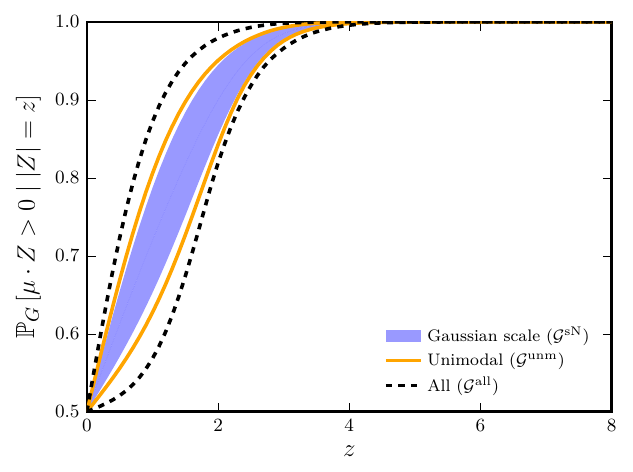}
        \label{fig:same_sign_subsample}
    \end{subfigure}
\caption{Confidence intervals analyses for Cochrane data on the first four estimands. Columns correspond to applying the truncation procedure (left), without any truncation (middle), and without truncation on a subset of the data (right).}
\label{fig:Cochrane_CIs_1}
\end{figure}

\begin{figure}
\setlength{\abovecaptionskip}{0pt}
\setlength{\belowcaptionskip}{0pt}
\centering
    \begin{subfigure}[t]{0.307\textwidth}
        \centering
        \caption{Symmetrized posterior mean (with truncation)}
        \includegraphics[width=\linewidth]{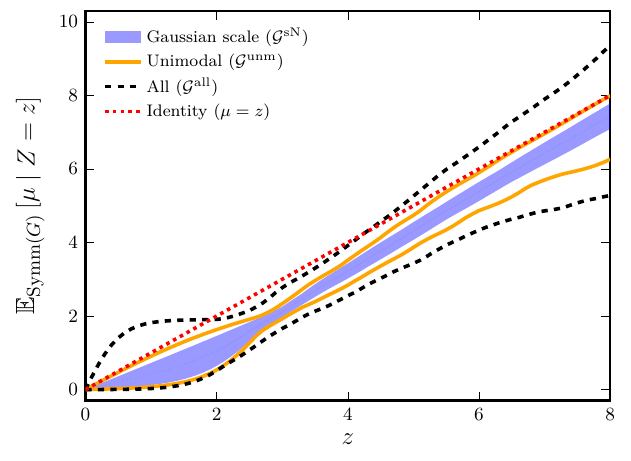}
        \label{fig:symmetrized_posterior_with}
    \end{subfigure}
    \hspace{0.01\linewidth}
    \begin{subfigure}[t]{0.307\textwidth}
        \centering
        \caption{Symmetrized posterior mean (without truncation)}
        \includegraphics[width=\linewidth]{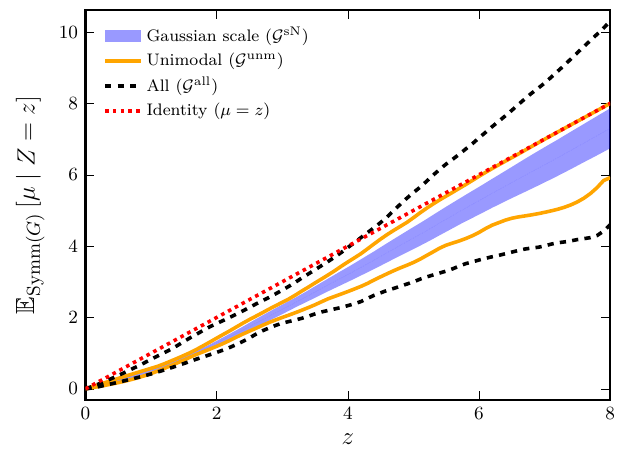}
        \label{fig:symmetrized_posterior_without}
    \end{subfigure}
    \hspace{0.01\linewidth}
    \begin{subfigure}[t]{0.307\textwidth}
        \centering
        \caption{Symmetrized posterior mean (with sub-sampling)}
        \includegraphics[width=\linewidth]{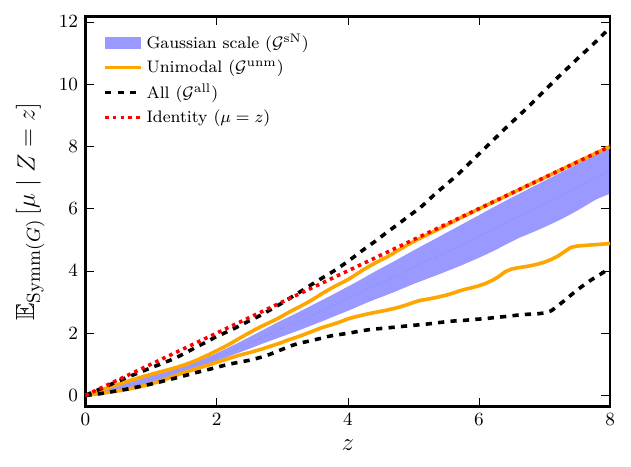}
        \label{fig:symmetrized_posterior_subsample}
    \end{subfigure}
    \\[0.005ex]
    \begin{subfigure}[t]{0.307\textwidth}
         \centering
        \caption{Replication probability (with truncation)}
        \includegraphics[width=\linewidth]{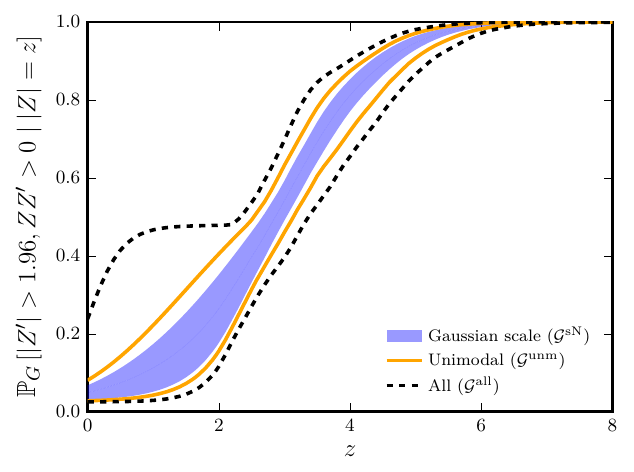}
        \label{fig:repl_prob_with}
    \end{subfigure}
    \hspace{0.01\linewidth}
    \begin{subfigure}[t]{0.307\textwidth}
          \centering
        \caption{Replication probability (without truncation)}
        \includegraphics[width=\linewidth] {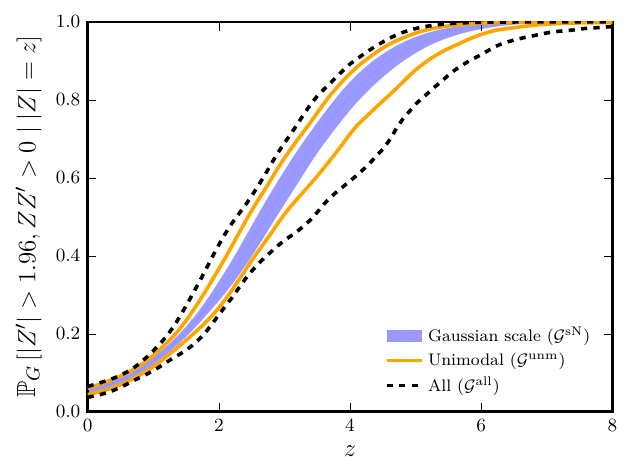}
        \label{fig:repl_prob_without}
    \end{subfigure}
    \hspace{0.01\linewidth}
    \begin{subfigure}[t]{0.307\textwidth}
          \centering
        \caption{Replication probability (with sub-sampling)}
        \includegraphics[width=\linewidth] {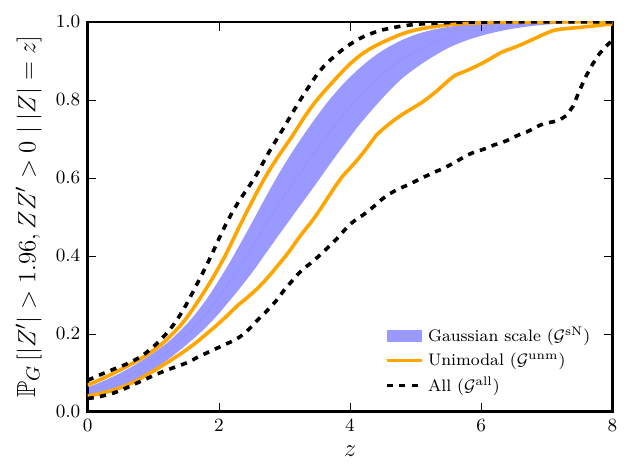}
        \label{fig:repl_prob_subsample}
    \end{subfigure}
    \\[0.005ex]
    \begin{subfigure}[t]{0.307\textwidth}
        \centering
       \caption{Future coverage probability (with truncation)}
        \includegraphics[width=\linewidth]{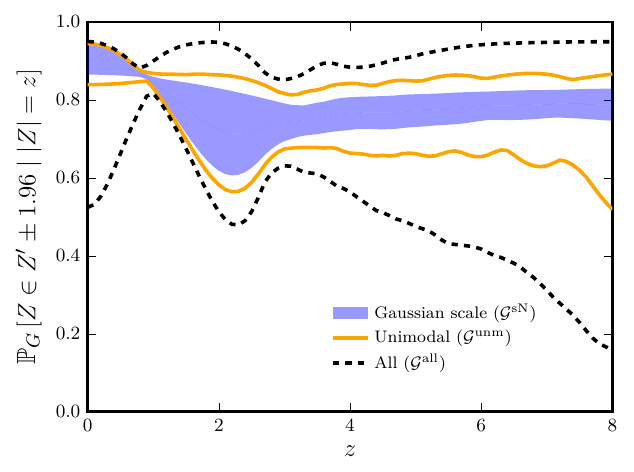}
        \label{fig:future_coverage_prob_with}
    \end{subfigure}
    \hspace{0.01\linewidth}
    \begin{subfigure}[t]{0.307\textwidth}
        \centering
         \caption{Future coverage probability (without truncation)}
        \includegraphics[width=\linewidth]{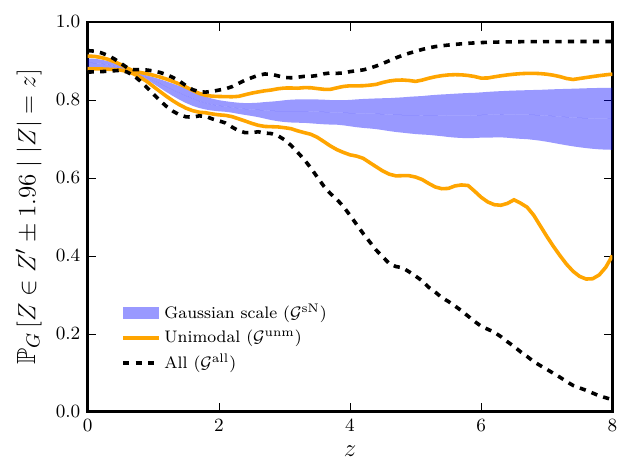}
        \label{fig:future_coverage_prob_without}
    \end{subfigure}
    \hspace{0.01\linewidth}
    \begin{subfigure}[t]{0.307\textwidth}
        \centering
        \caption{Future coverage probability (with sub-sampling)}
        \includegraphics[width=\linewidth]{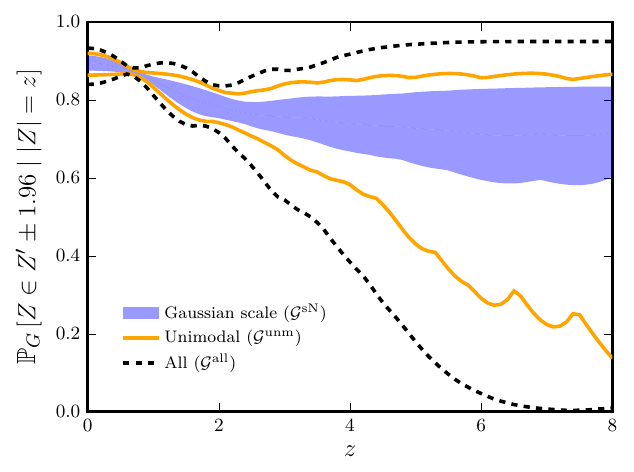}
        \label{fig:future_coverage_prob_subsample}
    \end{subfigure}
    \\[0.005ex]
    \begin{subfigure}[t]{0.307\textwidth}
        \centering
         \caption{Effect size replication probability\\ (with truncation)}
        \includegraphics[width=\linewidth]{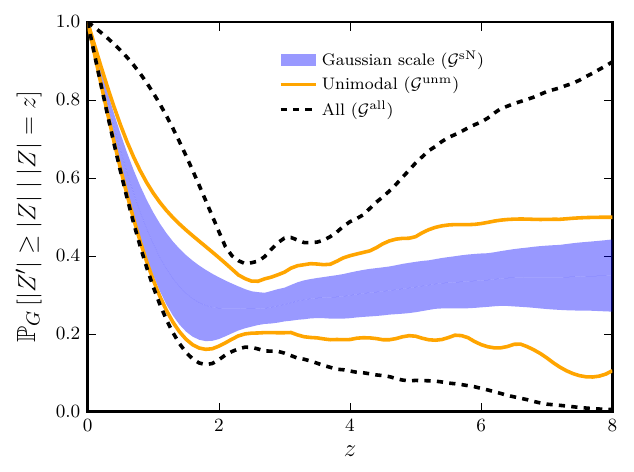}
        \label{fig:effect_repl_prob_with}
    \end{subfigure}
    \hspace{0.01\linewidth}
    \begin{subfigure}[t]{0.307\textwidth}
        \centering
        \caption{Effect size replication probability\\ (without truncation)}
        \includegraphics[width=\linewidth]{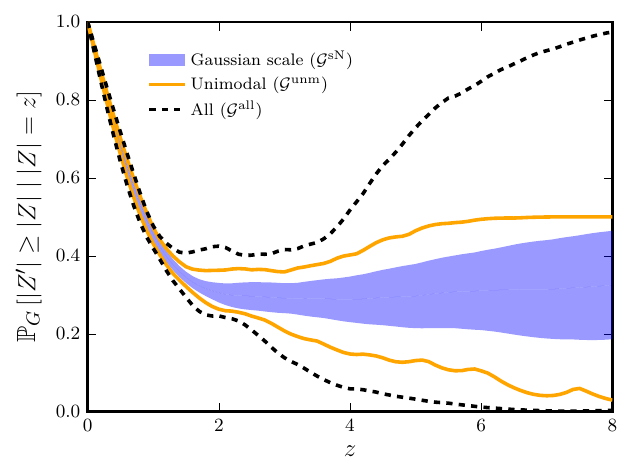}
        \label{fig:effect_repl_prob_without}
    \end{subfigure}
    \hspace{0.01\linewidth}
    \begin{subfigure}[t]{0.307\textwidth}
        \centering
         \caption{Effect size replication probability\\ (with sub-sampling)}
        \includegraphics[width=\linewidth]{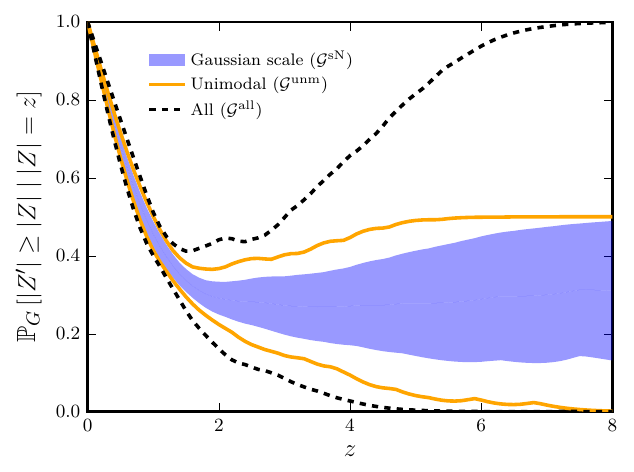}
        \label{fig:effect_repl_prob_subsample}
    \end{subfigure}
\caption{Confidence intervals analyses for Cochrane data on the other four estimands. Columns correspond to applying the truncation procedure (left), without any truncation (middle), and without truncation on a subset of the data (right).}
\label{fig:Cochrane_CIs_2}
\end{figure}

\begin{table}
\centering
\caption{95\% confidence intervals for the proportion of studies with at least 80\% power under different prior classes on Cochrane data, with and without truncation.}
\label{tab:power_above_80_Cochrane_combined}
\begin{tabular}{lcc@{\hspace{1em}}cc}
\toprule
& \multicolumn{2}{c}{\textbf{With truncation}} & \multicolumn{2}{c}{\textbf{Without truncation}} \\
\cmidrule(lr){2-3}\cmidrule(lr){4-5}
\textbf{Prior} & \textbf{FLOC} & \textbf{AMARI} & \textbf{FLOC} & \textbf{AMARI} \\
\midrule
$\mathcal{G}^{\mathrm{sN}}$  & (0.062, 0.128) & (0.067, 0.125) & (0.100, 0.131) & (0.101, 0.114) \\
$\mathcal{G}^{\mathrm{unm}}$ & (0.058, 0.145) & (0.056, 0.126) & (0.076, 0.143) & (0.088, 0.131) \\
$\mathcal{G}^{\mathrm{all}}$ & (0.023, 0.277) & (0.022, 0.253) & (0.036, 0.241) & (0.033, 0.244) \\
\bottomrule
\end{tabular}
\end{table}

\begin{table}
\centering
\caption{Confidence intervals for each estimand under different priors on Cochrane data.  
CIs for \(\omega_1\) and \(\omega_2\) are at the 97.5\% level;  
CI for \(\omega\) is at the 95\% level.}
\label{tab:omega_cochrane}
\begin{tabular}{lccccc}
\toprule
 & & \multicolumn{2}{c}{\(\omega_{2}\) (97.5\%)} & \multicolumn{2}{c}{\(\omega\) (95\%)} \\
\cmidrule(lr){3-4} \cmidrule(lr){5-6}
\textbf{Prior} & \(\omega_{1}\) (97.5\%) & \textbf{FLOC} & \textbf{AMARI} & \textbf{FLOC} & \textbf{AMARI} \\
\midrule
$\mathcal{G}^{\mathrm{sN}}$  & \multirow{3}{*}{(0.40, 0.42)} & (1.94, 6.05) & (1.90, 4.99) & (0.77, 2.55) & (0.75, 2.10) \\
$\mathcal{G}^{\mathrm{unm}}$ & & (1.53, 6.70) & (2.01, 6.09) & (0.60, 2.82) & (0.80, 2.57) \\
$\mathcal{G}^{\mathrm{all}}$ & & (0.79, 7.93) & (1.03, 6.94) & (0.31, 3.35) & (0.41, 2.93) \\
\bottomrule
\end{tabular}
\end{table}
\subsection{Cochrane robustness analysis}
\label{subsec: supp_Cochrane_analysis}
We replicate our analysis on the Cochrane database under three settings: (1) with truncation adjustment, (2) without truncation, using all the data, and (3) without truncation, using a random subset of the data of the same sample size as in (1). Across all settings, we ignore the sign information and only model $\abs{Z_i}$ as before. We note some key observations here:
\begin{itemize}
    \item By comparing the second column and the third column, we observe the effect of sample size reduction only partially accounts for the noticeably wider confidence interval with truncation. The additional factor at play here is extrapolation. For instance, comparing Fig.~\ref{fig:marginal_density_normalized_with} with Fig.~\ref{fig:marginal_density_normalized_without} and Fig.~\ref{fig:marginal_density_normalized_subsample}, we see that we are forced to extrapolate outside the truncation set $\selection$ in Fig.~\ref{fig:marginal_density_normalized_with}, whereas in Fig.~\ref{fig:marginal_density_normalized_without} and Fig.~\ref{fig:marginal_density_normalized_subsample} the same extrapolation is not needed since no truncation occurs.
    \item Although in general, confidence intervals constructed with truncation adjustment are wider than those without truncation due to both reduced sample size and extrapolation, we also note that the opposite occurs for some posterior estimands at large values of $\abs{z}$. For instance, in Fig.~\ref{fig:symmetrized_posterior_with} and Fig.~\ref{fig:symmetrized_posterior_with}, the confidence intervals constructed with truncation for the symmetrized posterior mean are shorter than those without truncation for $\abs{z}>4$. We also observe a similar phenomenon in Fig.~\ref{fig:repl_prob_with} and Fig.~\ref{fig:repl_prob_without} for large $\abs{z}$. In Section~\ref{sec:floc_posterior_detail}, we lay out an explanation for this observation.
    \item Figure~\ref{fig:marginal_density_unnormalized_with} and Figure~\ref{fig:marginal_density_unnormalized_without} showcase the marginal density of z-scores in the Cochrane data set. Notice here we also overlay the histogram of all z-scores as an empirical estimate of the marginal density, which was not plausible for the MEDLINE analysis, since we clearly observe the distortion of z-score distribution from Fig.~\ref{fig:hist_med}. Here, we believe the distribution of z-scores in Cochrane suffered relatively small publication bias as demonstrated by~\citet{zwet2021statistical} and ~\citet{Schwabe2021Cochrane}, and one can also observe from the histogram that clearly there is no heaping around the 0.05 significance threshold compared to the MEDLINE dataset. Figure~\ref{fig:marginal_density_unnormalized_with} shows that our confidence interval covers the empirical density across the entire spectrum with reasonable width, demonstrating the method's ability to account for potential selection while maintaining coverage. In particular, intervals around $z \approx 0$ are wider to accommodate selection. This is in contrast to the narrow intervals that track the histogram well we see in Fig.~\ref{fig:marginal_density_unnormalized_without}, where we ignore any selection.
    \item From Table~\ref{tab:power_above_80_Cochrane_combined}, Fig.~\ref{fig:binned_power_with} and~\ref{fig:binned_power_without}, we see that most trials in Cochrane have low power under both settings. Without truncation, we are 95\% confident that the proportion of studies with greater than 80\% power is between [10.1\%, 11.4\%]. This aligns closely with the estimate of 12\% reported by~\citet{zwet2021statistical} for the same data under similar assumption on $\mathcal{G}$ as our $\mathcal{G}^{\mathrm{sN}}$. Under the setting of truncation, the intervals are wider but the overall conclusion that most studies have low power remains unchanged.
\end{itemize}

\subsection{Simulation study: z-curve 2.0 and the bootstrap} 
\label{subsec:simulation_detail}

\begin{figure}[t]
\centering
\begin{minipage}[t]{0.32\linewidth}
      \begin{subfigure}[t]{\linewidth}
        \centering
        \includegraphics[width=\linewidth]{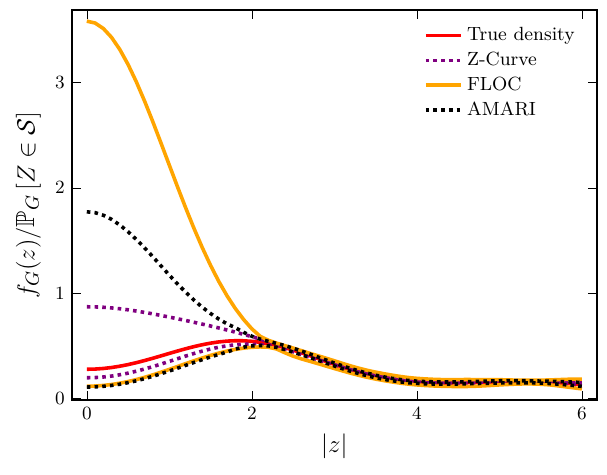}
        \caption{}
        \label{fig:average_ci}
    \end{subfigure}

    \vspace{0.8em}

    \begin{subfigure}[t]{\linewidth}
        \centering
        \includegraphics[width=\linewidth]{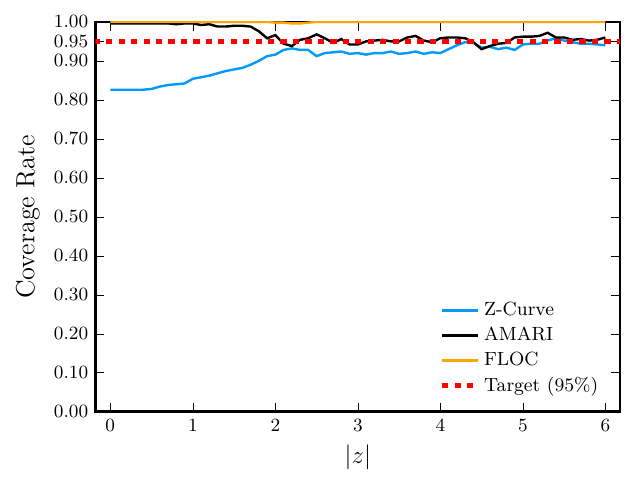}
        \caption{}
        \label{fig:pointwise_coverage}
    \end{subfigure}
\end{minipage}
\hfill
\begin{minipage}[t]{0.32\linewidth}
\centering
    \begin{subfigure}[t]{\linewidth}
        \centering
        \includegraphics[width=\linewidth]{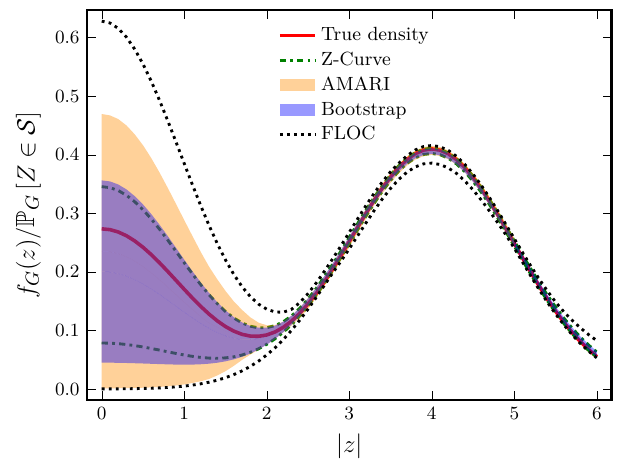}
        \caption{}
        \label{fig:average_ci_erik}
    \end{subfigure}

    \vspace{0.8em}

    \begin{subfigure}[t]{\linewidth}
        \centering
        \includegraphics[width=\linewidth]{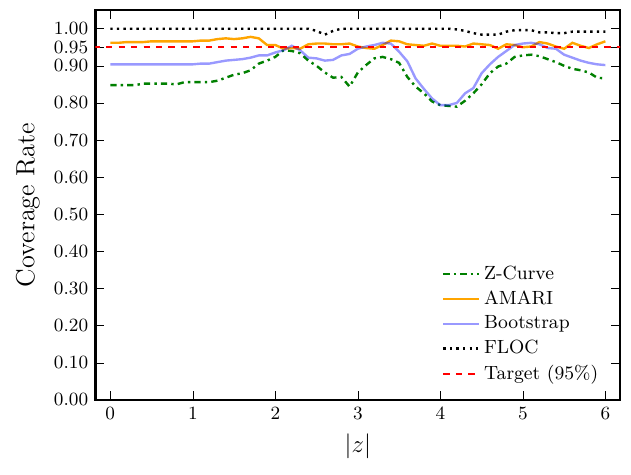}
        \caption{}
        \label{fig:pointwise_coverage_erik}
    \end{subfigure}
\end{minipage}
\hfill
\begin{minipage}[t]{0.32\linewidth}
    \centering
    \begin{subfigure}[t]{\linewidth}
        \centering
        \includegraphics[width=\linewidth]{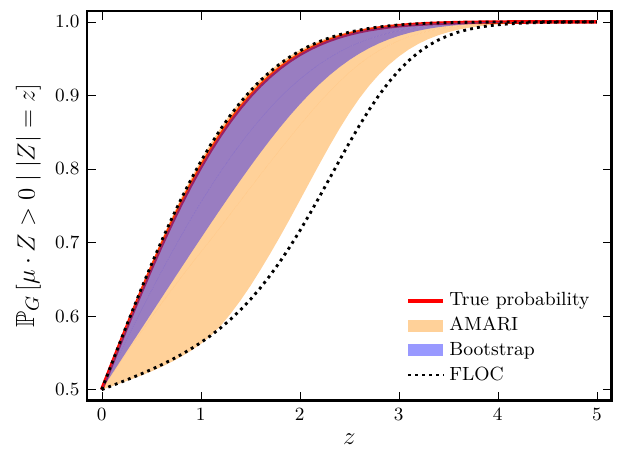}
        \caption{}
        \label{fig:average_ci_sign}
    \end{subfigure}

    \vspace{0.8em}

    \begin{subfigure}[t]{\linewidth}
        \centering
        \includegraphics[width=\linewidth]{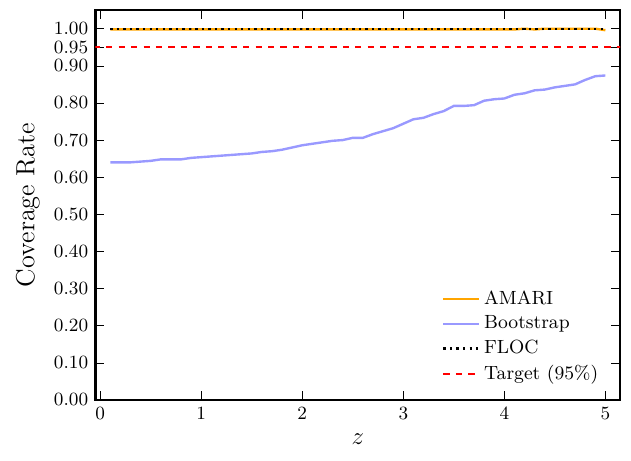}
        \caption{}
        \label{fig:pointwise_coverage_sign}
    \end{subfigure}
\end{minipage}
\caption{
Simulation results over 500 Monte Carlo replications. Columns correspond to inference for normalized marginal density under $G_1 \in \mathcal{G}^{\text{z-curve}}$ (left), inference for normalized marginal density under $G_2 \in \mathcal{G}^{\text{z-curve}}$ (middle), and inference for sign-agreement probability under $G_3 \in \mathcal{G}^{\mathrm{sN}}$ (right). (a) Average confidence intervals for z-curve, Bootstrap, AMARI, FLOC with $G_1 = 0.16\cdot \delta_1+0.59\cdot\delta_2+0.03\cdot\delta_4+0.12\cdot\delta_5+0.1\cdot\delta_6$, where $\delta_k$ denotes the point mass at $k$. (b) Point-wise coverage rate of the above confidence intervals. (c)-(d) Inference results in the simulation for z-curve, Bootstrap, AMARI, FLOC with $G_2 = 0.25\cdot\delta_0+0.75\cdot\delta_4$. (e)-(f) Inference results in the simulation for Bootstrap, AMARI, FLOC under $G_3 = 0.93\cdot\mathrm{N}(0,1.6)+0.07\mathrm{N}(0,6.5)$.
}
\label{fig:simulation_results}
\end{figure}

Here we conduct numerical studies where we know the ground truth prior $G$ of $\mu$. Our aims are threefold: (i) to connect our proposal to the z-curve method, (ii) to demonstrate that more standard inferential methods like the bootstrap can fail, and (iii) to show that our proposed methods achieve frequentist coverage.

We first evaluate the performance of each method under the z-curve setting, we specify a ground-truth prior $G_1 = 0.16\cdot \delta_1+0.59\cdot\delta_2+0.03\cdot\delta_4+0.12\cdot\delta_5+0.1\cdot\delta_6$, where $\delta_k$ denotes the point mass at $k$. Thus the class of SNR distributions used by z-curve 2.0 is well-specified in our simulation. We then set the publication probability as $\pi(z) = \ind(z \in [1.96,6])$ and generate $n_{\text{all}}=$10,000 samples according to the model in \eqref{eq:publication_bias_model}. We apply four methods: $F$-Localization, AMARI, bootstrap with interior point (IP) NPMLE, and z-curve 2.0, all using $\mathcal{G}^{\text{z-curve}}$ and $\selection = [1.96,6]$ (the defaults of z-curve 2.0) to form 95\% confidence intervals for the normalized marginal density at varying values of $z \in [0,6]$.  We report coverage averaged over 500 Monte Carlo replications of the simulation. Fig.~\ref{fig:average_ci} shows the averaged confidence intervals for each method (i.e., it shows the average of lower and upper bounds, averaged over Monte Carlo replicates). We see that z-curve has the shortest intervals, followed closely by bootstrap. We also note that AMARI produces significantly shorter confidence intervals than $F$-Localization. Fig.~\ref{fig:pointwise_coverage} shows the coverage of the four methods. The z-curve method fails to achieve the nominal coverage, particularly outside the truncation region where the method extrapolates. In contrast, the point-wise coverage rate for $F$-Localization is nearly 100\%. The coverage of AMARI is near nominal for larger values of $z$, but is also very high for smaller values of $z$. Another notable finding is that replacing the EM algorithm with IP yields improved coverage relative to the z-curve, especially outside the truncation region, where the bootstrap with IP yields coverage that is much closer to the nominal level. This improvement arises from better numerical stability and convergence of the NPMLE, whereas the EM algorithm may converge very slowly for nonparametric maximum likelihood estimation in empirical Bayes problems~\citep{koenker2014convex}.

We next repeat the preceding simulation with a new ground-truth prior $G_2 = 0.25\cdot\delta_0+0.75\cdot\delta_4 \in \mathcal{G}^{\text{z-curve}}$, which was considered by~\citet{vanzwet2026prior} to illustrate limitations of z-curve. Fig.~\ref{fig:average_ci_erik} shows the averaged confidence intervals for each method. We see that the z-curve again has the shortest intervals, followed by the bootstrap. Fig.~\ref{fig:pointwise_coverage_erik} reports coverage of all methods: the bootstrap method achieves improved or comparable coverage relative to z-curve across all $z$. Nevertheless, both methods fail to achieve nominal coverage, both outside and inside the truncation region. The point-wise coverage rate for $F$-Localization is nearly 100\%, and the coverage of AMARI is near nominal across $z$.

Finally, we evaluate the performance of our proposed methods and the bootstrap under our setting. We specify a ground-truth prior $G_3 = 0.93\cdot\mathrm{N}(0,1.6)+0.07\cdot\mathrm{N}(0,6.5) \in \mathcal{G}^{\mathrm{sN}}$. We then set the publication probability as $\pi(z) = \ind(z \in [2.1,\infty))$ and generate $n_{\text{all}}=$10,000 samples according to the model in \eqref{eq:publication_bias_model}. We consider three models: $F$-Localization, AMARI, bootstrap. For each method, we use $\mathcal{G}^{\mathrm{sN}}$ and $\selection = [2.1,\infty)$ to form 95\% confidence intervals for the sign-agreement probability $\PP[G]{\mu \cdot Z > 0 \mid \abs{Z}=z}$ over $z \in [0, 5]$.  We report coverage averaged over 500 Monte Carlo replications of the simulation in Fig.~\ref{fig:pointwise_coverage_sign}, and Fig.~\ref{fig:average_ci_sign} shows the averaged confidence intervals for each method. We see that AMARI produces shorter confidence intervals than $F$-Localization, while the bootstrap has the shortest intervals. Nevertheless, the bootstrap fails to achieve the nominal coverage across all $z$, with coverage below 70\% for $z$ outside the truncation region where the method extrapolates. In contrast, the point-wise coverage rate for $F$-Localization and AMARI are nearly 100\% for all $z$.

Overall, our simulations demonstrate that bootstrap-based methods, including z-curve, can suffer from substantial undercoverage, especially when extrapolation is needed. In contrast, our proposed methods have rigorous coverage guarantees both theoretically and empirically.
\section{More on \texorpdfstring{$F$-Localization}{F-Localization}}
\label{sec:floc_posterior_detail}
Our goal in this section is to provide a brief explanation for the finding in Section~\ref{subsec: supp_Cochrane_analysis} that, in fact for some posterior estimands, especially for large values of $|z|$, we had smaller confidence intervals using our proposed truncation-adjusted approach vs. an alternative approach that ignores truncation. (Recall that for most estimands and for most values of $|z|$ this finding was strongly reversed with confidence intervals that account for truncation being much longer.)

To explain this observation, we briefly comment on the $F$-Localization approach in a setting without publication bias, i.e., when
$$
\mu_i \simiid G,\;\; |Z_i| \mid \mu_i \simindep \mathrm{N}(\mu_i,1),\;\;i=1,\dots,n_{\text{all}}.
$$
(In other words, $D_i=1$ almost surely.) In this setting, we denote the marginal distribution of $|Z_i|$ by $F_G$. An $(1-\alpha)$-$F$-Localization is a confidence set of distributions that includes the marginal distribution of $|Z_i|$ with probability at least $1-\alpha$. Specifically, the DKW $F$-Localization includes all distributions with CDF $F(\cdot)$ that satisfies
\begin{equation}
\abs{F(t) - \hat{F}(t)} \leq \sqrt{\frac{\log(2/\alpha)}{2n_{\text{all}}}} \; \text{ for all }\; t \geq 0, \; \text{ where }\; \hat{F}(t):=\frac{1}{n_{\text{all}}} \sum_{i=1}^{n_{\text{all}}}  \ind(|Z_i| \leq t).
\label{eq:floc_general}
\end{equation}
However, many other choices are possible, and depending on the estimands of interest, other choices may yield tighter intervals with no choice choice dominating all others uniformly; see~\citet{ignatiadis2022confidence} for further discussion. (An important practical advantage of the DKW $F$-Localization is that it holds for all choices of prior class and likelihood.)

In~\eqref{eq:KS-ball} we constructed the DKW $F$-Localization for the truncated marginal distribution $F_{\Tilt{G}}^\mathrm{B}$, which is the same as $F_G^A$ by Theorem~\ref{theo:observational_equivalence}. Here we argue that this $F$-Localization can also be interpreted as an alternative $F$-Localization for the original (untruncated) $F$.

For simplicity, suppose $\selection=[c,\infty)$ (as in our main specification).
Notice that we can write~\eqref{eq:KS-ball} in terms of the (untruncated) $F$ as follows:
$$
\abs{\frac{F(t)}{1-F(c)} -  \frac{\hat{F}(t)}{1-\hat{F}(c)}} \leq \sqrt{\frac{\log(2/\alpha)}{2 (1-\hat{F}(c)) n_{\text{all}}}} \; \text{ for all }\; t \geq c.
$$
Above for simplicity, we assumed that there are no observations exactly equal to $c$. To gain further intuition, we (incorrectly) pretend that $\hat{F}(c) = F(c)$, so that we can further rearrange as follows,
$$
\abs{F(t) - \hat{F}(t)} \leq \sqrt{1-\hat{F}(c)}\sqrt{\frac{\log(2/\alpha)}{2n_{\text{all}}}} \; \text{ for all }\; t \geq c.
$$
Thus, to some order of approximation, the $F$-Localization on the truncated data can be interpreted as making the following tradeoffs compared to the $F$-Localization in~\eqref{eq:floc_general}: we do not have any bounds for $F(t)$ for all $t \in [0,c)$, but for $t \geq c$, we get slightly tighter bands around the empirical distribution. This explains why for most estimands that depend on the behavior of the marginal distribution near $0$, the confidence intervals that account for truncation are much wider. It also explains why for some posterior estimands with large $|z|$, the confidence intervals can be shorter. 

We note that all of the above comparisons only apply when there is no publication bias. In the presence of publication bias as in our main analysis, we cannot use the $F$-Localization in~\eqref{eq:floc_general}.

\section{Zero-truncated Poisson sampling and Corbet's butterflies}
\label{sec:butterflies}
In this section, we demonstrate the generality of our selective tilting framework beyond the folded normal distribution in the context of publication bias by applying it to the Poisson empirical Bayes problem under zero-truncation. Specifically, we explain the analogous observational equivalence results (which have been previously derived by~\citet{bohning2006equivalence} and~\citet{efron2019bayes}) and then demonstrate how we can form confidence intervals for empirical Bayes estimands using our techniques.

In particular, we revisit the classical Corbet's butterfly data from empirical Bayes analysis~\citep{fisher1943relation}. The butterfly data records species of butterfly that Alexander Corbet had trapped after two years in Malaysia: 118 rare species had been captured only once, 74 had been captured twice, etc. Table 3 in~\citet{efron2019bayes} records the butterfly data. Notice that we do not know how many butterfly species were never captured (this leads to the famous missing species problem). Formally, let $Z_i$ denote the number of butterflies of species $i$ Corbet captured in two years, then $Z_i$ we observe can only take values in $\cb{1,2,\dotsc}$, i.e., $0$ is truncated.

Adopting our End truncation model and Per-unit truncation model to this zero-truncated Poisson scenario:
\begin{equation}
\mbox{\textbf{End truncation:~~~~~}}
\mu_i \sim \gprior, \ \ Z_i \sim \Poisson{\mu_i},  \ \ \text{observe } Z_i \text{ only if } Z_i >0.\mbox{~~~~~~}
\label{eq:end_truncation_poisson} \tag{A$_{\text{Pois}}$}
\end{equation}

\begin{equation}
\mbox{\textbf{Per-unit truncation:~~~~~}}\mu_i \sim \gprior, \ \ Z_i \sim \TruncPoisson{\mu_i},\mbox{~~~~~~~~~~~~~~~~~~~~~~~~~~~}    
\label{eq:per_unit_truncation_poisson}\tag{B$_{\text{Pois}}$}
\end{equation}
where $\TruncPoisson{\mu_i}$ is the Poisson distribution $\Poisson{\mu_i}$ truncated to $\cb{1,2,\dotsc}$ with probability mass function:
$$p(z \mid \mu) = \frac{\exp(-\mu) \mu^z}{z! (1-\exp(-\mu))}, \; z \in \cb{1,2,\dotsc}.$$
The corresponding marginal density for models~\eqref{eq:end_truncation_poisson} and~\eqref{eq:per_unit_truncation_poisson}:
$$ f^{A_{\text{Pois}}}_G(z) = \int  \frac{\exp(-\mu) \mu^z}{z! (1-\exp(-\mu))}\, G(\dd\mu),\;\;  f^{B_{\text{Pois}}}_G(z) = \frac{\int  \exp(-\mu) \mu^z/z!\, G(\dd\mu)}{\int \p{1- \exp(-\mu)}\, G(\dd\mu)}, \,\;\, z\in \cb{1,2,\dotsc}.$$

We establish an observational equivalence between models~\eqref{eq:end_truncation_poisson} and~\eqref{eq:per_unit_truncation_poisson} similar to our Theorem~\ref{theo:observational_equivalence} in Section~\ref{subsec:simple_tilting}, by defining a tilting operation for priors adapted for the zero-truncated Poisson:

\noindent \textbf{Tilting of priors.} The tilting operation for priors under the zero-truncated Poisson,
\begin{equation}
  \label{eq:tilt_G_poisson}
\Tilt{G}(\dd\mu) := \frac{(1-\exp(-\mu)) G(\dd\mu)}{ \int (1-\exp(-\mu)) G(\dd\mu)}, 
\end{equation}
where our truncation set is $\selection = \cb{z \in \mathbb{N}^+: z \in \cb{1,2,\dotsc}}$. There is one subtlety to the $\Tilt{\cdot}$ mapping in the truncated Poisson setting: it is not injective if we allow for distributions that place mass on $0$.\footnote{ Let $G$ be a distribution on $[0,\infty)$ with $G(\{0\}) \in (0,1)$. Then there exists another distribution $H \neq G$ such that $\Tilt{G}=\Tilt{H}$.} For this reason, we restrict our attention to distributions with \smash{$G(\cb{0})=0$} and we write \smash{$\mathcal{G}_{>0}$} for the class of all such distributions. Given any \smash{$\tG \in \Tilt{\mathcal{G}_{>0}}$}, we can invert the tilting operation via the following untilting:
\begin{equation}
  \label{eq:untilt_G_poisson}
\Untilt{\tG}(\dd\mu) := \frac{(1-\exp(-\mu))^{-1} \tG(\dd\mu)}{ \int (1-\exp(-\mu))^{-1} \tG(\dd\mu)}, \;\; \tG \in \Tilt{\mathcal{G}} . 
\end{equation}

For a fixed prior $G \in \mathcal{G}$ where $\mathcal{G}$ is a convex class of prior, we have that the marginal distribution of $Z$ under model~\eqref{eq:end_truncation_poisson} with prior $G$ is equal to the marginal distribution of $Z$ under model~\eqref{eq:per_unit_truncation_poisson} with prior $\Tilt{G}$. Likewise, all the theoretical results established in Section~\ref{subsec:simple_tilting} extend to the zero-truncated Poisson setting. Hence, we can follow the inferential approaches described in Section~\ref{subsec:description_inference} to conduct inference on estimands.

Returning to the Corbet butterfly data, we use the above insight with $\mathcal{G} = \pp([0.01, 25])$, where:
$$
\pp(\mathcal{K}) = \cb{\text{$G$ distribution: support($G$) $\in \mathcal{K}$}} \text{for $\mathcal{K} \subset \mathbb{R}$.}
$$
A particular estimate of the posterior mean is obtained by applying Zipf's law, which assumes that $f(z) \propto 1/z$, where $f(z)$ is the marginal density of $Z$. Combining it with Robbins' formula, we obtain that $\mathbb{E}[\mu \mid Z=z] = z$. 
In Figure~\ref{respfig:butterfly}, we provide 95\% confidence intervals for the posterior mean using $F$-localization and AMARI. We also plot the Zipf's law estimate, which is contained in the $F$-localization intervals. Figure~5 in \citet{efron2019bayes} demonstrates other empirical Bayes estimates of the posterior mean on the same dataset; our intervals also contain those estimates well.

\begin{figure}
\centering
\begin{adjustbox}{width=0.5\linewidth} 
\includegraphics[width=\linewidth]{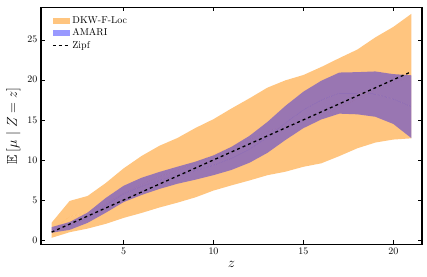}
\end{adjustbox} 
\caption{Application of $F$-Localization and AMARI for the posterior mean in Corbet's butterfly data. The black dashed line shows Zipf's estimate \smash{$\widehat{\mathbb{E}}[\mu \mid Z=z] = z$}.}
\label{respfig:butterfly}
\end{figure}

\end{document}